\documentclass{article}[12pt]
\usepackage[colorlinks=true,linkcolor=blue,citecolor=blue,urlcolor=blue,bookmarks=false,hypertexnames=true]{hyperref}
\usepackage{setspace}
\usepackage[utf8]{inputenc}
\usepackage{physics}
\usepackage{amsmath}
\usepackage[colorinlistoftodos]{todonotes}
\usepackage{url}
\usepackage{subcaption}
\usepackage{enumerate}
\usepackage[authoryear]{natbib}
\usepackage{tabularx}
\usepackage{cancel}
\usepackage[normalem]{ulem}
\usepackage{array}
\usepackage{framed}
\usepackage[]{mdframed}
\newcolumntype{P}[1]{>{\centering\arraybackslash}p{#1}}
\usepackage{media9}
\usepackage{enumitem}
\newlist{notes}{enumerate}{1}
\setlist[notes]{label=Note: ,leftmargin=*}
\usepackage{diagbox}
\usepackage{wrapfig,lipsum,booktabs}
\usepackage{tikz}
\usetikzlibrary{shapes.geometric,arrows}
\tikzstyle{startstop} = [rectangle, rounded corners, minimum width=3cm, minimum height=0.75cm,text centered, draw=black, fill=red!30]
\tikzstyle{io} = [rectangle, minimum width=3cm, minimum height=0.75cm, text centered, draw=black, fill=blue!30]
\tikzstyle{split} = [rectangle, minimum width=3cm, minimum height=1cm, text centered, text width = 4cm, draw=black, fill=orange!30]
\tikzstyle{merge} = [rectangle, minimum width=3cm, minimum height=1cm, text centered, text width = 4cm, draw=black, fill=orange!30]
\tikzstyle{test} = [rectangle, minimum width=3cm, minimum height=1cm, text centered, text width = 4cm, draw=black, fill=orange!30]
\tikzstyle{increment} = [rectangle, minimum width=3cm, minimum height=0.75cm, text centered, draw=black, fill=orange!30]

\tikzstyle{process} = [rectangle, minimum width=3cm, minimum height=1cm, text centered, text width = 3cm,draw=black, fill=orange!30]
\tikzstyle{decision} = [rectangle, minimum width=3cm, minimum height=0.75cm, text centered, draw=black, fill=green!30]
\tikzstyle{arrow} = [thick,->,>=stealth]

\usepackage{csquotes}
\usepackage{amssymb}
\usepackage{amsthm}
\usepackage{nccmath}
\usepackage{mathtools}
\usepackage{mdframed}
\usepackage{tgbonum}
\usepackage{lipsum}
\usepackage{lmodern}
\usepackage{pdfpages}
\usepackage{tcolorbox}
\usepackage{graphicx}
\usepackage{xfrac}
\usepackage{xcolor}
\usepackage{appendix}
\usepackage{lingmacros}
\usepackage{tree-dvips}
\usepackage{multicol}
\usepackage{multirow}
\usepackage[linesnumbered,ruled,vlined]{algorithm2e}
\usepackage{algpseudocode}
\SetKwInput{KwInput}{Input}                
\SetKwInput{KwOutput}{Output}              

\usepackage[left = 2 cm, right = 2 cm]{geometry}
\let\Abstract\abstract
\long\def\abstract{\mdframed[backgroundcolor=white!20,hidealllines=true]
  \vspace*{-0.1\baselineskip}\Abstract}
\let\endAbstract\endabstract
\def\endabstract{\endAbstract\endmdframed\par\bigskip}
\title{ \textbf{ \Large Near-perfect Clustering Based on Recursive Binary Splitting Using Max-MMD} }

\date{}
\newtheorem{theorem}{Theorem}
\newtheorem{definition}{Definition}[section]
\newtheorem{lemma}{Lemma}[]
\newtheorem{remark}{Remark}

\newcolumntype{Y}{>{\centering\arraybackslash}X}

\begin{document}
\maketitle



\begin{center}
	\null\vskip-2.5cm
	\author{  \textbf{\large Sourav Chakrabarty$^{1}$, Anirvan Chakraborty$^{2}$, 
			 Shyamal\ K.\ De}$^{1}$
		\\
		$^{1}$Applied Statistics Unit \\
		Indian Statistical Institute, Kolkata, India  \\
		$^{2}$Department of Mathematics and Statistics	\\
		Indian Institute of Science Education and Research Kolkata, India
		Emails:$^{1}$chakrabartysourav024@gmail.com, $^{2}$anirvan.c@iiserkol.ac.in, $^{1}$shyamalkd@isical.ac.in
	}
\end{center}
\footnotetext[2]{The research of Anirvan Chakraborty is supported by the MATRICS grant MTR/2021/000211 provided by the Science and Engineering Research Board, Government of India}

\begin{abstract}
We develop novel clustering algorithms for functional data when the number of clusters $K$ is unknown and also when it is prefixed. These algorithms are developed based on the Maximum Mean Discrepancy (MMD) measure between two sets of observations. The algorithms recursively use a binary splitting strategy to partition the dataset into two subgroups such that they are maximally separated in terms of an appropriate weighted MMD measure. When $K$ is unknown, the proposed clustering algorithm has an additional step to check whether a group of observations obtained by the binary splitting technique consists of observations from a single population. We also obtain a bonafide estimator of $K$ using this algorithm. When $K$ is prefixed, a modification of the previous algorithm is proposed which consists of an additional step of merging subgroups which are similar in terms of the weighted MMD distance. The theoretical properties of the proposed algorithms are investigated in an oracle scenario that requires the knowledge of the empirical distributions of the observations from different populations involved. In this setting, we prove that the algorithm proposed when $K$ is unknown achieves perfect clustering while the algorithm proposed when $K$ is prefixed has the perfect order preserving (POP) property. Extensive real and simulated data analyses using a variety of models having location difference as well as scale difference show near-perfect clustering performance of both the algorithms which improve upon the state-of-the-art clustering methods for functional data.
\end{abstract}


\noindent
\section{Introduction}
\indent Clustering of a given dataset aims at partitioning the data into subgroups such that each subgroup is as homogeneous as possible and any two distinct subgroups are as heterogeneous as possible. For functional data, the problem of clustering has received quite a lot of attention in the past primarily due to the challenges involved, namely, the functional nature of the data which entails that many of the multivariate techniques may not be directly applicable, and the fact that one could also incorporate the functional nature of the data while building a clustering procedure (e.g., also consider the derivative curves in addition to the original data curves). One of the classical approaches for clustering of functional data is to define an appropriate notion of distance between two data curves e.g., the usual $L_2$ distance and their modifications (see \cite{ferratynonparametric}, \cite{CUESTAALBERTOS20074864}, \cite{chen2014optimally}, \cite{meng2018new}, \cite{martino2019k}, \cite{chen2021clustering}) followed by a k-means/k-medoids algorithm. A second approach which is typical to functional data is to use projections of the data onto a suitable finite dimensional subspace and then use an appropriate multivariate clustering technique. Functional data may be projected in a finite dimension mainly in two ways: (i) for fully observed curves, one can project the the curves onto an orthonormal basis in the function space (see \cite{chiou2007functional}, \cite{yamamoto2012clustering}, \cite{LUZLOPEZGARCIA2015231}, \cite{delaigle2019clustering}), and (ii) represent the discretely observed data through an appropriate truncated basis expansion using spline basis, principal component basis, wavelet basis, among others (see \cite{abraham2003unsupervised}, \cite{doi:10.1198/000313007X171016}, \cite{kayano2010functional}, \cite{10.1111/j.1541-0420.2012.01828.x} and \cite{tzeng2018dissimilarity}). While some methods for functional clustering use model-based techniques (see \cite{funfem}, \cite{chamroukhi2016piecewise}, \cite{suarez2016bayesian}, \cite{margaritella2021parameter}), a few papers use pseudo-density based approaches (see \cite{JACQUES2013164}, \cite{10.1111/j.1541-0420.2012.01828.x}).
\\
\indent Majority of the literature on functional clustering deals with the situation where there is difference in location. However, the case where the difference is in the scale, or more generally, in the distributional structure has only been sparsely addressed. Furthermore, the papers which consider scale difference mostly require additional difference in the location for the method to work (see, for example, \cite{delaigle2019clustering}). On the other hand, model based or pseudo density based approaches (which potentially can deal with distributional differences) are too restrictive in their approaches. Model based approaches usually make Gaussian or some other parametric mixture model assumptions while pseudo-density based procedures use approximations of the likelihood due to the nonexistence of a conventional notion of density for functional data. However, the existing work in these directions are still restricted to detecting difference only in the location. 
\noindent
\\
\indent One possible approach to detect differences in location/scale/distributional structure is to consider a suitable notion of distance between two clusters or distributions and to use this distance to develop the clustering algorithm. In this paper, we develop novel nonparametric clustering algorithms whose building block is an appropriate notion of distance between two sets of observations. This is unlike most of the distance based clustering techniques which generally use distances between observations to build the clustering algorithm. 
\par
Recall that the objective of clustering is to create homogeneous groups that are heterogeneous between themselves. The distance function approach taken in this paper achieves both ends. If we start with a single observation as a group, then we can add to it those observations which maximize the distance between the group thus formed and the rest of the dataset. The choice of the distance is crucial since it must simultaneously ensure that the observations that are added to the singleton group are as ``similar'' to the existing observation as possible while also ensuring that the distance between the two groups is as large as possible. This paper will exploit the Maximum Mean Discrepancy (MMD) \cite{gretton2006kernel} as a measure of distance between two groups of observations. Such an application of MMD for functional data clustering is novel to the best of our knowledge. We will show later in the paper that the MMD ensures that both within group homogeneity and between group heterogeneity can be achieved. To provide an insight into this phenomenon, we state the following property of MMD: given two distributions $P$ and $Q$, the MMD between any two mixtures $\alpha P + (1-\alpha)Q$ and $\beta P + (1-\beta) Q$ is strictly smaller than the MMD between $P$ and $Q$ with equality holding if and only if $(\alpha,\beta) = (0,1)$ or $(1,0)$ (cf. Remark \ref{MMDremark}). Thus, if we have a two-class clustering problem with unknown empirical distributions $\widehat{P}$ and $\widehat{Q}$ corresponding to the two samples, the above property ensures that if we want to maximize the MMD between any two sub-groups of this data, this is achieved when the data is partitioned ``perfectly'' (cf. Definition \ref{def:perfect}), that is, when the clusters are pure. Note that any subgroup of the original data can be viewed as a sample from a mixture distribution involving $\widehat{P}$ and $\widehat{Q}$, and this link will be exploited heavily in this paper to develop and investigate the clustering algorithms. The above property also ensures the following. Suppose that we start with a partition with one single observation as a group (the most obvious choice of a group containing observations from a single distribution), say $C_1$, and the rest of the data as another group, say $C_2$. Then, inclusion of another observation from $C_2$ to $C_1$ yielding the maximum MMD between the new groups also ensures that the transferred observation belongs to the same population as the observation already in $C_1$. Thus, heterogeneity between groups and homogeneity within groups are preserved simultaneously. Since in a typical clustering problem, the knowledge of $\widehat{P}$ and $\widehat{Q}$ are unavailable, we may view this as an oracle version of the clustering problem. Indeed, the clustering algorithms proposed in this paper are motivated by the above crucial phenomena, and the understanding of the performance of these algorithms can be obtained by a detailed investigation in such an oracle scenario (see Section \ref{oracle}). Our clustering algorithms will exploit the above phenomena to carefully find a path to achieve the optimal partition in a finite number of steps instead of searching for all possible partitions.
\\
\indent The main contributions of the paper are as follows: \\
	1. The algorithms proposed in this paper achieves near-perfect clustering of functional data, and they achieve perfect clustering in the oracle scenario, where we assume the knowledge of the empirical distributions of the observations from the different populations. Such a standpoint of an oracle analysis is not known for other functional clustering procedures. \\
	2. Although we use a binary splitting strategy in our clustering algorithms, the use of an appropriate maximization step prior to the binary splitting step reduces the computational overload. Indeed, the computational complexity turns out to be $O(n^2)$ for a sample of size $n$. \\
	3. The proposed algorithms can detect differences in location as well as in scale (even when there is no difference in location). \\
	4. One of the proposed algorithms (Algorithm CURBS-I) do not need a prior specification of the number of clusters $K$, and in fact, provides a bonafide estimator of $K$ in a finite number of steps. \\
	5. If we have a prior specification of $K$, an improvement of Algorithm CURBS-I (Algorithm CURBS-II) is proposed which is computationally cheaper and also achieves the Perfect Order Preserving (POP) property (see Definition \ref{pop}) in the oracle scenario.
\\
\indent The remaining part of this paper is organized as follows. In Section \ref{prologue}, we give a brief introduction to the MMD measure. In Section \ref{methodology}, we propose two clustering algorithms for the scenario where the number of clusters is unknown (Algorithm CURBS-I) as well as when it is prefixed (Algorithm CURBS-II). In Section \ref{oracle}, we provide the oracle versions of these algorithms, referred to as Algorithms CURBS-I* and CURBS-II*, respectively, in order to theoretically investigate the behaviour of the algorithms proposed in Section \ref{methodology} in the oracle setting. We establish that the Algorithm CURBS-I* achieves perfect clustering while the Algorithm CURBS-II* enjoys the POP property in this setting. A detailed simulation study and some real data  analyses are provided in Section \ref{simulation} and we compare the proposed algorithms with some of the existing clustering procedures for functional data. We conclude the paper with some additional remarks and future scope of research in Section \ref{epilogue}.

\section{Prologue: A metric on the space of probability distributions} \label{prologue}
\noindent
Let $\mathcal{P}$ be the set of Borel probability measures on a separable metric space $\mathcal{X}$. A well known measure of discrepancy between two measures $P$,$Q$ $\in$ $\mathcal{P}$ is the Maximum Mean Discrepancy (MMD) \cite{gretton2006kernel} which is defined as 
\begin{center}
	$\mbox{MMD}(P,Q) = \underset{f \in {\cal F}} {\sup}|\int f\,dP - \int f\,dQ|$
\end{center}
for an appropriate class of functions ${\cal F}$. The class ${\cal F}$ is chosen so that it is rich enough to discriminate between any two distinct probability measures. For instance, ${\cal F}$ can be chosen to be the space of bounded continuous functions on ${\cal X}$. However, it does not yield any useful finite sample estimate. On the other hand, if we take  ${\cal F}$ to be a reproducing kernel Hilbert space (RKHS) associated with a symmetric positive semi-definite kernel $k : {\cal X} \times {\cal X} \rightarrow \mathbb{R}$ such that $\int \sqrt{k(x,x)}dP(x) < \infty$  for all $P \in {\cal P}$, an equivalent closed form expression of the square of the MMD can be obtained as
\begin{align*}
	d(P,Q) =& \int\int k(x,x')\,dP(x)dP(x') + \int\int k(y,y')\,dQ(y)dQ(y') - 2\int\int k(x,y)\,dP(x)dQ(y).
\end{align*}
See, for example, \cite{sriperumbudur2010hilbert}. Hereafter, the square of the MMD between two distributions $P$ and $Q$ will be denoted by $d(P,Q)$. This form easily allows an unbiased finite sample estimate. The choices of $k$ for which the metric property holds are called characteristic kernels. For finite dimensional ${\cal X}$, several characteristic kernels have been well-studied (see \cite{fukumizu2007nips}), while for infinite dimensional Hilbert spaces, Gaussian, Laplacian kernels have been shown to be characteristic (see  \cite{ziegel2024characteristic}).

\par We can estimate MMD in the finite sample case as follows. Suppose $X = (x_1, x_2, \dots,x_{n_1})$ and $Y = (y_1, y_2, \dots,y_{n_2})$ are two ensembles of i.i.d. observations drawn from the distributions $P$ and $Q$ respectively. Then a V-statistic type estimate of $d(P,Q)$ is given by
\begin{align*}
d(\widehat{P},\widehat{Q}) =  &\frac{1}{n_1^2}\sum_{i,i'=1}^{n_1} k(x_i,x_{i'}) + \frac{1}{n_2^2}\sum_{j,j'=1}^{n_2} k(y_j,y_{j'}) - \frac{2}{n_1 n_2}\sum_{i=1}^{n_1}\sum_{j=1}^{n_2}k(x_i,y_j),
\end{align*}

\noindent
where $\widehat{P} = {n_1}^{-1}\sum_{i=1}^{n_1} \delta_{x_{i}}$ and $\widehat{Q} = {n_2}^{-1}\sum_{j=1}^{n_2} \delta_{y_{j}}$ are the empirical measures. In the next section, we will develop novel clustering techniques using the MMD, which will allow us to identify clusters having different probability distributions.  
 
\section{Clustering Methodology} \label{methodology}
Consider a sample ${\cal D} = \{X_1,X_2,\ldots,X_n\}$ taking values in the ambient space ${\cal X}$. The goal is to partition ${\cal D}$ into $K$ (may or may not be specified) groups such that the each group is as homogeneous as possible, and two distinct groups are as heterogeneous as possible. Mathematically, each group should have all observations from one single population/distribution and observations from any two distinct groups should belong to two distinct populations/distributions. To achieve this end, we need to have a measure of discrepancy between two groups/clusters. Consider two groups of observations $A = \{a_1, a_2, \dots, a_{n_1}\}$ and $B = \{b_1, b_2, \dots, b_{n_2}\}$. A measure of discrepancy between the sets $A$ and $B$ is given by
\begin{align}
    d_w(A, B) &= \frac{n_1 n_2}{n_1 + n_2 }d(\widehat{P}_A, \widehat{P}_B)\nonumber \\ &= \frac{n_1 n_2}{n_1 + n_2}\left\{\frac{1}{n_1^2}\sum_{i,i'=1}^{n_1} k(a_i, a_{i'}) + \frac{1}{n_2^2}\sum_{j,j'=1}^{n_2} k(b_j, b_{j'}) - \frac{2}{n_1 n_2}\sum_{i=1}^{n_1}\sum_{j=1}^{n_2}k(a_i, b_j)\right\}, \nonumber
\end{align}
where $\widehat{P}_A$ and $\widehat{P}_B$ are the empirical distributions of the observations of $A$ and $B$ respectively. Throughout this paper, we will take $k$ to be a characteristic kernel. Since $k$ is characteristic kernel, $d(\widehat{P}_A, \widehat{P}_B) \geq 0$ with equality holding if and only if $\widehat{P}_A= \widehat{P}_B$, i.e., $A = B$. Hence, $d_w(A,B) \geq 0$ and equality holds if and only if $A = B$. Here, the extra factor involving the cluster sizes acts as a balancing term to counter the situation where the group/cluster sizes may be particularly small (see \citep{jegelka2009generalized}). If we have observations from two populations (i.e., $K = 2$), then the goal is to partition ${\cal D}$ into subsets $A$ and $B$ such that $d_w(A,B)$ is maximum. For a general $K$ class problem, we would like to partition ${\cal D}$ into $K$ subsets $A_1, A_2, \dots, A_K$ such that $d_w(A_i,A_j)$ is maximized for all $i \neq j$ and $i, j \in \{1,2,\dots, K\}$.

\par We start with the case when $K = 2$. Since the objective is to get two homogeneous clusters, we start by selecting two subsets $C_1$ and $C_2$ of ${\cal D}$ with $C_1$ being a singleton. Since the objective is to get two homogeneous clusters, we start by selecting two subsets $C_1$ and $C_2$ of ${\cal D}$ where $C_1$ contains a single observation for which the $d_w
$ distance between $C_1$ and $C_2$ is maximum. This step automatically ensures that $C_1$ is homogeneous. The next step would be to add to $C_1$ those observations from $C_2$ which belong to the same population as the observation already in $C_1$. To achieve this, we start transferring observations from $C_2$ to $C_1$ one at a time such that the $d_w$-distance between $C_1$ and $C_2$ is maximum at each step. Formally, suppose that $C_i^{(r)}$ denotes the $i$-th set after the $r$-th iteration, $r = 1, 2, \dots, n-1$ and $i = 1, 2$.  After transferring observations by maximizing the $d_w$-distance at each iteration, we consider $C_1^{(N_{max})}$ and $C_2^{(N_{max})}$ as the final clusters where
\begin{align}
	N_{max} = \underset{r = 1,2,\dots,n-1}{\arg\max} d_{w}(C_1^{(r)},C_2^{(r)}).
\end{align}
This procedure, referred to as the Binary Splitting (BS) algorithm, is detailed below for any given dataset ${\cal E} \subseteq {\cal D}$.
\\
\small{
\textbf{Algorithm BS: Binary Splitting}
\smallskip
\hrule
\smallskip
\noindent \textbf{Input}: Dataset ${\cal E}$. \\
	Set $C_1^{(0)} = \emptyset$ and $C_2^{(0)} = {\cal E}$.
	\begin{enumerate}
		\item for $r=1,2,\dots,|{\cal E}|-1$ do
		\item Find an element $c$ from $C_2^{(r-1)}$ which maximizes $d_w(C_{1}^{(r-1)} \cup \{c\}, C_{2}^{(r-1)}\backslash\{c\})$.
		\item $C_1^{(r)} = C_{1}^{(r-1)} \cup \{c\}$ and $C_2^{(r)} = C_{2}^{(r-1)}\backslash\{c\}$.
		\item endfor
	\end{enumerate}
\textbf{Output}: $N_{max} := \underset{r = 1,2,\dots,|{\cal E}|-1}{\arg\max} d_{w}(C_1^{(r)},C_2^{(r)})$,  $C_1^{(N_{max})}$ and $C_2^{(N_{max})}$
}
\smallskip
\hrule
\vspace{0.05in}

\noindent\\\noindent The following Remark \ref{dw-computation} states some interesting properties of the $d_w$-distance which facilitates computation of the Algorithm BS. 
\begin{remark}  \label{dw-computation}
Suppose that in the $(r+1)$-th iteration of the Algorithm BS, the observation $c$ is transferred from $C_2^{(r)}$ to $C_1^{(r)}$, so that $C_2^{(r+1)} = C_{2}^{(r)}\backslash\{c\}$ and $C_1^{(r+1)} = C_1^{(r)} \cup \{c\}$. Then, it can be shown that
     \begin{equation}
         d_w(C_1^{(r+1)}, C_2^{(r+1)}) - d_w(C_1^{(r)}, C_2^{(r)}) = d_w(C_2^{(r)}\backslash\{c\},\{c\}) - d_w(C_1^{(r)},\{c\}).  \label{eq1}
     \end{equation}
This identity has a two-fold implication: \\
(a) The change in the $d_w$-distance between the $r$-th and the $(r+1)$-th iteration only depends on the element $c$ and the two sets in the $r$-th iteration. If $c$ is closer to $C_1^{(r)}$ than $C_2^{(r)}\backslash\{c\}$, then the $d_w$-distance increases from the $r$-th to the $(r+1)$-th iteration and vice-versa.  \\
(b) The computation of $d_w(C_1^{(r+1)}, C_2^{(r+1)})$ can be easily done by adding $d_w(C_1^{(r)}, C_2^{(r)})$ to the difference in the right hand side of \eqref{eq1}. The latter computation is cheaper since one of the sets involved is a singleton.
\end{remark}

When the number of populations $K$ is larger than two, the Algorithm BS has must be applied repeatedly. The details will be given in the next subsection. We will propose clustering algorithms for both situation: (a) when the number of populations is unspecified, and (b) when it is specified. In the former case, our algorithm will automatically provide a bonafide estimator of the number of populations. 
\subsection{Methodology for unspecified number of populations}
\indent Since the number of populations is unspecified, we will have to first check whether the sample contains observations from more than one distribution or not. To achieve this goal, we propose an algorithm, referred to as the Algorithm SCC, to detect whether a dataset can be ascertained as a single cluster. 
\\
\small{
	\textbf{Algorithm SCC: Single Cluster Checking}
	\smallskip
	\hrule
	\smallskip
	\noindent \textbf{Input}: Dataset ${\cal E}$. \\
	Set $C_1^{(0)} = \emptyset$ and $C_2^{(0)} = {\cal E}$.
	\begin{enumerate}
		\item for $r=1,2,\ldots,|{\cal E}|-1$ do
		\item Find an element $c$ from $C_2^{(r-1)}$ which maximizes $d_w(C_{1}^{(r-1)} \cup \{c\}, C_{2}^{(r-1)}\backslash\{c\})$.
		\item Set $C_1^{(r)} = C_{1}^{(r-1)} \cup \{c\}$ and $C_2^{(r)} = C_{2}^{(r-1)}\backslash\{c\}$.
		\item endfor
	\end{enumerate}
    Set $V = \underset{r = 1,2,\dots,|{\cal E}|-1}{\max} d_{w}(C_1^{(r)},C_2^{(r)})/\underset{r = 1,2,\dots,|{\cal E}|-1}{\min} d_{w}(C_1^{(r)},C_2^{(r)})$. \\
    Set $N_{max} := \underset{r = 1,2,\dots,|{\cal E}|-1}{\arg\max} d_{w}(C_1^{(r)},C_2^{(r)})$ and $N_{min} := \underset{r = 1,2,\dots,|{\cal E}|-1, \  r \neq N_{max}}{\arg\min} d_{w}(C_1^{(r)},C_2^{(r)})$.\\
    Set $H = N_{max}(|{\cal E}| - N_{max})/[N_{min}(|{\cal E}| - N_{min})]$. \\
    If $N_{min} \leq N_{max}$ set $R = N_{max}(|{\cal E}| - N_{min})/[N_{min}(|{\cal E}| - N_{max})]$ else set $R = N_{min}(|{\cal E}| - N_{max})/[N_{max}(|{\cal E}| - N_{min})]$.  
    \textbf{Output}: If $|V-(R/H)| > |V-1|$ out = 0 (Accept) else out = 1 (Reject).
}
\smallskip
\hrule
\bigskip

\indent This algorithm will be justified through an oracle analysis provided in Section \ref{oracle}. If the Algorithm SCC determines that the entire dataset consists of observations from more than one populations, we will have to apply the Algorithm BS and the Algorithm SCC repeatedly to split any subset of the whole data into two disjoint subsets and also check whether each of these two subsets conform to the single population hypothesis. Now, we describe the algorithm, referred to as Clustering Using Recursive Binary Splitting (CURBS)-I, for clustering the entire dataset. We will discuss the justification of this algorithm through an oracle analysis in Section \ref{oracle}. \\
\small{
	\textbf{Algorithm CURBS-I}
	\smallskip
	\hrule
	\smallskip
	\noindent \textbf{Input}: Dataset ${\cal D}$. 
\begin{enumerate}
    \item Apply Algorithm SCC to ${\cal D}$.
    \item If Accepted, END, else apply Algorithm BS to ${\cal D}$.
    \item Apply Algorithm SCC to the output clusters obtained. If Algorithm SCC is Accepted for a cluster, keep that cluster unchanged, else apply Algorithm BS on that cluster and obtain two new clusters.
    \item Repeat Step 3 until Algorithm SCC is Accepted for each of the clusters obtained. 
\end{enumerate}
    \textbf{Output}: $\widehat{K}$ and the clusters $C_1,C_2,\dots,C_{\widehat{K}}$, where $\widehat{K}$ is the number of clusters finally obtained.}
    \smallskip
\hrule

\smallskip

\subsection{Methodology for prefixed number of populations}
In case the number of populations $J$ is specified in advance (may or may not be the true number of populations $K$), we will have to modify the Algorithm CURBS-I to accommodate this information. More specifically, in place of the Algorithm SCC which was applied to all the clusters obtained in a given stage (cf. steps 3 and 4 of the Algorithm CURBS-I), we deliberately split each of these clusters using the Algorithm BS. Consequently, it may so happen that some of the clusters containing observations from a single population is split into two sub-clusters which now needs to be merged. This merging is done based on $d_w$-distances between sub-clusters, and it also allows us to control the number of output clusters using the information about $J$. The details of such an algorithm, referred to as the Algorithm CURBS-II, is given below. We will discuss the justification of this algorithm through an oracle analysis in Section \ref{oracle}. \\
\small{
	\textbf{Algorithm CURBS-II}
	\smallskip
	\hrule
	\smallskip
	\noindent \textbf{Input}: Dataset ${\cal D}$ and specified number $J ~(> 1)$ of clusters.
	\begin{enumerate}
		\item Apply Algorithm BS to ${\cal D}$.
		\item If $J=2$ then END, else set $i=2$.
		\item Apply Algorithm BS to each of the $i$ clusters obtained.
		\item Compute $d_w$-distances between all the $2i$ clusters obtained in Step 3 and arrange them in the increasing order of magnitude.
		\item Merge those $(i-1)$ pair(s) of clusters obtained in Step 3 having the least $d_w$-distances. 
		\item  If $i=J-1$ END, else update $i=i+1$ and go to Step 3.
	\end{enumerate}
	\textbf{Output}: The clusters $C_1,C_2,\ldots,C_{J}$.
	}
\smallskip
\hrule

\smallskip

\section{Oracle analysis of proposed clustering algorithms} \label{oracle}
\indent This section is aimed at providing a motivation and justification for the algorithms detailed in Section 3. The analysis done here will be based on an oracle situation which will be described shortly. To this end, consider first the case when the observations come from two different distributions. Denote the observations by $X_1, X_2, \dots, X_{n_1}$ from a distribution ${P}_1$ and $Y_1, Y_2, \dots, Y_{n_2}$ from a distribution ${P}_2$. Let us denote the empirical distributions of the $X_i$'s and the $Y_j$'s by $\widehat{P}_1$ and $\widehat{P}_2$, respectively. Suppose we have two clusters $S_1 = \{X_1, \dots, X_{m_1}, Y_1, \dots, Y_{m_2}\}$ and $S_2 = \{X_1, \dots, X_{m_3}, Y_1, \dots, Y_{m_4}\}$, for $m_1 + m_3 \le n_1$ and $m_2+m_4 \le n_2$. As opposed to the usual empirical distributions of the sets $S_1$ and $S_2$, we shall define the ``representative distributions'' of the sets $S_1$ and $S_2$. Note that $S_1$ can be considered as a possible realization from $\alpha\widehat{P}_1 + (1 - \alpha)\widehat{P}_2$ for some $\alpha$ $\in$ [0,1]. The most canonical choice of $\alpha$ is its natural estimator, namely, the sample proportion of observations in $S_1$ coming from $P_1$, i.e., $m_1$/($m_1 + m_2$). The representative distribution of $S_1$ is defined as $\widetilde{P}_{S_{1}} = m_1$/($m_1 + m_2$)$\widehat{P}_1 + m_2$/($m_1 + m_2$)$\widehat{P}_2$. Note in passing that if we draw a random sample of size ($m_1 + m_2$) from $\widetilde{P}_{S_{1}}$, the event that the outcome will be $S_1$ has a positive probability. The representative distribution of $S_2$, denoted by $\widetilde{P}_{S_{2}}$, is given by  $\widetilde{P}_{S_{2}} = m_3/(m_3 + m_4)\widehat{P}_1 + m_4/(m_3 + m_4)\widehat{P}_2$. Observe that even if $S_1$ has some other configuration of $m_1$ elements from $\widehat{P}_1$ and $m_2$ elements from $\widehat{P}_2$, its representative distribution remains unchanged while the empirical distribution changes. This property of representative distributions allows us to capture the intrinsic composition of a set of observations. Note that in a clustering problem, the degree of homogeneity/heterogeneity of a cluster is determined not by which elements of each population are present in the cluster, but by the proportions of elements corresponding to the different populations (that is, the intrinsic composition of a cluster). This motivates us to use the representative distributions in the oracle analysis. In case the observations come from $K \:(> 2)$ different distributions, the definitions of the representative distributions can be modified similarly.
\\
\indent Our oracle analysis will assume the knowledge of the empirical distributions, thus converting the unsupervised problem into a supervised problem. This will imply that for any cluster, we will be able to construct its representative distribution. The theoretical analyses of this section shows that the full understanding of the near-perfect clustering achieved by the Algorithms CURBS-I and CURBS-II (see Sections \ref{oracle} and \ref{simulation}) can be effectively obtained through the perspective of representative distributions. More specifically, we provide oracle analogs of the algorithms in Section \ref{methodology} using a version of the MMD distance between representative distributions. These oracle analogs of the Algorithm CURBS-I and CURBS-II will be theoretically shown to have perfect clustering property. Indeed, the algorithms in Section \ref{methodology} should be viewed as data-driven versions of these oracle algorithms. 

\begin{remark} \label{MMDremark}
	One of the key components in the proof of the perfect clustering behaviour of the oracle algorithms is the following remarkable property of MMD. If $F$ and $G$ are two probability distributions and $\alpha,\beta \in [0,1]$, then
\begin{eqnarray}
	d(\alpha F + (1-\alpha)G, \beta F + (1-\beta) G) \ = \ (\alpha - \beta)^2d(F,G). \label{MMDproperty}
\end{eqnarray}
The above equation implies that the square of the MMD between a pair of mixture distributions with two components will always be less than or equal to the square of the MMD between the two component distributions with equality holding if and only if $(\alpha,\beta) = (0,1)$ or $(1,0)$. Indeed, we are not aware of any other probability metric which enjoys this property, and this motivated us to use MMD in our clustering algorithms. Since the representative distributions are also mixture distributions, the clustering algorithms are designed by carefully utilizing the property in \eqref{MMDproperty}. Moreover, all the theoretical results of this paper also heavily rely on the same property.
\end{remark}

\indent Now, we formally describe the oracle algorithms starting with a measure of dissimilarity between two clusters. An oracle version of the $d_w$-distance between two clusters $S_1$ and $S_2$ is given by 
\begin{eqnarray*}
d_{w}^{*}(S_1,S_2) = \frac{|S_1|~|S_2|}{|S_1| + |S_2|}d\left(\widetilde{P}_{S_1},\widetilde{P}_{S_2}\right),
\end{eqnarray*}
where for any set $S$, its cardinality is denoted by $|S|$. Note that the $d_w^*$ measure defined above uses the representative distributions of the two sets whereas the $d_w$-distance uses the empirical distributions of the two sets. The oracle analog of the Algorithm BS, referred to as the Algorithm BS*, is detailed below.\\
\noindent
\small{
\textbf{Algorithm BS*: Oracle Binary Splitting}
	\smallskip
	\hrule
	\smallskip
	\noindent \textbf{Input}: Dataset ${\cal E}$. \\
	Set $C_1^{(0)} = \emptyset$ and $C_2^{(0)} = {\cal E}$.
	\begin{enumerate}
		\item for $r=1,2,\ldots,|{\cal E}|-1$, do
		\item Find an element $c$ from $C_2^{(r-1)}$ which maximizes $d_w^*(C_{1}^{(r-1)} \cup \{c\}, C_{2}^{(r-1)}\backslash\{c\})$.
		\item Set $C_1^{(r)} = C_{1}^{(r-1)} \cup \{c\}$ and $C_2^{(r)} = C_{2}^{(r-1)}\backslash\{c\}$.
		\item endfor
	\end{enumerate}
	\textbf{Output}: $N_{max}^* := \underset{r = 1,2,\dots,|{\cal E}|-1}{\arg\max} d_{w}^*(C_1^{(r)},C_2^{(r)})$,  $C_1^{(N_{max}^*)}$ and $C_2^{(N_{max}^*)}$.
}
\smallskip
\hrule
\vspace{0.05in}
Observe that in Step 2, the choice of the observation $c$ is not unique. Indeed, if an observation from $P_1$ is a maximizer, then so is any other observation from $P_1$ in the set $C_2$. 
We will now discuss the oracle versions of the other proposed algorithms.

\subsection{Oracle version of Algorithm CURBS-I} \label{curbs1_oracle}
\indent In the Algorithm CURBS-I, the question of whether a cluster has observations from a single population was addressed using the Algorithm SCC. In the oracle scenario, single cluster checking can be done simply by investigating the values of $d_w^*$ over different iterations as in the Algorithm BS*. If a cluster contains observations from a single distribution, then the values of $d_w^*$ over all iterations will be zero. Otherwise, the maximum value of $d_w^*$ will be greater than its minimum value. Recall that in the Algorithm SCC, we accepted that the cluster has observations from more than one distribution if $V$ is closer to $R/H$ than $1$. This can now be motivated as follows. Let us define the oracle versions of $R$ and $H$ as given below when the dataset $\mathcal {D}$ consists of $n$ observations: 
\begin{eqnarray}
&& R^* = \left \{
		\begin{array}{l@{\quad} l }
			\frac{N^{*}_{max}(n-N^{*}_{min})}{N^{*}_{min}(n-N^{*}_{max})} & \text{if $N^{*}_{min} < N^{*}_{max}$}\\\\
			\frac{N^{*}_{min}(n-N^{*}_{max})}{N^{*}_{max}(n-N^{*}_{min})} & \text{if $N^{*}_{min} > N^{*}_{max}$}
		\end{array}
		\right . \ \ \mbox{and} \ \
H^* = \frac{N^{*}_{max}(n-N^{*}_{max})}{N^{*}_{min}(n-N^{*}_{min})},		\label{oracle-vrh-1}
\end{eqnarray}
where $C_1^{(r)}$'s and $C_2^{(r)}$'s are as in the Algorithm BS*, and 
\begin{eqnarray}
N^{*}_{max} := \underset{1 \leq r \leq n-1}{\arg\max} d_w^*(C_1^{(r)},C_2^{(r)}) \ \mbox{and} \  N^{*}_{min} := \underset{1 \leq r \leq n-1,\ r \neq N^*_{max}}{\arg\min}d_w^*(C_1^{(r)},C_2^{(r)}).    \label{nmax-nmin-def}
\end{eqnarray}
In case the number of populations $K > 1$, the oracle version of $V$ is defined as 
\begin{eqnarray}
V^* = \frac{\underset{1 \leq r \leq n-1}{\max}d_w^*(C_{1}^{(r)},C_{2}^{(r)})}{\underset{1 \leq r \leq n-1}{\min}d_w^*(C_{1}^{(r)},C_{2}^{(r)})}.		\label{oracle-vrh-2}
\end{eqnarray}

\begin{theorem}\label{lma1}
  Consider a set of $n$ observations of sizes $n_1, n_2, \ldots, n_K$ from $K > 1$ populations. The following properties are true: \\
  (a) $V^* = R^*$, \\
  (b) there exists $n_0 \geq 1$ such that $N^*_{min} = n-1$ for all $n \geq n_0$ provided that $\liminf_{n \rightarrow \infty} (n_i/n) > 0$, for all $i=1,2,\ldots,K$, and \\
  (c) under the condition in part (b), there exists $n^* \geq n_0$ such that $H^* > 1$ for $n \geq n^*$.
\end{theorem}

The quantities $n_0$ and $n^*$ in the previous theorem depend on the inter-population MMD distances between the $K$ populations and also the proportions of each population present in the sample (see the proof of Theorem \ref{lma1} in the Supplementary Material). Since $R^* > H^*$ (which follows from \eqref{oracle-vrh-1} and \eqref{nmax-nmin-def}), it can be deduced from Theorem \ref{lma1} that when the number of populations $K > 1$, we have
\begin{eqnarray}
|V^* - (R^*/H^*)| = |R^* - (R^*/H^*)| = R^* - (R^*/H^*) < R^* - 1 = V^* - 1 = |V^* - 1|. \label{scc*-validity}
\end{eqnarray}
If the cluster has observations from a single distribution, then the numerator and the denominator of $V^*$ are both equal to zero. In the empirical situation, it is observed that the numerator and the denominator of $V$ are close to each other, and hence, we use the convention that $V^* = 1$ when the observations come from a single distribution. Thus, since $R^* > H^*$, we have $|V^* - (R^*/H^*)| > |V^* - 1|$. These discussions inspire an oracle version of the Algorithm SCC, referred to as the Algorithm SCC*.\\
\bigskip
\noindent
\noindent
\small{
	\textbf{Algorithm SCC*: Oracle Single Cluster Checking}
	\smallskip
	\hrule
	\smallskip
	\noindent \textbf{Input}: Dataset ${\cal E}$. \\
	Set $C_1^{(0)} = \emptyset$ and $C_2^{(0)} = {\cal E}$.
	\begin{enumerate}
		\item for $r=1,2,\ldots,|{\cal E}|-1$, do
		\item Find an element $c$ from $C_2^{(r-1)}$ which maximizes $d_w^*(C_{1}^{(r-1)} \cup \{c\}, C_{2}^{(r-1)}\backslash\{c\})$.
		\item Set $C_1^{(r)} = C_{1}^{(r-1)} \cup \{c\}$ and $C_2^{(r)} = C_{2}^{(r-1)}\backslash\{c\}$.
		\item endfor
	\end{enumerate}
Set $V^* = \underset{1 \leq r \leq |{\cal E}|-1}{\max}d_w^*(C_{1}^{(r)},C_{2}^{(r)})/\underset{1 \leq r \leq |{\cal E}|-1}{\min}d_w^*(C_{1}^{(r)},C_{2}^{(r)})$. \\	
Set $N^{*}_{max} := \underset{1 \leq r \leq |{\cal E}|-1}{\arg\max} d_w^*(C_1^{(r)},C_2^{(r)})$ and $N^{*}_{min} := \underset{1 \leq r \leq |{\cal E}|-1,\ r \neq N^*_{max}}{\arg\min}d_w^*(C_1^{(r)},C_2^{(r)})$. \\
Set $H^* = N^{*}_{max}(|{\cal E}|-N^{*}_{max})/[N^{*}_{min}(|{\cal E}|-N^{*}_{min})]$. \\
If $N^{*}_{min} < N^{*}_{max}$, set $R^* = N^{*}_{max}(|{\cal E}|-N^{*}_{min})/[N^{*}_{min}(|{\cal E}|-N^{*}_{max})]$ else set $R^* = N^{*}_{min}(|{\cal E}|-N^{*}_{max})/[N^{*}_{max}(|{\cal E}|-N^{*}_{min})]$. \\
    \textbf{Output}: If $|V^* - (R^*/H^*)| > |V^* - 1|$ out = 0 (Accept) else out = 1 (Reject).
}
\smallskip
\hrule
\bigskip
\noindent Observe that in view of Theorem \ref{lma1} and the discussion after \eqref{scc*-validity}, the decision made by the Algorithm SCC* is always correct for all sufficiently large sample sizes. We now provide the oracle version of the Algorithm CURBS-I, referred to as the Algorithm CURBS-I*, by applying the Algorithms BS* and SCC* repeatedly.
\\
\bigskip
\noindent
\noindent
\small{
	\textbf{Algorithm CURBS-I*}
	\smallskip
	\hrule
	\smallskip
	\noindent \textbf{Input}: Dataset ${\cal D}$. 
	\begin{enumerate}
		\item Apply Algorithm SCC* to ${\cal D}$.
		\item If Accepted, END, else apply Algorithm BS* to ${\cal D}$.
		\item Apply Algorithm SCC* to the output clusters obtained. If Algorithm SCC* is Accepted for a cluster, keep that cluster unchanged, else apply Algorithm BS* on that cluster and obtain two new clusters.
		\item Repeat Step 3 until Algorithm SCC* is Accepted for each of the clusters obtained. 
	\end{enumerate}
	\textbf{Output}: $K^*$ and the clusters $C_1,C_2,\dots,C_{K^*}$, where $K^*$ is the number of clusters finally obtained.}
\smallskip
\hrule

\medskip
\indent Next, we discuss the behaviour of the Algorithm CURBS-I* for different values of the number of populations $K$. The backbone of this analysis is the understanding of the  values of $d_w^*$ over different iterations as in the Algorithm BS*. \\
\noindent \underline{Case 1 ($K = 2$)}: Denote the two empirical distributions by $\widehat{P}_1$ and $\widehat{P}_2$, which have been constructed based on samples of sizes $n_1$ and $n_2$, respectively. Denote $n = n_1 + n_2$. As in the Algorithm BS*, if we start with two sets $C_1$ and $C_2$ such that $C_1$ is empty and $C_2$ contains all the observations, then an observation from $\widehat{P}_1$ will be transferred to $C_1$ in the first iteration if and only if $n_1 \leq n_2$ (See part $(a)$ of Lemma 1 in the Supplementary Material). Without loss of generality, let us assume that an observation from $\widehat{P}_1$ is transferred to $C_1$ in the first iteration, and as such $n_1 \leq n_2$. Then, till the $n_1$-th iteration, observations only from $\widehat{P}_1$ will be transferred to $C_1$ at every iteration and the $d_w^*$ curve will keep on increasing (See part $(a)$ of Lemma 1 in the Supplementary Material). 
In fact, it can be shown that 
   \begin{eqnarray}
    d_w^*(C_{1}^{(r)},C_{2}^{(r)}) = \left \{
  \begin{array}{l@{\quad} l }
  \frac{rn_2^2}{n(n-r)}d(\widehat{P}_1,\widehat{P}_2) ; & \text{if $r = 1, 2, \dots, n_1$},\\\\
  \frac{n_1^2(n-r)}{rn}d(\widehat{P}_1,\widehat{P}_2) ; & \text{if $r = n_1 + 1, n_1 + 2, \dots, n$}.
  \end{array}
  \right . \label{dwcurve-two}
  \end{eqnarray}
We have plotted the $d_w^*$ curve (given by \eqref{dwcurve-two}) in plot (a) of Figure \ref{fig:curves} (dashed curve) with the value of $d(\widehat{P}_1,\widehat{P}_2)$ as given by Model (2a) in Section \ref{simulation} for sample sizes $n_1 = n_2 = 100$. We have also provided the $d_w$ curve for the above model (solid curve) in the same plot, and it is seen that the two curves are quite similar in shape. This similarity plays a crucial role in the near-perfect clustering performance of the Algorithm CURBS-I observed for the above model (see Table \ref{tab1}).

\begin{figure}[ht]
	\centering
	\includegraphics[width=0.8\linewidth,height=0.4\textheight]{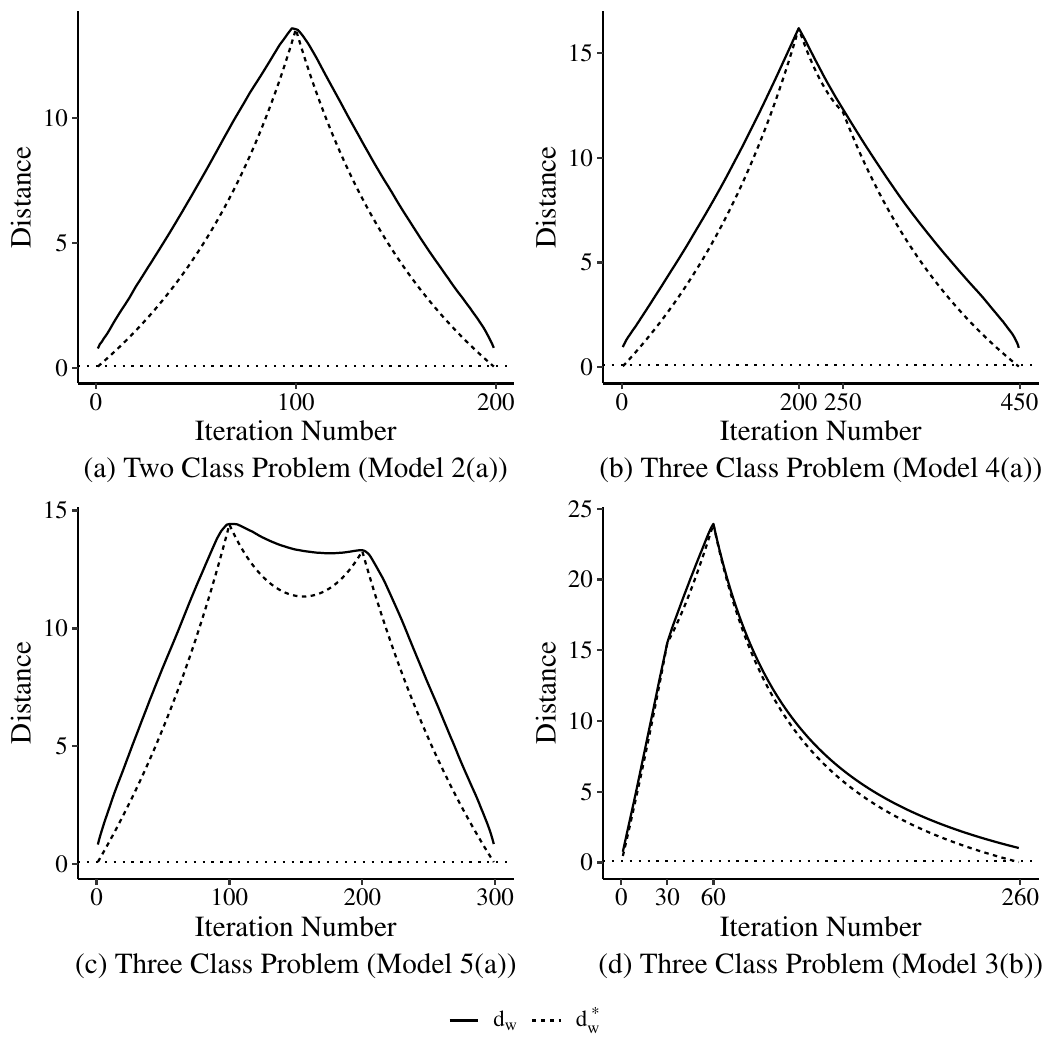}
	\caption{Curves of $d_w$ and $d_w^*$}
	\label{fig:curves}
\end{figure}
\noindent
\underline{Case 2 ($K = 3$)}: Denote the three empirical distributions by $\widehat{P}_1$, $\widehat{P}_2$ and $\widehat{P}_3$, which have been constructed based on samples of sizes $n_1$, $n_2$ and $n_3$, respectively. Denote $n = n_1 + n_2 + n_3$. We start with two sets $C_1$ and $C_2$ such that $C_1$ is empty and $C_2$ contains all observations and continue transferring the observations from $C_2$ to $C_1$. An observation from $\widehat{P}_1$ will be transferred to $C_1$ in the first iteration if and only if the following two conditions hold (see part $(a)$ of Lemma 2 in the Supplementary Material). \\
$(A1)$ ($n_2 - n_1$)$d(\widehat{P}_1,\widehat{P}_2) + n_3$\{$d(\widehat{P}_1,\widehat{P}_3) - d(\widehat{P}_2,\widehat{P}_3)$\} $\geq 0$,
\\
$(A2)$ ($n_3 - n_1$)$d(\widehat{P}_1,\widehat{P}_3) + n_2$\{$d(\widehat{P}_1,\widehat{P}_2) - d(\widehat{P}_2,\widehat{P}_3)$\} $\geq 0$.
\\
Under the same two conditions $(A1)$ and $(A2)$, it can be ensured that till the $n_1$-th iteration, observations only from $\widehat{P}_1$ will be transferred to $C_1$ at each iteration and the $d_w^*$ curve will keep on increasing (see part $(a)$ of Lemma 2 in the Supplementary Material). In the balanced case (when $n_1 = n_2 = n_3$), the conditions $(A1)$ and $(A2)$ boils down to requiring $$d(\widehat{P}_2,\widehat{P}_3) \leq \min\{d(\widehat{P}_1,\widehat{P}_2),d(\widehat{P}_1,\widehat{P}_3)\},$$ 
i.e., the distributions $\widehat{P}_2$ and $\widehat{P}_3$ are closer to each other than to  $\widehat{P}_1$. This also explains why at the first step $\widehat{P}_1$ gets separated from the cluster containing observations from $\widehat{P}_2$ and $\widehat{P}_3$. Next, assuming conditions $(A1)$ and $(A2)$, an observation from $\widehat{P}_2$ will be transferred to $C_1$ after the $n_1$-th iteration if and only if the following condition holds (see part $(b)$ of Lemma 2 in the Supplementary Material).
\\
$(A3)$ $n_1(n_2 + n_3 - 1)\{d(\widehat{P}_1,\widehat{P}_3) - d(\widehat{P}_1,\widehat{P}_2)\} + (n_3 - n_2)(n_1 + 1)d(\widehat{P}_2,\widehat{P}_3)$ $\geq 0$.    
\\
It can also be shown that under the assumptions $(A1)$, $(A2)$ and $(A3)$, all the observations from $\widehat{P}_2$ will be transferred to $C_1$ in the subsequent iterations till the $(n_1+n_2)$-th iteration (see part $(b)$ of Lemma 2 in the Supplementary Material). In the balanced case, the above condition reduces to requiring $d(\widehat{P}_1,\widehat{P}_3) \geq d(\widehat{P}_1,\widehat{P}_2)$, i.e., $\widehat{P}_2$ is closer to $\widehat{P}_1$ than $\widehat{P}_3$.  This condition is expected since $\widehat{P}_2$ being closer to $\widehat{P}_1$ than $\widehat{P}_3$ would naturally imply that it should be the group that first gets merged to $\widehat{P}_1$. Further, one can explicitly derive the form of the $d_w^*$ curve over all iterations (see Lemma 2 in the Supplementary Material), namely, 
 \begin{align}
	&d_w^*(C_{1}^{(r)},C_{2}^{(r)}) \nonumber \\
	&= \left \{    
	\begin{array}{l@{\quad} l }
		\frac{r}{n(n-r)}\{n_2(n_2+n_3)d(\widehat{P}_1,\widehat{P}_2)+n_3(n_2+n_3)d(\widehat{P}_1,\widehat{P}_3)-n_2n_3d(\widehat{P}_2,\widehat{P}_3)\} ; \quad \text{if $r = 1, 2, \dots, n_1$},\\\\
		\frac{nn_1-r(n_1+n_3)}{n}\{\frac{n_1}{r}d(\widehat{P}_1,\widehat{P}_2) - \frac{n_3}{n-r}d(\widehat{P}_2,\widehat{P}_3)\}+\frac{n_1n_3}{n}d(\widehat{P}_1,\widehat{P}_3) ; \quad  \text{if $r = n_1 + 1, \dots, n_1 + n_2$},\\\\     \label{dwcurve-three}
		\frac{n-r}{rn}\{n_2(n_1+n_2)d(\widehat{P}_2,\widehat{P}_3)+n_1(n_1+n_2)d(\widehat{P}_1,\widehat{P}_3)-n_1n_2d(\widehat{P}_1,\widehat{P}_2)\} ; \quad \text{if $r = n_1 + n_2 + 1, \dots, n$}.
	\end{array}
	\right. 
\end{align}
As mentioned earlier, the above formula reveals that the $d_w^*$ curve keeps on increasing till the $n_1$-th iteration. Moreover, the $d_w^*$ curve keeps on decreasing after the $(n_1+n_2)$-th iteration. Unlike the case $K = 2$, the $d_w^*$ curve may increase or decrease between the $n_1$-th and the $(n_1+n_2)$-th iterations depending on the sample sizes and the values of the pairwise MMD between the three empirical distributions. It can be shown that between the $n_1$-th and the $(n_1+n_2)$-th iterations, the $d_w^*$ curve will be a convex function of the iteration number $r$ (See part $(b)$ of Lemma 2 in the Supplementary Material). From \eqref{dwcurve-three}, doing some simple algebraic calculations, it can be shown that for $r=n_1+1,n_1+2,\dots,n_1+n_2$,
\begin{align} \label{eq:62}
    d_w^*(C_{1}^{(r)},C_{2}^{(r)})-d_w^*(C_{1}^{(r-1)},C_{2}^{(r-1)}) = \frac{n_3^2}{(n-r)(n-r+1)}d(\widehat{P}_2,\widehat{P}_3)-\frac{n_1^2}{r(r-1)}d(\widehat{P}_1,\widehat{P}_2), 
\end{align}
which is an increasing function of $r$. Thus, one possible scenario is when it is increasing, and this holds if and only if $d_w^*(C_{1}^{(n_1+1)},C_{2}^{(n_1+1)})>d_w^*(C_{1}^{(n_1)},C_{2}^{(n_1)})$ (as the right hand side of \eqref{eq:62} is an increasing function of iteration number). One can also write the previous inequality in the following equivalent form:
\begin{align}\label{eq:4}
    \frac{n_1}{n_1+1}d(\widehat{P}_1,\widehat{P}_2) < \frac{n_3^2}{(n_2+n_3)(n_2+n_3-1)}d(\widehat{P}_2,\widehat{P}_3).
\end{align}
Observe that in the balanced case, the above relation \eqref{eq:4} cannot hold in view of the Assumptions $(A1)$, $(A2)$ and $(A3)$. Thus, the $d_w^*$ curve cannot increase after the $n_1$-th iteration in the balanced case. Another possible scenario is when the $d_w^*$ curve is decreasing after the $n_1$-th iteration which holds if and only if $d_w^*(C_{1}^{(n_1+n_2)},C_{2}^{(n_1+n_2)})<d_w^*(C_{1}^{(n_1+n_2-1)},C_{2}^{(n_1+n_2-1)})$ (as the right hand side of \eqref{eq:62} is an increasing function of iteration number). The following condition is equivalent to the previous inequality:
\begin{align}\label{eq:5}
    \frac{n_3}{n_3+1}d(\widehat{P}_2,\widehat{P}_3) < \frac{n_1^2}{(n_1+n_2)(n_1+n_2-1)}d(\widehat{P}_1,\widehat{P}_2).
\end{align}
In the balanced case, the above relation \eqref{eq:5} holds if and only if $d(\widehat{P}_1,\widehat{P}_2) > 4d(\widehat{P}_2,\widehat{P}_3)$, using $n_3+1 \approx n_3$ and $n_1+n_2+1 \approx n_1+n_2$.
In general, if neither of the conditions \eqref{eq:4} and \eqref{eq:5} holds, the $d_w^*$ curve will be a U-shaped curve between the $n_1$-th and the $(n_1+n_2)$-th iterations. \\
\indent We have plotted the $d_w^*$ curve (given by \eqref{dwcurve-three}) with the values of the square MMDs between the three empirical distributions as given by each of the three models, namely, Model (3a), Model (5a)  and Model (3b) given in Section \ref{simulation}  (see the plots (b), (c) and (d) of Figure \ref{fig:curves}). The sample sizes for Model (3a) and Model (5a) are taken to be $n_1 = n_2 = n_3  = 100$, while for Model (3b), we have taken $n_1 = 30, n_2 = 30, n_3 = 200$. In Model (3a), neither of the conditions \eqref{eq:4} and \eqref{eq:5} holds, and therefore, $d_w^*$ will be U-shaped between the $n_1$-th and the $(n_1+n_2)$-th iterations. The populations chosen for Model (5a) satisfies condition \eqref{eq:5}, and hence the $d_w^*$ curve will decrease between the $n_1$-th and the $(n_1+n_2)$-th iterations. The unequal sample sizes chosen for Model (3b) ensures that the condition \eqref{eq:4} is satisfied, in which case, the $d_w^*$ curve will be increasing between the $n_1$-th and the $(n_1+n_2)$-th iterations. We have also plotted the $d_w$ curves corresponding to each of the above models with their respective $d_w^*$ curves. The marked similarity in shapes of plays a crucial role in the near-perfect clustering performance of the Algorithm CURBS-I (see Table \ref{tab1}).\\
\noindent \underline{Case 3 ($K > 3$)}:  Denote the $K$ empirical distributions by $\widehat{P}_1, \dots, \widehat{P}_K$, which have been constructed based on samples of sizes $n_1, \dots, n_K$, respectively. We start with two sets $C_1$ and $C_2$ such that $C_1$ is empty and $C_2$ contains all observations and continue transferring the observations from $C_2$ to $C_1$. Consider the iteration where $C_1$ contains observations from $\widehat{P}_1, \ldots, \widehat{P}_m$ for some $1 \leq m \leq K-2$ (without loss of generality). A crucial observation is that the analysis of the next iteration can be viewed from the perspective of a three-population problem considered earlier: the mixtures of the distributions in $C_1$ plays the role of the first population, the distribution of the candidate to be transferred from $C_2$ to $C_1$ plays the role of the second population, and the mixture of the remaining distributions in $C_2$ plays the role of the third population. This analogy with the three-population problem is also exploited to show that, at any step, the Algorithm BS* will produce two clusters such that for each $\widehat{P}_j$, observations from this distribution will be contained in exactly one of the two clusters (see Lemma 3 in the Supplementary Material). Consequently, any cluster containing observations from a single distribution will no longer be split, while clusters containing observations from more than one population will continue to be split so long as each of the final clusters contain observations from a single distribution.
 \par

An observation from $\widehat{P}_1$ will be transferred to $C_1$ in the first iteration if and only if the following condition holds\\
 $(A4)$ $(n^2-n-nn_1+1)\sum_{i=1}^{K}n_{i}d(\widehat{P}_1,\widehat{P}_i) > (n^2-n-nn_j+1)\sum_{i=1}^{K}n_{i}d(\widehat{P}_j,\widehat{P}_i)$ for all $j=2,3,\dots,K$.
 \\
Under the condition $(A4)$, it can be ensured that till the $n_1$-th iteration, observations only from $\widehat{P}_1$ will be transferred to $C_1$ at each iteration, and the $d_w^*$ curve will keep on increasing. In the balanced case (when $n_1=n_2=\cdots=n_K$), the condition $(A4)$ reduces to requiring $\sum_{i=1}^{K}d(\widehat{P}_1,\widehat{P}_i) > \sum_{i=1}^{K}d(\widehat{P}_l,\widehat{P}_i)$ for all $l=2,3,\dots,K$.  For $l=1,2,\dots,K-2$, after transferring all observations of $\widehat{P}_1, \widehat{P}_2, \dots, \widehat{P}_{l}$ to $C_1$ (without any loss of generality) till iteration $n_1 + n_2 + \cdots + n_l$, an observation from $\widehat{P}_{l+1}$ will be transferred to $C_1$ in the next iteration if and only if the following condition holds:
\\
$(A5)$  $n(\sum_{i=l+1}^{K}n_i-1)\sum_{i=1}^{l}n_i\{d(\widehat{P}_i,\widehat{P}_j)-d(\widehat{P}_i,\widehat{P}_{l+1})\}+(\sum_{i=1}^{l}n_i+1)^{2}\sum_{i=l+1}^{K}n_{i}\{(n_j-1)d(\widehat{P}_i,\widehat{P}_j)-(n_{l+1}-1)d(\widehat{P}_i,\widehat{P}_{l+1})\}+(\sum_{i=1}^{l}n_{i}+1)(\sum_{i=l+1}^{K}n_{i}-1)\sum_{i=l+1}^{K}n_{i}\{d(\widehat{P}_i,\widehat{P}_{l+1})-d(\widehat{P}_i,\widehat{P}_j)\} > 0$ for all $j=l+2,l+3,\dots,K$.
\\
Under the same assumptions, it can also be shown that observations only from $\widehat{P}_{l+1}$ will be transferred to $C_1$ at every iteration until the $(n_1 + n_2 + \cdots + n_{l+1})$-th iteration. In the balanced case, using $\sum_{i=l+1}^{K}n_i -1 \approx \sum_{i=l+1}^{K}n_i$ and $\sum_{i=1}^{l}n_i+1 \approx \sum_{i=1}^{l}n_i$ in $(A5)$, the condition $(A5)$ reduces to the following condition:\\ $(nl+K)(nl+K-K^{2})\sum_{i=l+1}^{K}\{d(\widehat{P}_i,\widehat{P}_{l+1})-d(\widehat{P}_i,\widehat{P}_j)\}+K^{2}\{n(K-l)-K\}\sum_{i=1}^{l}\{d(\widehat{P}_i,\widehat{P}_j)-d(\widehat{P}_i,\widehat{P}_{l+1})\}>0$ for all $j=l+2,l+3,\dots,K$.\\ 
The above condition is natural since an observation from $\widehat{P}_{l+1}$ will be transferred to $C_1$ in the $(n_1 + n_2 + \cdots + n_l + 1)$-th iteration provided that it is closer to the mixture of $\widehat{P}_1,\widehat{P}_2,\dots,\widehat{P}_l$ than the mixture of $\widehat{P}_{l+2},\dots,\widehat{P}_{K}$.\\
One can explicitly derive the form of the $d_w^*$ curve over all iterations, namely,
 \begin{align}
	&d_w^*(C_{1}^{(r)},C_{2}^{(r)}) \nonumber \\
	&= \left \{    
	\begin{array}{l@{\quad} l }
		\frac{r(n-n_1)^2}{n(n-r)}\{\frac{1}{n-n_1}\sum_{i=2}^{K}n_{i}d(\widehat{P}_1,\widehat{P}_i)-\frac{1}{2(n-n_1)^2}\sum_{i,j=2}^{K}n_{i}n_{j}d(\widehat{P}_i,\widehat{P}_j)\} ; \quad \text{if $r = 1, 2, \dots, n_1$},\\\\
		\frac{n\left(\sum_{i=1}^{l}n_i \right)-r(n-n_{l+1})}{n}\left\{\frac{\sum_{i=1}^{l}n_i}{r}d\left(\frac{\sum_{i=1}^{l}n_{i}\widehat{P}_i}{\sum_{i=1}^{l}n_i},\widehat{P}_{l+1} \right) - \frac{\sum_{i=l+2}^{K}n_i}{n-r}d\left(\widehat{P}_{l+1},\frac{\sum_{i=l+2}^{K}n_{i}\widehat{P}_i}{\sum_{i=l+2}^{K}n_i}\right)\right\}\\+\frac{\left(\sum_{i=1}^{l}n_i \right)\left(\sum_{i=l+2}^{K}n_i \right)}{n}d\left(\frac{\sum_{i=1}^{l}n_{i}\widehat{P}_i}{\sum_{i=1}^{l}n_i},\frac{\sum_{i=l+2}^{K}n_{i}\widehat{P}_i}{\sum_{i=l+2}^{K}n_i}\right) ; \quad\quad  \text{if $r = \sum\limits_{i=1}^{l}n_i + 1, \dots, \sum_{i=1}^{l+1}n_i$ for $l=1,2,\dots,K-2$},\\\\     \label{dwcurve-morethanthree}
		\frac{(n-r)}{rn}\{(n-n_K)\sum_{i=1}^{K-1}n_{i}d(\widehat{P}_{i},\widehat{P}_{K})-\frac{1}{2}\sum_{i,j=1}^{K-1}n_{i}n_{j}d(\widehat{P}_i,\widehat{P}_j)\} ; \quad \text{if $r = \sum_{i=1}^{K-1}n_i + 1, \dots, n$}.
	\end{array}
	\right. 
\end{align}
As mentioned earlier, the above formula reveals that the $d_w^*$ curve keeps on increasing till the $n_1$-th iteration and keeps on decreasing after the $(\sum_{i=1}^{K-1}n_i)$-th iteration. Further, it can be shown that for each $l=1,2,\dots,K-2$, the $d_w^*$ curve will be a convex function of iteration number $r$ between the $(\sum_{i=1}^{l-1}n_i)$-th and the $(\sum_{i=1}^{l}n_i)$-th iterations.
 
For $l=1,2,\dots,K-2$, one possible scenario is that the $d_w^*$ curve is increasing between the $(\sum_{i=1}^{l}n_{i})$-th and  the $(\sum_{i=1}^{l+1}n_{i})$-th iterations. This holds if and only if the following condition is true:
\begin{align} \label{eq:38}
    \frac{\sum_{i=1}^{l}n_i}{\sum_{i=1}^{l}n_i+1}d\left(\frac{\sum_{i=1}^{l}n_{i}\widehat{P}_{i}}{\sum_{i=1}^{l}n_i},\widehat{P}_{l+1}\right) < \frac{(\sum_{i=l+2}^{K}n_j)^2}{(\sum_{i=l+1}^{K}n_j)(\sum_{i=l+1}^{K}n_j-1)}d\left(\widehat{P}_{l+1},\frac{\sum_{j=l+2}^{K}n_{j}\widehat{P}_{j}}{\sum_{i=l+2}^{K}n_j}\right).
\end{align}
Another possible scenario is when $d_w^*$ curve is linear and decreasing between  the $(\sum_{i=1}^{l}n_{i})$-th and  the $(\sum_{i=1}^{l+1}n_{i})$-th iterations, and this holds if and only if the following condition is true:
\begin{align} \label{eq:39}
    \frac{\sum_{i=l+2}^{K}n_j}{\sum_{i=l+2}^{K}n_j+1}d\left(\widehat{P}_{l+1},\frac{\sum_{j=l+2}^{K}n_{j}\widehat{P}_{j}}{\sum_{i=l+2}^{K}n_j}\right)<\frac{(\sum_{i=1}^{l}n_i)^2}{(\sum_{i=1}^{l+1}n_i)(\sum_{i=1}^{l+1}n_i-1)}d\left(\frac{\sum_{i=1}^{l}n_{i}\widehat{P}_{i}}{\sum_{i=1}^{l}n_i},\widehat{P}_{l+1}\right).
\end{align}
If neither \eqref{eq:38} nor \eqref{eq:39} hold, the $d_w^*$ curve will be a U-shaped curve between  the $(\sum_{i=1}^{l}n_{i})$-th and the $(\sum_{i=1}^{l+1}n_{i})$-th iterations. Hence, the $d_w^*$ curve attains maximum at iteration $r=\sum_{i=1}^{l}n_{i}$ for some $l \in \{1,2,\dots,K-1\}$.

\begin{remark} \label{remark-nmax*}
Now, we discuss the possible values of $N^*_{max}$ for difference choices of $K$. \\
\underline{Case 1 ($K = 2$)}: It follows from the discussion of the behaviour of Algorithm BS* for $K = 2$ that $N^*_{max} = \min\{n_1,n_2\}$. \\
\underline{Case 2 ($K = 3$)}: If we assume that an observation from $\widehat{P}_1$ is transferred to $C_1$ in the first iteration, and an observation from $\widehat{P}_2$ is transferred to $C_1$ after the $n_1$-th iteration, then it follows from \eqref{dwcurve-three} that $N^*_{max} \geq n_1$ and $N^*_{max} \leq n_1+n_2$.  \\
\underline{Case 3 ($K > 3$)}: If we assume that an observation from $\widehat{P}_1$ is transferred to $C_1$ in the first iteration, and an observation from $\widehat{P}_2$ is transferred to $C_1$ after the $n_1$-th iteration, then it follows from \eqref{dwcurve-morethanthree} that $N^*_{max} \geq n_1$ and $N^*_{max} \leq \sum_{i=1}^{K-1}n_i$.
\end{remark}

\indent Now, let us define the perfectness of a clustering algorithm and then make the formal statement about the perfect clustering property of the Algorithm CURBS-I*.
\begin{definition}
	\label{def:perfect}
	For any fixed $J$, let $C_{1}, C_2, \ldots, C_{J}$ be $J$ clusters obtained using a clustering algorithm on ${\cal D} = \cup_{j=1}^{K}{\cal D}_{j}$, which consists of observations from $K$ populations with ${\cal D}_j$ denoting the set of observations from the $j$-th population. We call the algorithm Perfect if $J = K$ and there exists a permutation $\pi$ of $\{1, 2, \ldots, K\}$ such that $C_{j} = \mathcal{D}_{\pi(j)}$ for each $j = 1,2,\ldots,K$.
\end{definition}
\begin{theorem} \label{curbs1_perfect}
      Consider a dataset $\cal D$ of $n$ observations of sizes $n_1, n_2, \ldots, n_K$ from $K > 1$ populations. Assume that $\liminf\limits_{n\rightarrow\infty} (n_{i}/n) >0$ for all $i=1,2,\dots,K$. Then there exists $n^* \geq 1$ such that for any $n\geq n^*$, the Algorithm CURBS-I* is Perfect.
\end{theorem}

\subsection{Oracle version of Algorithm CURBS-II}
The oracle version of the Algorithm CURBS-II, referred to as the Algorithm CURBS-II*, that uses the Algorithm BS* repeatedly along with some merging steps, are detailed below.  \\\\
\noindent
\small{
	\textbf{Algorithm CURBS-II*}
	\smallskip
	\hrule
	\smallskip
	\noindent \textbf{Input}: Dataset ${\cal D}$ and specified number $J ~(> 1)$ of clusters
	\begin{enumerate}
		\item Apply Algorithm BS* to ${\cal D}$.
		\item If $J=2$ then END, else set $i=2$.
		\item Apply Algorithm BS* to each of the $i$ clusters obtained.
		\item Compute $d_w^*$-distances between all the $2i$ clusters obtained in Step 3 and arrange them in the increasing order of magnitude.
		\item Merge those $(i-1)$ pair(s) of clusters obtained in Step 3 having the least $d_w^*$-distances. 
		\item  If $i=J-1$ END, else update $i=i+1$ and go to Step 3.
	\end{enumerate}
	\textbf{Output}: The clusters $C_1,C_2,\ldots,C_{J}$.
}
\smallskip
\hrule

\smallskip

\vspace{0.3cm}
Algorithm CURBS-II* has an order preserving property in the oracle scenario. The concept of Perfect and Order Preserving (POP) clustering algorithm was introduced in \citep{sarkar2019perfect}. 
  \begin{definition} \label{pop}
      For any fixed $J$, let $G_1^J$, $\dots$, $G_J^J$ be $J$ clusters estimated using a clustering algorithm on $\mathcal{D} = \cup_{i=1}^{K}\mathcal{D}_{i}$, which consists of observations from $K$ populations. We call the algorithm Perfect and Order Preserving (POP) at $K$ if the following conditions hold.
      \begin{enumerate}[label=\alph*.]
      \item For $J = K$, the clustering algorithm is perfect.
      \item For any $J < K$ and for every $i = 1, 2, \ldots, K$, there exists a unique $l \leq J$ such that $\mathcal{D}_{i}$ $\subseteq$ $G_{l}^{J}$.
      \item For any $J > K$ and for every $l = 1, 2, \ldots, J$, there exists a unique $i \leq K$ such that $G_{l}^{J}$ $\subseteq$ $\mathcal{D}_{i}$.
      \end{enumerate}
  \end{definition}
  \begin{figure}[h]
  \centering
      \includegraphics[width = 0.7\textwidth, height = 0.4\textheight]{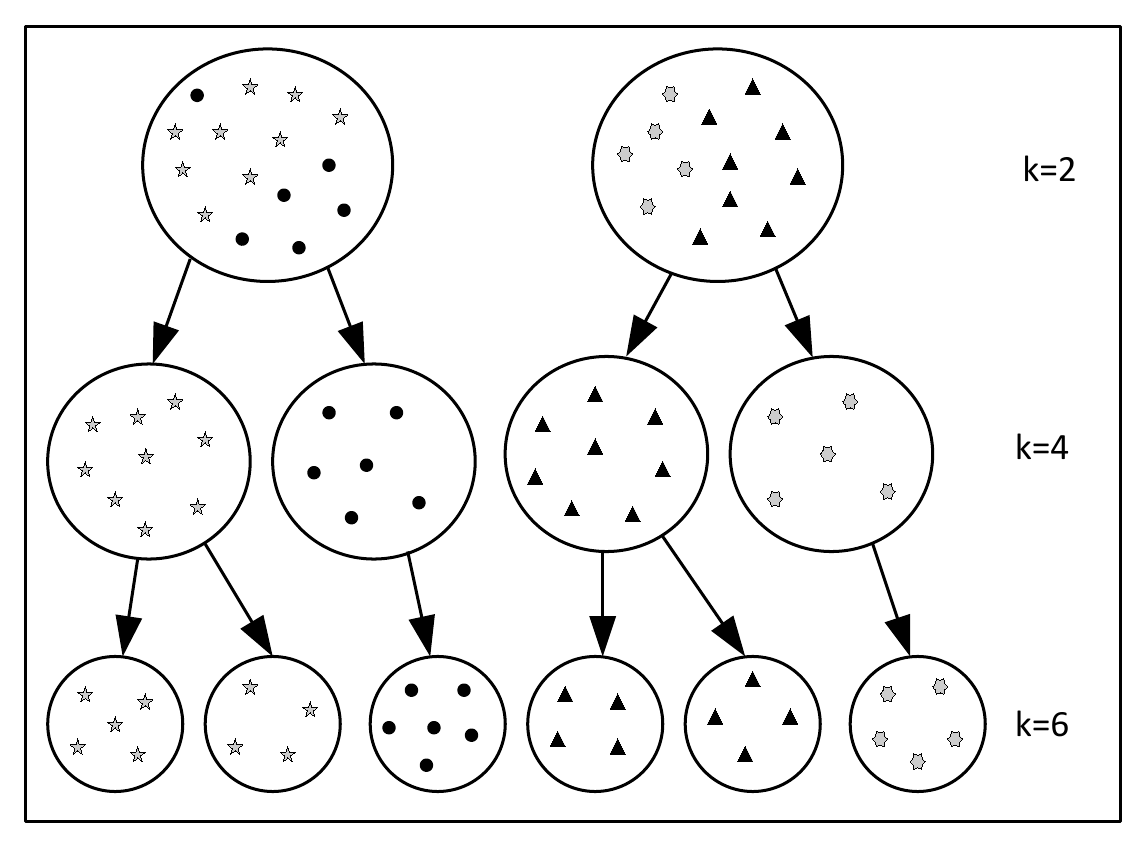}\\
  \caption[Short caption]{A clustering algorithm which is POP at 4.}
  \label{fig:POP}
  \end{figure}
  
  Figure \ref{fig:POP}, taken from \citep{sarkar2019perfect}, demonstrates a POP clustering (at $k:=J=4$) by taking only one value of $k$ smaller than 4 and one value of $k$ bigger than 4. Note that, (b) and (c) must hold for all $k < 4$ and for all $k > 4$, respectively.
  \par\noindent
\begin{theorem} \label{curbs2_pop}
    The Algorithm CURBS-II* applied on a dataset ${\cal D}$ containing observations from $K \:(\geq 2)$ populations is POP in the oracle scenario.
\end{theorem}
\noindent
This theorem states that if the specified cluster number $J$ is equal to the true number of populations $K$, then the Algorithm CURBS-II* provides perfect clustering. If the specified cluster number is less than the true number of populations, then the Algorithm CURBS-II* produces clusters with the property that no two of them have observations from any common distribution. If the specified cluster number is greater than the true number of populations, then each the produced clusters will contain observations from a single distribution.

\section{Numerical Results} \label{simulation}
\noindent
In this section, we investigate the behaviour of the Algorithms CURBS-I and CURBS-II. The latter algorithm is also compared with some of the state-of-the-art methods for functional data clustering, namely, (i) the \text{$\text{L}_{2}$PC} method mentioned in \cite{delaigle2019clustering} which is the standard $k$-means algorithm applied on the principal component scores, 
(ii) the \text{kCFC} method developed by \cite{chiou2007functional}, and (iii) the methods \text{$\text{DHP}_{\text{HA}}$}, \text{$\text{DHP}_{\text{DB}}$} and \text{$\text{DHP}_{\text{PC}}$} proposed by \cite{delaigle2019clustering} using an optimal projection of the data based on a suitable choice of orthonormal basis. Since \text{$\text{DHP}_{\text{HA}}$} gives better performance than the other two in most of the setups considered in the simulation, we only report the Rand index for this method only. The principal component scores required for the \text{$\text{L}_{2}$PC} method were obtained using the function `FPCA' in the R package \texttt{fdapace}. 

The \text{kCFC} method is available in the function `FClust' from the R package \texttt{fdapace}. The three methods \text{$\text{DHP}_{\text{HA}}$}, \text{$\text{DHP}_{\text{DB}}$} and \text{$\text{DHP}_{\text{PC}}$} mentioned above correspond to the Haar wavelet basis, the Daubechies wavelet basis and the standard principal component basis choices, respectively.

\subsection{Simulated Data}
\noindent For the simulation study, we generated functional data on a discrete grid of 128 equi-spaced points in the interval $[0,1]$ from the following two-population and three-population models. The observations from the $l$-th population are denoted by $X_{l}$ for $l =1,2,\ldots, K$ and $K \in \{2,3\}$.
  
\noindent \textbf{Model 1(a)} [Location change]:  $X_{l}(t) = \sum_{j=1}^{50}\sqrt{\theta_{jl}}Z_{jl}\phi_{j}(t) + \mu_{l}(t)$, $l = 1, 2$,
where for each $j$, $\phi_{j}(t) = \sqrt{2}\sin{(jt\pi)}$. $\theta_{j1} = \theta_{j2} = e^{-j/3}$ for $j = 1, 2, \dots, 50$, $\mu_1(t) = 2t$ and $\mu_2(t) = 6t(1-t)$. 
\\
\noindent \textbf{Model 1(b)} [Scale change]:   $X_{l}(t) = \sum_{j=1}^{50}\sqrt{\theta_{jl}}Z_{jl}\phi_{j}(t) + \mu_{l}(t)$, $l = 1, 2$,
where for each $j$, $\phi_{j}(t) = \sqrt{2}\sin{(jt\pi)}$. $\theta_{j1} = e^{-j/3}$, $\theta_{j2} = 3e^{-j/3}$ for $j = 1, 2, \dots, 50$, $\mu_1(t) = \mu_2(t) = 0$. 
\\
\noindent \textbf{Model 2(a)} [Location change]:    $X_{l}(t) = \sum_{j=1}^{40}(\sqrt{\theta_{jl}}Z_{jl} + \mu_{jl})\phi_{j}(t)$, $l = 1, 2$,
where for $j=1,2,\dots,40$, $\phi_{j}(t) = \sqrt{2}\sin{(jt\pi)}$, $\theta_{j1} = \theta_{j2} = j^{-2}$, $\mu_{j1} = 0$ and $\mu_{j2} = 0.75(-1)^{j+1} I\{1 \leq j \leq 3\}$. 
\\
\noindent \textbf{Model 2(b)} [Scale change]: $X_{l}(t) = \sum_{j=1}^{40}(\sqrt{\theta_{jl}}Z_{jl} + \mu_{jl})\phi_{j}(t)$, $l = 1, 2$,
where for $j=1,2,\dots,40$, $\phi_{j}(t) = \sqrt{2}\sin{(jt\pi)}$, $\theta_{j1} = j^{-2}$, $\theta_{j2} = 3j^{-2}$, $\mu_{j1} = \mu_{j2} = 0$. 
\\
\noindent \textbf{Model 3} [Change in eigenvalues]:   $X_{l}(t) = \sum_{j=1}^{50}\sqrt{\theta_{jl}}Z_{jl}\phi_{j}(t) + \mu_{l}(t)$, $l = 1, 2$,
where for each $j$, $\phi_{j}(t) = \sqrt{2}\sin{(jt\pi)}$. $\theta_{j1} = j^{-2}$, $\theta_{j2} = e^{-j}$ for $j = 1, 2, \dots, 50$, $\mu_1(t) = \mu_2(t) = 0$. 
\\
\noindent \textbf{Model 4} [Change in eigenfunctions]: $X_{l}(t) = \sum_{j=1}^{40}(\sqrt{\theta_{jl}}Z_{jl} + \mu_{jl})\phi_{jl}(t)$, $l = 1, 2$,
where for $j=1,2,\dots,40$, $\phi_{j1}(t) = \sqrt{2}\sin{(jt\pi)}$, $\phi_{j2}$ is the $j$-th term of Fourier Basis, $\theta_{j1} = \theta_{j2} = j^{-2}$, $\mu_{j1} = \mu_{j2} = 0$. 
\\
\noindent \textbf{Model 5(a)} [Location change]:    
$X_{l}(t) = \sum_{j=1}^{40}(\sqrt{\theta_{jl}}Z_{jl} + \mu_{jl})\phi_{j}(t)$, $l = 1, 2, 3$,
where for $j=1,2,\dots,40$, $\phi_{j}(t) = \sqrt{2}\sin{(jt\pi)}$, $\theta_{j1} = \theta_{j2} = \theta_{j3} = e^{-j/3}$.   
$\mu_{j1} = 0$ for $j = 1, 2, \dots, 40,  (\mu_{12}, \mu_{22}, \mu_{32}, \mu_{42}, \mu_{52}, \mu_{62}) = (0, -0.45, 0.45, -0.09, 0.84, 0.60)$, ($\mu_{13}, \mu_{23}, \mu_{33}, \mu_{43}$, $\mu_{53}, \mu_{63})
= (0, -0.30, 0.60, -0.30, 0.60, -0.30)$ and $\mu_{jl} = 0$ for $j > 6$, $l=2,3$.
\\
\noindent \textbf{Model 5(b)} [Scale change]: 
$X_{l}(t) = \sum_{j=1}^{40}(\sqrt{\theta_{jl}}Z_{jl} + \mu_{jl})\phi_{j}(t), l = 1, 2, 3$,
where for each $j=1,2,\dots,40$, $\phi_{j}(t) = \sqrt{2}\sin{(jt\pi)}$, $\theta_{j1} = e^{-j/3}, \theta_{j2} = 3e^{-j/3}, \theta_{j3} = 15e^{-j/3}$ for $l = 1, 2, 3$.
\\
\noindent  \textbf{Model 6(a)} [Location change]:    $X_{l}(t) = \sum_{j=1}^{50}\sqrt{\theta_{jl}}Z_{jl}\phi_{j}(t) + \mu_{l}(t)$, where $\phi_{j}(t) = \sqrt{2}\sin{(jt\pi)}$, $\theta_{j1} = \theta_{j2} = \theta_{j3} = j^{-1.01}$ for $j = 1, 2, \dots, 50$, $\mu_1(t) = 20t^{1.5}(1-t)$, $\mu_2(t) = 20t(1-t)^{1.5}$ and $\mu_3(t) = 0$, respectively.
\\
\noindent \textbf{Model 6(b)} [Scale change]:    $X_{l}(t) = \sum_{j=1}^{50}\sqrt{\theta_{jl}}Z_{jl}\phi_{j}(t) + \mu(t)$, where $\phi_{j}(t) = \sqrt{2}\sin{(jt\pi)}$, $\theta_{j1} = j^{-1.01}$, $\theta_{j2} = 3j^{-1.01}$, $\theta_{j3} = 9j^{-1.01}$ for $j = 1, 2, \dots, 50$ and $\mu(t) = 0$.
\\
\noindent \textbf{Model 7(a)} [Location change]: 
$X_{l}(t) = \sum_{j=1}^{40}(\sqrt{\theta_{jl}}Z_{jl} + \mu_{jl})\phi_{j}(t), l = 1, 2, 3$,
where for $j=1,2,\dots,40$, $\phi_{j}(t) = \sqrt{2}\sin{(jt\pi)}$, $\theta_{j1} = \theta_{j2} = \theta_{j3} =  j^{-2}$, $\mu_{j1} = 0$, $\mu_{j2} = 0.5(-1)^{j+1} I\{1 \leq j \leq 3\}$, $\mu_{j3} = (-1)^{j+1} I\{1 \leq j \leq 3\}$.
\\
\noindent \textbf{Model 7(b)} [Scale change]: $X_{l}(t) = \sum_{j=1}^{40}\sqrt{\theta_{jl}}Z_{jl} + \mu_{jl}\phi_{j}(t), l = 1, 2, 3$,
where for $j=1,2,\dots,40$, $\phi_{j}(t) = \sqrt{2}\sin{(jt\pi)}$, $\theta_{j1} = j^{-2}, \theta_{j2} = 3j^{-2}, \theta_{j3} =  9j^{-2}$.
\\
\noindent \textbf{Model 8} [Change in eigenvalues]: 
$X_{l}(t) = \sum_{j=1}^{40}(\sqrt{\theta_{jl}}Z_{jl} + \mu_{jl})\phi_{j}(t), l = 1, 2, 3$,
where for each $j=1,2,\dots,40$, $\phi_{j}(t) = \sqrt{2}\sin{(jt\pi)}$, $\theta_{j1} = j^{-2}, \theta_{j2} = j^{-1.05}, \theta_{j3} = e^{-j}$.
\\
In all of the above models except for models 7(a) and 7(b), $Z_{jk}$ $\overset{\mathrm{i.i.d}}{\sim}$ $N(0,1)$, and in models 7(a) and 7(b), $Z_{jk} \overset{\mathrm{i.i.d}}{\sim} t_{3}/\sqrt{3}$. 
We use the Laplacian kernel $k(x,y) = \exp\{-||x-y||/h\}$ in the $d_w$ measure, where the tuning parameter $h$ is chosen heuristically as $h = \mbox{Median}\{||a-b|| : a,b \in \mathcal{D}, a \neq b\}$ (see Section 5 of \cite{NIPS2009_685ac8ca}). Here, $\mathcal{D}$ denotes the set of all observations and $||\cdot||$ denotes the usual norm on the space $L^2[0,1]$. We have also experimented with $h = 0.1,0.2,1,5,10$. However, these choices did not change the results significantly, and hence, we do not report them. 
     \begin{table}
	\begin{center}
    \caption{Performance of the Algorithm CURBS-I}
	\label{tab1}
       \resizebox{0.7\textwidth}{0.5\textheight}{
	\begin{tabular}
    {|P{0.2cm}|P{3cm}|P{1cm}|P{1.5cm}|P{1.5cm}|P{1.5cm}|P{1cm}|}
		\hline
{$\mathbf{K}$} &{\textbf{Cluster Sizes}} &{\textbf{Model}} &$\widehat{K}<K$ &  $\widehat{K}=K$ &$\widehat{K}>K$&{\textbf{RI}}\\
		\hline
		\multirow{24}{*}{2}& \multirow{6}{*}{$(30, 30)$}& 1(a)& 0\%& 100\%& 0\%& 0.98\\
		\cline{3-7}
		& & 1(b)& 0\%& 100\%& 0\%& 0.997\\
		\cline{3-7}
		& & 2(a)& 0\%& 100\%& 0\%& 1\\
		\cline{3-7}
		& & 2(b)& 0\%& 100\%& 0\%& 0.996\\
        \cline{3-7}
                & & 3& 71\%& 29\%& 0\%& 0.639\\
                \cline{3-7}
        & & 4& 100\%& 0\%& 0\%& 0.492\\
		\cline{2-7}
		& \multirow{6}{*}{$(20,40)$}& 1(a)& 1\%& 98\%& 1\%& 0.992\\
		\cline{3-7}
		& & 1(b)& 0\%& 96\%& 4\%& 0.979\\
		\cline{3-7}
		& & 2(a)& 0\%& 100\%& 0\%& 0.999\\
		\cline{3-7}
		& & 2(b)& 0\%& 99\%& 1\%& 0.999\\
        \cline{3-7}
                & & 3& 1\%& 96\%& 3\%& 0.976\\
                \cline{3-7}
        & & 4& 84\%& 16\%& 0\%& 0.614\\
		\cline{2-7}
		& \multirow{6}{*}{$(100, 100)$}& 1(a)& 0\%& 100\%& 0\%& 0.995\\
		\cline{3-7}
		& & 1(b)& 0\%& 100\%& 0\%& 1\\
		\cline{3-7}
		& & 2(a)& 0\%& 100\%& 0\%& 1\\
		\cline{3-7}
		& & 2(b)& 0\%& 100\%& 0\%& 0.999\\
        \cline{3-7}
                & & 3& 0\%& 100\%& 0\%& 0.999\\
                \cline{3-7}
        & & 4& 29\%& 71\%& 0\%& 0.847\\
		\cline{2-7}
		& \multirow{6}{*}{$(50, 150)$}& 1(a)& 2\%& 97\%& 1\%& 0.962\\
		\cline{3-7}
		& & 1(b)& 0\%& 100\%& 0\%& 1\\
		\cline{3-7}
		& & 2(a)& 0\%& 100\%& 0\%& 0.999\\
		\cline{3-7}
		& & 2(b)& 0\%& 99\%& 1\%& 0.999\\
        \cline{3-7}
                & & 3& 0\%& 100\%& 0\%& 0.997\\
                \cline{3-7}
        & & 4& 0\%& 100\%& 0\%& 0.952\\
		\hline
		\multirow{28}{*}{3}& \multirow{7}{*}{$(30,30,30)$}& 5(a)& 0\%& 100\%& 0\%& 0.998\\
		\cline{3-7}
		& & 5(b)& 0\%& 100\%& 0\%& 0.985\\
		\cline{3-7}
		& & 6(a)& 0\%& 100\%& 0\%& 0.991\\
		\cline{3-7}
		& & 6(b)& 0\%& 100\%& 0\%& 0.998\\
		\cline{3-7}
		& & 7(a)& 0\%& 98\%& 2\%& 0.972\\
		\cline{3-7}
		& & 7(b)& 1\%& 97\%& 2\%& 0.969\\
        \cline{3-7}
		& & 8& 73\%& 26\%& 1\%& 0.827\\
		\cline{2-7}
		& \multirow{7}{*}{$(20, 30, 40)$}& 5(a)& 0\%& 100\%& 0\%& 0.998\\
		\cline{3-7}
		& & 5(b)& 0\%& 100\%& 0\%& 0.978\\
		\cline{3-7}
		& & 6(a)& 0\%& 98\%& 2\%& 0.993\\
		\cline{3-7}
		& & 6(b)& 0\%& 100\%& 0\%& 0.978\\
		\cline{3-7}
		& & 7(a)& 4\%& 90\%& 6\%& 0.945\\
		\cline{3-7}
		& & 7(b)& 2\%& 95\%& 5\%& 0.963\\
        \cline{3-7}
		& & 8& 100\%& 0\%& 0\%& 0.844\\
		\cline{2-7}
		& \multirow{7}{*}{$(100, 100, 100)$}& 5(a)& 0\%& 100\%& 0\%& 0.998\\
		\cline{3-7}
		& & 5(b)& 0\%& 100\%& 0\%& 0.991\\
		\cline{3-7}
		& & 6(a)& 0\%& 100\%& 0\%& 0.997\\
		\cline{3-7}
		& & 6(b)& 0\%& 100\%& 0\%& 0.999\\
		\cline{3-7}
		& & 7(a)& 0\%& 100\%& 0\%& 0.974\\
		\cline{3-7}
		& & 7(b)& 0\%& 100\%& 0\%& 0.98\\
        \cline{3-7}
		& & 8& 0\%& 99\%& 1\%& 0.981\\
		\cline{2-7}
		& \multirow{7}{*}{$(150, 25, 25)$}& 5(a)& 2\%& 98\%& 0\%& 0.967\\
		\cline{3-7}
		& & 5(b)& 0\%& 100\%& 0\%& 0.996\\
		\cline{3-7}
		& & 6(a)& 0\%& 98\%& 2\%& 0.993\\
		\cline{3-7}
		& & 6(b)& 0\%& 100\%& 0\%& 0.978\\
		\cline{3-7}
		& & 7(a)& 0\%& 96\%& 4\%& 0.949\\
		\cline{3-7}
		& & 7(b)& 1\%& 98\%& 1\%& 0.989\\
        \cline{3-7}
		& & 8& 100\%& 0\%& 0\%& 0.592\\
		\hline
	\end{tabular}
    }
	\vspace{0.5cm}
    \end{center}
\end{table}
To assess the effectiveness of clustering methods, we compute the average Rand index \cite{rand1971} by applying the clustering procedure to 100 independent datasets for each of the above models. The Rand index measures the agreement between two partitions of the data and takes values between $0$ and $1$ with the latter value representing perfect association between the two partitions. To compute Rand index, we use `RRand' function in R package \texttt{Phyclust}. 
\\
\noindent Table \ref{tab1} summarizes the performance of the Algorithm CURBS-I for the above models. The results show that our method has estimated the true number of populations almost accurately in all the cases except Models 3,4 and 8 with small sample size. These models are very difficult to investigate and are not mentioned in any clustering paper to the best of our knowledge. It is not surprising at all as we have already seen that the Algorithm SCC* gives accurate decision provided there are sufficient number of observations from each of the populations (See Theorem \ref{lma1}). This problem gets solved when there are sufficient number of observations from each of the populations producing the Rand index very close to $1$. This indicates that our proposed method has identified the true clusters nearly perfectly. We also compared this method with the algorithm developed by \cite{funfem} (implementation available in the R package \texttt{FunFEM}). However, the performance of the latter procedure was quite poor in all the models considered here, and thus, was not worth reporting. 

\par
We have also investigated the scenario when observations are coming from a single population. Samples of size $n$ are drawn from the following models.\\
\noindent \textbf{Model N1}:   $X(t) = \sum_{j=1}^{50}\sqrt{\theta_{j}}Z_{j}\phi_{j}(t)$,
where $\phi_{j}(t) = \sqrt{2}\sin{(jt\pi)}$, $\theta_{j} = e^{-j/3}$, for $j = 1, 2, \dots, 50$.
\\
\noindent \textbf{Model N2}: $X(t) = \sum_{j=1}^{40}\sqrt{\theta_{j}}Z_{j} \phi_{j}(t)$, where $\phi_{j}(t) = \sqrt{2}\sin{(jt\pi)}$, $\theta_{j} = j^{-2}$ for $j=1,2,\dots,40$. 
\\
\noindent \textbf{Model N3}:    $X(t) = \sum_{j=1}^{50}\sqrt{\theta_{jk}}Z_{jk}\phi_{j}(t)$, where $\phi_{j}(t) = \sqrt{2}\sin{(jt\pi)}$, $\theta_{j} = j^{-1.01}$ for $j = 1, 2, \dots, 50$.
\\
\noindent \textbf{Model N4}: $X(t) = \sum_{j=1}^{40}\sqrt{\theta_{j}}Z_{j}\phi_{j}(t)$,
where $\phi_{j}(t) = \sqrt{2}\sin{(jt\pi)}$, $\theta_{j} = j^{-2}$ for each $j=1,2,\dots,40$.
\noindent
\newline
\noindent
In the models N1, N2 and N3, $Z_{j}$ $\overset{\mathrm{i.i.d}}{\sim}$ $N(0,1)$, and in model N4, $Z_{j} \overset{\mathrm{i.i.d}}{\sim} t_{3}/\sqrt{3}$. In each of the above models, we generated functional data on a discrete grid of 128 equi-spaced points in the interval $[0,1]$. When observations come from a single population, evaluating the efficacy of the Algorithm CURBS-I is equivalent to assessing the effectiveness of the Algorithm SCC. Table \ref{size} reports the proportion of times the estimated cluster number ($\widehat{K}$) is equal to 1, i.e, the Algorithm SCC correctly identifies the observations as coming from a single population.   
     \begin{table}[!h]
	\centering    	\caption{Performance of the Algorithm CURBS-I when $K=1$}
	\label{size}
	\begin{tabular} {|c|c|c|c|c|}
		\hline
\multirow{2}{*}{\textbf{Sample Size}}  &\multicolumn{4}{c|}{\textbf{Proportion of times $\widehat{K}=1$}} \\
\cline{2-5}
		 &\textbf{Model N1} &\textbf{Model N2} &\textbf{Model N3} &\textbf{Model N4}\\
		\hline
		30& 94\%& 95\%& 96\%& 100\%\\
		\hline
		100& 95\%& 98\%& 100\%& 100\%\\
        \hline
	\end{tabular}
	\vspace{0.5cm}
\end{table}
The results in Table \ref{size} show that the Algorithm SCC makes the correct decision most of the time. 
\par
Next, we study the performance of the Algorithm CURBS-II along with some state-of-the-art clustering techniques in case the number of clusters is specified. The performance of these methods are summarized in Table \ref{tab2}. For each row in this table, the maximum Rand index obtained is boldfaced.  
    \begin{table}
    \begin{center}
	\caption{Rand indices of different clustering algorithms for specified cluster number}
	\label{tab2}	
       \resizebox{0.7\textwidth}{0.5\textheight}{
    \begin{tabular}{|P{0.2cm}|P{3cm}|P{1cm}|P{2cm}|P{1.5cm}|P{1.5cm}|P{1.5cm}|}
		\hline
		$\mathbf{K}$ &\textbf{Cluster Sizes} &\textbf{Model} &\textbf{CURBS-II} &\textbf{$\text{L}_2$PC} &\textbf{kCFC} &\textbf{$\text{DHP}_{\text{HA}}$}  \\
		\hline
		\multirow{24}{*}{2}& \multirow{6}{*}{$(30,30)$}& 1(a)& 0.972& $\mathbf{0.983}$& 0.79& 0.973\\
		\cline{3-7}
		& & 1(b)& $\mathbf{0.984}$& 0.53& 0.492& 0.492\\
		\cline{3-7}
		& & 2(a)& 0.988& 0.741& $\mathbf{1}$& $\mathbf{1}$\\
		\cline{3-7}
		& & 2(b)& $\mathbf{0.981}$& 0.512& 0.501& 0.492\\
        \cline{3-7}
        & & 3& $\mathbf{0.992}$& 0.53& 0.496& 0.52\\
        \cline{3-7}
        & & 4& $\mathbf{0.907}$& 0.522& 0.493& 0.5\\
		\cline{2-7}
		& \multirow{6}{*}{$(20,40)$}& 1(a)& 0.949& $\mathbf{0.953}$& 0.655& 0.873\\
		\cline{3-7}
		& & 1(b)& $\mathbf{0.984}$& 0.501& 0.497& 0.497\\
		\cline{3-7}
		& & 2(a)& $\mathbf{1}$& $\mathbf{1}$& $\mathbf{1}$& $\mathbf{1}$\\
		\cline{3-7}
		& & 2(b)& $\mathbf{0.98}$& 0.498& 0.512& 0.492\\
        \cline{3-7}
        & & 3& $\mathbf{0.977}$& 0.498& 0.494& 0.456\\
        \cline{3-7}
        & & 4& $\mathbf{0.777}$& 0.497& 0.492& 0.542\\
		\cline{2-7}
		& \multirow{6}{*}{$(100, 100)$}& 1(a)& 0.939& $\mathbf{0.945}$& 0.765& 0.961\\
		\cline{3-7}
		& & 1(b)& $\mathbf{0.995}$& 0.521& 0.503& 0.5\\
		\cline{3-7}
		& & 2(a)& $\mathbf{1}$& $\mathbf{1}$& $\mathbf{1}$& $\mathbf{1}$\\
		\cline{3-7}
		& & 2(b)& $\mathbf{0.994}$& 0.52& 0.509& 0.501\\
        \cline{3-7}
        & & 3& $\mathbf{0.964}$& 0.503& 0.496& 0.464\\
        \cline{3-7}
        & & 4& $\mathbf{0.985}$& 0.508& 0.492& 0.467\\
		\cline{2-7}
		& \multirow{6}{*}{$(50,150)$}& 1(a)& 0.903& $\mathbf{0.992}$& 0.563& 0.99\\
		\cline{3-7}
		& & 1(b)& $\mathbf{0.988}$& 0.499& 0.497& 0.499\\
		\cline{3-7}
		& & 2(a)& $\mathbf{1}$& $\mathbf{1}$& $\mathbf{1}$& $\mathbf{1}$\\
		\cline{3-7}
		& & 2(b)& $\mathbf{0.939}$& 0.5& 0.464& 0.501\\
        \cline{3-7}
        & & 3& $\mathbf{0.997}$& 0.5& 0.498& 0.521\\
        \cline{3-7}
        & & 4& $\mathbf{0.963}$& 0.5& 0.498& 0.497\\
		\hline
		\multirow{28}{*}{3}& \multirow{7}{*}{$(30, 30, 30)$}& 5(a)& 0.979& 0.973& $\mathbf{1}$& 0.675\\
		\cline{3-7}
		& & 5(b)& $\mathbf{0.982}$& 0.464& 0.557& 0.555\\
		\cline{3-7}
		& & 6(a)& $\mathbf{0.98}$& 0.896& 0.771& 0.779\\
		\cline{3-7}
		& & 6(b)& $\mathbf{0.978}$& 0.479& 0.548& 0.589\\
		\cline{3-7}
		& & 7(a)& $\mathbf{0.98}$& 0.947& $\mathbf{1}$& $\mathbf{1}$\\
		\cline{3-7}
		& & 7(b)& $\mathbf{0.908}$& 0.41& 0.641& 0.555\\
        \cline{3-7}
        & & 8& $\mathbf{0.974}$& 0.474& 0.503& 0.481\\
		\cline{2-7}
		& \multirow{7}{*}{$(20, 30, 40)$}& 5(a)& $\mathbf{0.974}$& 0.944& 0.754& 0.751\\
		\cline{3-7}
		& & 5(b)& $\mathbf{0.978}$& 0.45& 0.508& 0.547\\
		\cline{3-7}
		& & 6(a)& $\mathbf{0.995}$& 0.862& 0.695& 0.751\\
		\cline{3-7}
		& & 6(b)& $\mathbf{0.974}$& 0.484& 0.528& 0.559\\
		\cline{3-7}
		& & 7(a)& 0.935& 0.942& $\mathbf{1}$& 0.983\\
		\cline{3-7}
		& & 7(b)& $\mathbf{0.946}$& 0.411& 0.527& 0.541\\
        \cline{3-7}
        & & 8& $\mathbf{0.806}$& 0.468& 0.495& 0.492\\
		\cline{2-7}
		& \multirow{7}{*}{$(100,100,100)$}& 5(a)& 0.991& 0.98& $\mathbf{1}$& 0.781\\
		\cline{3-7}
		& & 5(b)& $\mathbf{0.986}$& 0.473& 0.564& 0.546\\
		\cline{3-7}
		& & 6(a)& 0.964& 0.922& 0.726& 0.78\\
		\cline{3-7}
		& & 6(b)& $\mathbf{0.994}$& 0.49& 0.557& 0.586\\
		\cline{3-7}
		& & 7(a)& 0.987& 0.948& 0.991& $\mathbf{0.996}$\\
		\cline{3-7}
		& & 7(b)& $\mathbf{0.955}$& 0.436& 0.554& 0.584\\
        \cline{3-7}
        & & 8& $\mathbf{0.98}$& 0.503& 0.498& 0.503\\
		\cline{2-7}
		& \multirow{7}{*}{$(150, 25, 25)$}& 5(a)& 0.99& 0.786& $\mathbf{1}$& 0.616\\
		\cline{3-7}
		& & 5(b)& $\mathbf{0.994}$& 0.697& 0.507& 0.713\\
		\cline{3-7}
		& & 6(a)& 0.963& $\mathbf{0.981}$& 0.969& 0.626\\
		\cline{3-7}
		& & 6(b)& $\mathbf{0.986}$& 0.579& 0.494& 0.547\\
		\cline{3-7}
		&  & 7(a)& 0.9& 0.922& $\mathbf{1}$& $\mathbf{1}$\\
		\cline{3-7}
		& & 7(b)& $\mathbf{0.913}$& 0.454& 0.469& 0.545\\
        \cline{3-7}
        & & 8& $\mathbf{0.95}$& 0.53& 0.497& 0.491\\
		\hline
	\end{tabular}
    }
    \end{center}
\end{table}
From Table \ref{tab2}, it is clear that in location difference problems, the Algorithm CURBS-II does exceedingly well and provides comparable performance compared to its competitors. The Algorithm CURBS-II is in fact the best performer in several of these models. On the other hand, in the scale difference problems, all the competing clustering algorithms has poor performance while the Algorithm CURBS-II performs extremely well. Since the standard $k$-means method is applied on the principal component scores in the $\text{L}_2{\text{PC}}$ method, it does not perform well in scale difference problems. The $\text{kCFC}$ method has a mixed performance even in the models having location difference, whereas it fails for the models having scale difference. The performance of the \text{DHP} methods (\text{$\text{DHP}_{\text{HA}}$}, \text{$\text{DHP}_{\text{DB}}$} and \text{$\text{DHP}_{\text{PC}}$}) are quite well in the location problems, but we cannot pick any single method which performs well across all models having location difference. The near-perfect clustering performance of the \text{DHP} methods claimed by \cite{delaigle2019clustering} for scale difference models require a sufficient condition that there is additional difference in the location as well. The scale difference models considered in this paper do not have any difference in the locations. Consequently, the poor performance of the \text{DHP} methods can be viewed as an empirical evidence that the additional location difference required for the near-perfect clustering of these methods are perhaps both necessary and sufficient. 

\subsection{Analysis of real datasets}
\indent We applied the proposed clustering algorithms and their competitors to two real datasets. The first dataset we considered was the Phoneme dataset which is available in the \texttt{fds} package in R. The dataset is from the domain of speech recognition and contains log-periodograms constructed from recordings available at different equispaced frequencies for each of the five phonemes: \enquote{sh} as in \enquote{she}, \enquote{dcl} as in \enquote{dark}, \enquote{iy} as in \enquote{she}, \enquote{aa} as in \enquote{dark} and \enquote{ao} as in \enquote{water}. For each phoneme, we have 400 log-periodograms at a 16-kHz rate observed over $150$ different frequencies. This dataset has been considered in a phoneme discrimination problem in \cite{hastie2009elements} and \cite{ferraty2003curves} where the aim was to predict the phoneme class for a new log-periodogram. Here, we reformulate the problem into a clustering problem with the aim to partition the phoneme data into different clusters. It is known that there are five classes in this dataset, but we assume that the class labels are missing. We first applied the Algorithm CURBS-II and its competitors with the specified number of clusters being $5$. We could not perform the $\text{DHP}_{\text{DB}}$ method since the code required the grid size to be equal to some power of $2$. Table \ref{tab3} summarizes the performance of the different clustering methods. 
\begin{table}[h!]
        \centering
        \begin{tabular}{|P{3cm}|P{2cm}|P{1.5cm}|P{1.5cm}|P{1.5cm}|P{1.5cm}|}
        \hline
        \textbf{Methods} &\textbf{CURBS-II}& \textbf{$\text{L}_2$PC} &\textbf{kCFC}  &\textbf{$\text{DHP}_{\text{HA}}$} &\textbf{$\text{DHP}_{\text{PC}}$}\\
        \hline
             \textbf{Rand Index}& $\mathbf{0.922}$& 0.918& 0.801& 0.715& 0.771\\
             \hline
        \end{tabular}
        
        \caption{Performance of different clustering algorithms in Phoneme data clustering}
        \label{tab3}
    \end{table}
The performance of the Algorithm CURBS-II is the best followed closely by $\text{L}_2$PC. The Algorithm CURBS-I applied to this dataset produced four clusters with the observed Rand index being $0.896$. The two groups corresponding to the phonemes \enquote{aa} and \enquote{ao} were difficult to distinguish which can be attributed to the fact that the square MMD between these two groups is significantly less compared to the other pairwise distances (see Table \ref{pairwise phoneme}).


    \begin{table}[h!]
        \centering
        \begin{tabular}
        {|c|c|c|c|c|c|}
        \hline
   & sh& dcl& iy& aa& ao\\
   \hline
        sh &    &  0.716&  0.216&  0.422& 0.426\\
        \hline
          dcl &   &  &  0.561&  0.658& 0.576\\
          \hline
          iy&   &  &  &  0.218& 0.211\\
          \hline
            aa& &  &  &  & 0.056\\
            \hline
           ao&  &  &  &  & \\
           \hline
        \end{tabular}
        \caption{Pairwise squared MMD distances for Phoneme data}
        \label{pairwise phoneme}
    \end{table}
    
\indent The second real dataset considered was the `Meatspectrum' data which is available in the  \texttt{R} package \texttt{faraway}. A Tecator Infratec Food and Feed Analyzer was used to collect this data which contains the fat content of $215$ finely chopped samples of meat. A Near Infrared Transmission (NIT) principle was used to measure the absorbance at 100 wavelengths within the infrared spectrum ranging from $850 - 1050$ nanometres (see Figure \ref{fig:tecator}). The classification problem in this dataset consists in separating meat samples with a high fat content (more than $20\%$) from samples with a low fat content (less than $20\%$). Among the $215$ samples, $77$ have high fat content and $138$ have low fat content. 
\begin{figure}[h!]
   \centering
   \vspace{-0.4cm}
\includegraphics[width=0.6\linewidth,height=0.35\textheight]{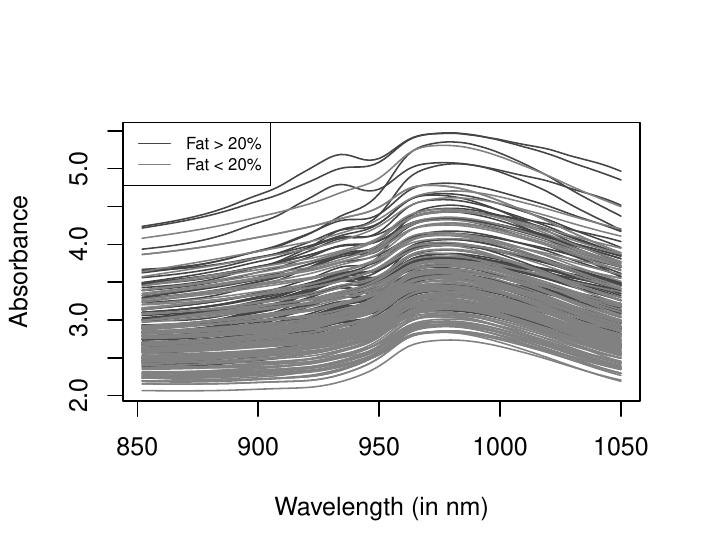}
\caption{Raw Spectra of Meatspectrum dataset}
  \label{fig:tecator}
  \end{figure}
Among others, \cite{ferraty2003curves}, \cite{rossi2006support} and \cite{li2008classification} have considered the original spectrum along with its derivatives for classification purposes. It is observed that the spectra of a meat sample with high fat content often have two local maxima while a single maxima is mostly observed in case of low fat content. Therefore, these authors decided to focus on the curvature of the spectra, and thus, used the second derivative for their analysis. We also considered the second order derivative curve for clustering purposes (see plot (a) of Figure \ref{fig:meat}). To find the second derivative, smoothing of the data was done using B-spline smoothing technique with 11 basis functions and the curves were evaluated at $256$ equidistant wavelengths ranging from 850-1050 nanometres. 
\begin{figure}[h]
   \centering
  \begin{subfigure}{.45\textwidth}
\includegraphics[width=1\linewidth,height=0.25\textheight]{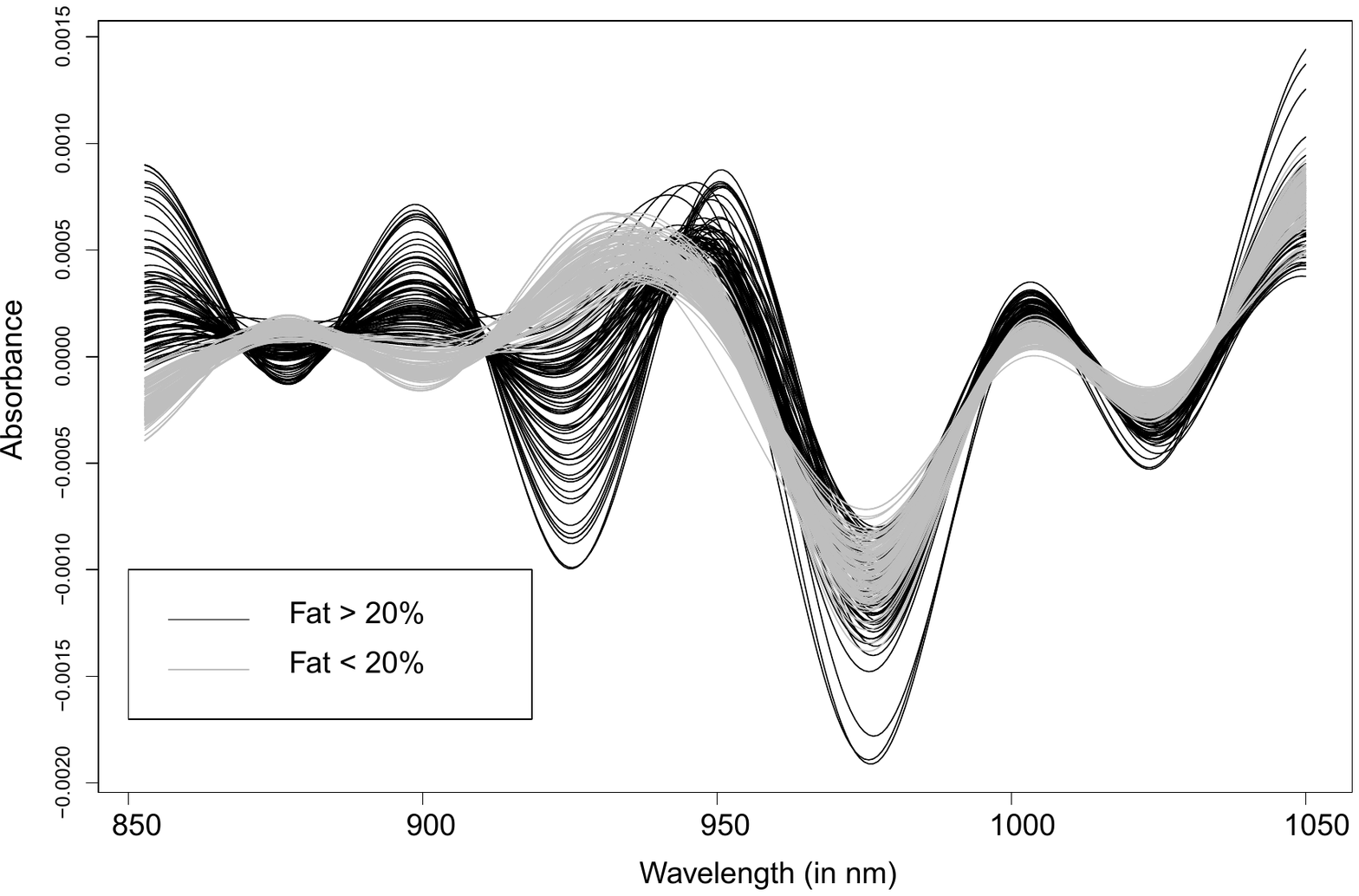}
\caption{Actual partition}
  \end{subfigure}
  \begin{subfigure}{.45\textwidth}
\includegraphics[width=1\linewidth,height=0.25\textheight]{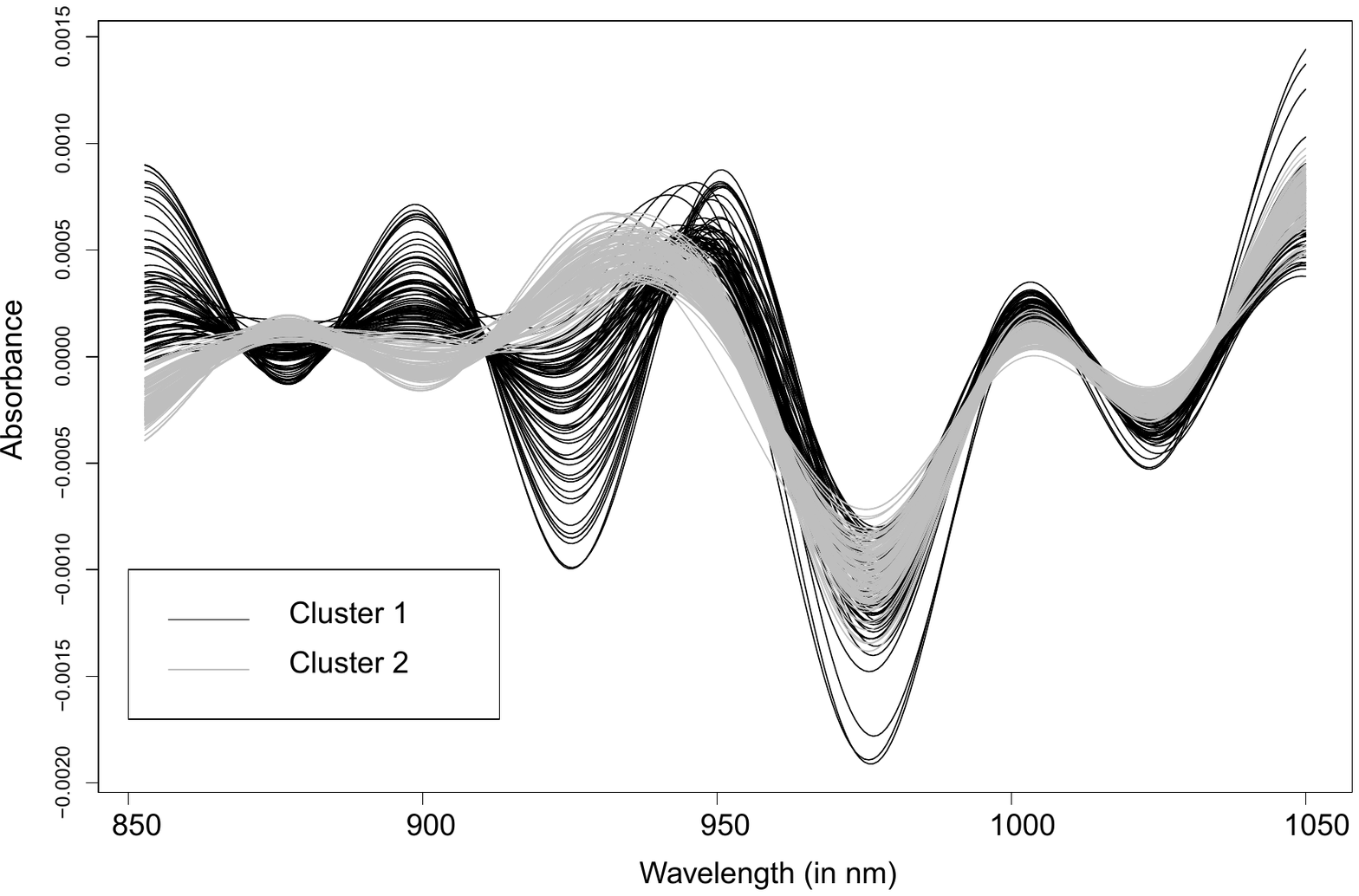}
  \caption{Partition obtained by the Algorithm CURBS-II}
  \end{subfigure}
  \caption{Second order derivative of Spectra for Meatspectrum dataset}
  \label{fig:meat}
  \end{figure}
The partition produced by the Algorithm CURBS-II is displayed in plot (b) of Figure \ref{fig:meat} which is quite close to the original partition.  Table \ref{tab:meat} shows the Rand indices of the Algorithm CURBS-II and its competitors for this dataset. 
\begin{table}[h!]
        \centering
        \begin{tabular}{|P{3cm}|P{2cm}|P{1.5cm}|P{1.5cm}|P{1.5cm}|P{1.5cm}|P{1.5cm}|}
        \hline
        \textbf{Methods} &\textbf{CURBS-II}& \textbf{$\text{L}_2$PC} &\textbf{kCFC}  &\textbf{$\text{DHP}_{\text{HA}}$} &\textbf{$\text{DHP}_{\text{DB}}$} &\textbf{$\text{DHP}_{\text{PC}}$}\\
        \hline
             \textbf{Rand Index}& $\mathbf{0.963}$& 0.862& 0.646& 0.759& 0.886 &0.886\\
             \hline
        \end{tabular}        
        \caption{Performance of different clustering algorithms in Meatspectrum data clustering}
        \label{tab:meat}
    \end{table}
The Algorithm CURBS-II turns out to be the clear winner among its competitors. The $\text{L}_2$PC, $\text{DHP}_{\text{DB}}$ and $\text{DHP}_{\text{PC}}$ methods performing quite similarly. The Algorithm CURBS-I applied on this dataset estimated the cluster number to be $2$, and the partition obtained by the Algorithm CURBS-I is same as the partition obtained by the Algorithm CURBS-II. Thus, the Algorithm CURBS-I produces the same observed Rand index $0.963$.

\section{Epilogue} \label{epilogue}

The Algorithm CURBS-I proposed in the paper provides a bonafide estimator of the number of clusters in addition to producing the individual clusters. Both Algorithms CURBS-I and CURBS-II are able to capture distributional differences (which follows from the oracle analysis). In particular, we do not know of any method which provides near-perfect clustering for scale problems, which is one of the main advantages of our methods. Although the numerical study in this paper is conducted for univariate functional data, both the proposed algorithms are directly applicable to multivariate functional data as well. Moreover, both the proposed clustering algorithms can be used for data in any general metric space. The theoretical analysis showing perfect clustering only requires the kernel to be characteristic in that metric space. If the functional data are observed at irregular time points, or are partially observed over the domain, then the proposed algorithms need to be modified. One possible approach would be to project the data onto an appropriate finite-dimensional Euclidean space using some suitable techniques such as random projections and then apply the proposed algorithms. We leave this as a future research project.
\section{Proofs and technical details}
All the proofs and technical details are given in the Supplementary Material.
\bibliographystyle{plainnat}
\bibliography{Ref}
\newpage


\section*{Supplementary Material} \label{supplementary}
The supplementary material contains the proofs and the technical details of the main paper.
\subsection*{Proof of Theorem 1}
 \textit{(a)} Since the observations come from more than one distributions, if we split the $d_w^*$ curve at its maximum value, we get two clusters which can be considered as a possible realization of size $N_{max}^*$ from distribution $P$ and of size ($n - N_{max}^*$) from distribution $Q$ for some distributions $P$ and $Q$ ($P$ and $Q$ may be single or mixture distributions). Therefore,
\begin{align*}
    \underset{r = 1,2,\dots,n-1}{\max}d_w^*(C_{1}^{(r)},C_{2}^{(r)}) = d_w^*(C_{1}^{(N_{max}^*)},C_{2}^{(N_{max}^*)}) = \frac{N_{max}^*(n-N_{max}^*)}{n}d(P,Q).
\end{align*}
\noindent
If $N_{min}^* < N_{max}^*$, $C_1^{(N_{min}^*)}$ can be considered as a possible realization of size $N_{min}^*$ from $P$ while $C_2^{(N_{min}^*)}$ can be considered as a possible realization of size ($n - N_{min}^*$) from ($N_{max}^* - N_{min}^*$)/($n - N_{min}^*$)$P$ + ($n - N_{max}^*$)/($n - N_{min}^*$)$Q$. Thus,
\begin{align*}
    \underset{r = 1,2,\dots,n-1}{\min}d_w^*(C_{1}^{(r)},C_{2}^{(r)}) &= d_w^*(\widetilde{P}_{C_{1}^{(N_{min}^*)}},\widetilde{P}_{C_{2}^{(N_{min}^*)}}) \\ &= \frac{N_{min}^*(n-N_{min}^*)}{n}d(P,\frac{N_{max}^* - N_{min}^*}{n - N_{min}^*}P+\frac{n - N_{max}^*}{n - N_{min}^*}Q) \\
    &= \frac{N_{min}^*(n-N_{min}^*)}{n}\frac{(n - N_{max}^*)^2}{(n - N_{min}^*)^2}d(P,Q)\\
    &= \frac{N_{min}^*(n - N_{max}^*)^2}{n(n - N_{min}^*)}d(P,Q).
\end{align*}
Therefore, in this case,
\begin{align*}
    V^* = \frac{\underset{r = 1,2,\dots,n-1}{\max}d_w^*(C_{1}^{(r)},C_{2}^{(r)})}{\underset{r = 1,2,\dots,n-1}{\min}d_w^*(C_{1}^{(r)},C_{2}^{(r)})} = \frac{N_{max}^*(n - N_{max}^*)}{n}\frac{n(n - N_{min}^*)}{N_{min}^*(n - N_{max}^*)^2}    = \frac{N^{*}_{max}(n-N^{*}_{min})}{N^{*}_{min}(n-N^{*}_{max})}.
\end{align*}
When $N_{min}^* > N_{max}^*$, $C_1^{(N_{min}^*)}$ can be considered as a possible realization of size $N_{min}^*$ from \{$N_{max}^*$/$N_{min}^*$\}$P$ + \{($N_{min}^* - N_{max}^*$)/$N_{min}^*$\}$Q$ while $C_2^{(N_{min}^*)}$ can be considered as a possible realization of size ($n - N_{min}^*$) from $Q$. 
﻿
Thus,
\begin{align*}
    \underset{r = 1,2,\dots,n-1}{\min}d_w^*(C_{1}^{(r)},C_{2}^{(r)}) &= d_w^*(C_{1}^{(N_{min}^*)},C_{2}^{(N_{min}^*)}) \\ &= \frac{N_{min}^*(n-N_{min}^*)}{n}d(\frac{N_{max}^*}{N_{min}^*}P + \frac{N_{min}^* - N_{max}^*}{N_{min}^*}Q,Q) \\
    &= \frac{N_{min}^*(n-N_{min}^*)}{n}\frac{{N_{max}^*}^2}{ {N_{min}^*}^2}d(P,Q)\\
    &= \frac{{N_{max}^*}^2(n - N_{min}^*)}{n N_{min}^*}d(P,Q).
\end{align*}
Therefore, in this case,
\begin{align*}
    V^* = \frac{\underset{r = 1,2,\dots,n-1}{\max}d_w^*(C_{1}^{(r)},C_{2}^{(r)})}{\underset{r = 1,2,\dots,n-1}{\min}d_w^*(C_{1}^{(r)},C_{2}^{(r)})} = \frac{N_{max}^*(n - N_{max}^*)}{n}\frac{n N_{min}^*}{(n - N_{min}^*){N_{max}^*}^2} = \frac{N^{*}_{min}(n-N^{*}_{max})}{N^{*}_{max}(n-N^{*}_{min})}.
\end{align*}
\noindent  Hence, $V^* = R^*$.\\ 
\textit{(b)} First observe that if we compare the $d_w^*$ curve at $r = 1$ and $r = n-1$, it follows that the value at $r = 1$ is larger than the value at $r = n-1$. This follows from the fact that, at both $r = 1$ and $r = n-1$, one of the subsets between $C_1$ and $C_2$ have $1$ observation while the other contains the rest. Furthermore, we maximize $d_w^*(C_1,C_2)$ at $r = 1$. Hence, the case $r = 1$ will not be considered when finding $N^*_{min}$. The remainder of the proof will be done separately for the cases $K = 2$, $K = 3$ and $K > 3$. We will assume without loss of any generality that till the $n_1$-th iteration, observations only from $\widehat{P}_1$ will be transferred to $C_1$. Further, between the $\left(\sum_{l=1}^{i-1} n_l+1\right)$-th and the $\left(\sum_{l=1}^i n_l\right)$-th iterations, observations only from $\widehat{P}_i$ will be transferred to $C_1$ for each $i=2,3,\ldots,K$ (cf. the discussion about the behaviour of $d_w^*$ curve for different values of $K$ provided after the description of the Algorithm CURBS-I*). \\
\underline{Case 1 ($K = 2$)}: It follows from equation (8) of the main paper that $N^*_{min} = n-1$. \\
\underline{Case 2 ($K = 3$)}: The $d_w^*$ curve is increasing between $r = 1$ to $r = n_1$, and it is decreasing between $r = n_1+n_2$ and $r = n-1$. Thus, we only need to compare the value at $r = n-1$ with the minimum value between $n_1+1$ and $n_1+n_2$. Note that the $d_w^*$ curve is convex between these iterations as can be seen from equation (9) of the main paper. \\
In case the $d_w^*$ curve is linear and increasing between the ($n_1+1$)-th and the ($n_1+n_2$)-th iterations (which would happen under condition (11) of the main paper), the minimum value between these iterations will be at $r = n_1+1$. This minimum value is also larger than the value at $r = n_1$ under the same condition (11) of the main paper. Thus, this minimum value is larger than the value at $r = n-1$. Hence, $N^*_{min} = n-1$. \\
In case the $d_w^*$ curve is linear and decreasing between the ($n_1+1$)-th and the ($n_1+n_2$)-th iterations (which would happen under condition (12) of the main paper), the minimum value between these iterations will be at $r = n_1+n_2$. Since the $d^*_w$ curve keeps on decreasing between the ($n_1+n_2$)-th and the ($n-1$)-th iterations, it follows that $N^*_{min} = n-1$. \\
In case the $d_w^*$ curve is U-shaped between the ($n_1+1$)-th and the ($n_1+n_2$)-th iterations (which would happen if neither condition (11) nor condition (12) hold), it can be shown that the minimum value of the $d_w^*$ curve between these iterations is larger than the value at $r = n-1$ if and only if
\begin{eqnarray*}
	n \geq \frac{u_n}{v_n} := \frac{(\alpha_{1n} + \alpha_{2n})\{\alpha_{1n}d(\widehat{P}_1,\widehat{P}_3) + \alpha_{2n}d(\widehat{P}_2,\widehat{P}_3)\} - \alpha_{1n}\alpha_{2n}d(\widehat{P}_1,\widehat{P}_2)}{(\alpha_{1n}+\gamma_{n})(1-\alpha_{1n}-\gamma_{n})d\left(\frac{\alpha_{1n}}{\alpha_{1n}+\gamma_{n}}\widehat{P}_{1}+\frac{\gamma_{n}}{\alpha_{1n}+\gamma_{n}}\widehat{P}_{2},\frac{\alpha_{2n}-\gamma_{n}}{1-\alpha_{1n}-\gamma_{n}}\widehat{P}_{2}+\frac{1-\alpha_{1n}-\alpha_{2n}}{1-\alpha_{1n}-\gamma_{n}}\widehat{P}_{3}\right)},
\end{eqnarray*}
where $\alpha_{in} = n_i/n$ for $i=1,2$, and 
\begin{eqnarray*}
	\gamma_n = \frac{\alpha_{1n}\left\{(1-\alpha_{1n})\sqrt{d(\widehat{P}_{1},\widehat{P}_{2})}-(1-\alpha_{1n}-\alpha_{2n})\sqrt{d(\widehat{P}_{2},\widehat{P}_{3})}\right\}}{\alpha_{1n}\sqrt{d(\widehat{P}_{1},\widehat{P}_{2})}+(1-\alpha_{1n}-\alpha_{2n})\sqrt{d(\widehat{P}_{2},\widehat{P}_{3})}}.
\end{eqnarray*}
It remains to show that $\limsup_{n \rightarrow \infty} u_n < \infty$ and $\liminf_{n \rightarrow \infty} v_n > 0$ which would then imply that\\ $\limsup_{n \rightarrow \infty} u_n/v_n < \infty$. Consequently, there exists $n_0 \geq 1$ such that whenever $n \geq n_0$, we will have $n \geq u_n/v_n$, and so $N^*_{min} = n-1$ for all $n \geq n_0$. Note that since the kernel $k(\cdot,\cdot)$ used in the MMD measure is bounded by $1$, and all of the $\alpha_{in}$'s are less or equal to $1$, it follows that $u_n \leq 4$ for all $n$. Thus, $\limsup_{n \rightarrow \infty} u_n < \infty$. To prove that $\liminf_{n \rightarrow \infty} v_n > 0$, note that since \{$v_{n}$\} is a bounded sequence, there exists a subsequence of $\{v_{n}\}$ which converges to $\liminf\limits_{n\rightarrow\infty}v_{n}$. Also, along that subsequence, there exist subsubsequences of $\{\alpha_{1n}\}$, $\{\alpha_{2n}\}$ and $\{\gamma_{n}\}$ which converge to $\alpha_1$, $\alpha_2$ and $\gamma$, respectively. Denote $\{m_n\}$ to be such a subsubsequence such that $v_{m_n} \rightarrow \liminf\limits_{n\rightarrow\infty}v_{n}$, $\alpha_{1m_{n}} \rightarrow \alpha_{1}$, $\alpha_{2m_{n}} \rightarrow \alpha_{2}$ and $\gamma_{m_n} \to \gamma$. This type of subsubsequence always exist as the sequences $\{\alpha_{1{n}}\}$, $\{\alpha_{2{n}}\}$ and $\{\gamma_{n}\}$ are bounded. In the following calculations, $\widehat{P}_1$, $\widehat{P}_2$ and $\widehat{P}_3$ represent three empirical distributions along $\{m_n\}$ sequence. To simplify the notation, we are using $\widehat{P}_{i}$ instead of the more cumbersome notation $\widehat{P}_{im_{n}}$ to represent these distributions for $i=1,2,3$. Then, 
\begin{align*}
	\left|v_{m_n}-(\alpha_{1}+\gamma)(1-\alpha_{1}-\gamma)d\left(\frac{\alpha_{1}}{\alpha_{1}+\gamma}P_{1}+\frac{\gamma}{\alpha_{1}+\gamma}P_{2},\frac{\alpha_{2}-\gamma}{1-\alpha_{1}-\gamma}P_{2}+\frac{1-\alpha_{1}-\alpha_{2}}{1-\alpha_{1}-\gamma}P_{3}\right) \right| \le \sum_{i=1}^6 J_i, \text{ where}
\end{align*}	
\begin{align*}	
	J_1 & =  d(\widehat{P}_{1},\widehat{P}_{2})\left\{|\alpha_{1m_{n}}(\alpha_{2m_{n}}-\gamma_{m_n})-\alpha_{1}(\alpha_{2}-\gamma)|+\left|\frac{\alpha_{1m_{n}}\gamma_{m_n}(1-\alpha_{1m_{n}}-\gamma_{m_n})}{\alpha_{1m_{n}}+\gamma_{m_n}}-\frac{\alpha_{1}\gamma(1-\alpha_{1}-\gamma)}{\alpha_{1}+\gamma}\right|\right\}, 
	\\
	J_2 & = \left|\alpha_{1}(\alpha_{2}-\gamma)-\frac{\alpha_{1}\gamma(1-\alpha_{1}-\gamma)}{\alpha_{1}+\gamma}\right|~|d(\widehat{P}_{1},\widehat{P}_{2})-d(P_1,P_2)|, 
	\\
	J_3 & = d(\widehat{P}_{1},\widehat{P}_{3})|\alpha_{1m_{n}}(1-\alpha_{1m_{n}}-\alpha_{2m_{n}})
	-\alpha_{1}(1-\alpha_{1}-\alpha_{2})|,
	\\
	J_4 & = \alpha_{1}(1-\alpha_{1}-\alpha_{2})|d(\widehat{P}_{1},\widehat{P}_{3})-d(P_1,P_3)|,
	\\
	J_5 & =  d(\widehat{P}_{2},\widehat{P}_{3})\left\{|\gamma_{m_n}(1-\alpha_{1m_{n}}-\alpha_{2m_{n}}) 
	-\gamma(1-\alpha_{1}-\alpha_{2})| \textcolor{white}{+\left|\frac{(\alpha_{1m_{n}}+\gamma_{m_n})(\alpha_{2m_{n}}-\gamma_{m_n})(1-\alpha_{1m_{n}}-\alpha_{2m_{n}})}{1-\alpha_{1m_{n}}-\gamma_{m_n}}-\frac{(\alpha_{1}+\gamma)(\alpha_{2}-\gamma)(1-\alpha_{1}-\alpha_{2})}{1-\alpha_{1}-\gamma}\right|}\right. \\ 
	& + \left. \left|\frac{(\alpha_{1m_{n}}+\gamma_{m_n})(\alpha_{2m_{n}}-\gamma_{m_n})(1-\alpha_{1m_{n}}-\alpha_{2m_{n}})}{1-\alpha_{1m_{n}}-\gamma_{m_n}}-\frac{(\alpha_{1}+\gamma)(\alpha_{2}-\gamma)(1-\alpha_{1}-\alpha_{2})}{1-\alpha_{1}-\gamma}\right|\right\}, \text{ and}
	\\
	J_6 & = \left|\gamma(1-\alpha_{1}-\alpha_{2})-\frac{(\alpha_{1}+\gamma)(\alpha_{2}-\gamma)(1-\alpha_{1}-\alpha_{2})}{1-\alpha_{1}-\gamma}\right|~|d(\widehat{P}_{2},\widehat{P}_{3})-d(P_{2},P_{3})|.
\end{align*}
Since each term of the RHS in the above inequality converges to zero as $n\to\infty$, it follows that $v_{m_n} \to (\alpha_{1}+\gamma)(1-\alpha_{1}-\gamma)d(\frac{\alpha_{1}}{\alpha_{1}+\gamma}P_{1}+\frac{\gamma}{\alpha_{1}+\gamma}P_{2},\frac{\alpha_{2}-\gamma}{1-\alpha_{1}-\gamma}P_{2}+\frac{1-\alpha_{1}-\alpha_{2}}{1-\alpha_{1}-\gamma}P_{3})$. Due to the assumptions of theorem, it follows that $(\alpha_{1}+\gamma)(1-\alpha_{1}-\gamma)>0$ which now ensures that $\liminf\limits_{n\rightarrow\infty}v_{n}>0$.\\
\underline{Case 3 ($K > 3$)}: The $d_w^{*}$ curve is increasing between $r=1$ and $r=n_1$, and it is decreasing between $r=\sum_{i=1}^{K-1}n_{i}$ and $r=n-1$. Thus, we need to compare the value at $r=n-1$ with the minimum value between the $(n_{1}+1)$-th and the $(\sum_{i=1}^{K-1}n_{i})$-th iterations. It is to be noted that, for $l=2,3,\dots,K-1$, the $d_w^*$ curve is convex between the $(\sum_{i=1}^{l-1}n_{i})$-th and the $(\sum_{i=1}^{l}n_{i})$-th iterations as can be seen from equation (13) of the main paper. 
	\\
In case for $l=2,3,\dots,K-1$, the $d_w^*$ curve is increasing between the ($\sum_{i=1}^{l-1}n_{i}+1$)-th and the $(\sum_{i=1}^{l}n_{i})$-th iterations (which would happen under condition (14) of the main paper), the minimum value between these iterations will be at $r=\sum_{i=1}^{l-1}n_{i}+1$, and the minimum value is larger than the value at $r=\sum_{i=1}^{l-1}n_{i}$. If the $d_w^*$ curve is decreasing between the ($\sum_{i=1}^{l-1}n_{i}+1$)-th and the $(\sum_{i=1}^{l}n_{i})$-th iterations (which would happen under condition (15) of the main paper), the minimum value between these iterations will be at $r=\sum_{i=1}^{l}n_{i}$. Thus, it is required to check whether the value at $r=n-1$ is smaller than the value at $r=\sum_{i=1}^{l}n_{i}$ for all $l=2,3,\dots,K-1$. It can be shown that for $l=2,3,\dots,K-1$, the value of $d_w^*$ at $r=n-1$ is smaller than the value at $r=\sum_{i=1}^{l}n_{i}$ if and only if
\begin{eqnarray*}
    n \geq \frac{u_n^{'}}{v_n^{'}} := \frac{d\left(\frac{\sum_{i=1}^{K-1}n_{i}\widehat{P}_{i}}{n-1}+\frac{n_{K}-1}{n-1}\widehat{P}_{K},\widehat{P}_{K}\right)}{(\sum_{i=1}^{l}\alpha_{in})(\sum_{j=l+1}^{K}\alpha_{jn})d\left(\frac{\sum_{i=1}^{l}n_{i}\widehat{P}_{i}}{\sum_{i=1}^{l}n_{i}},\frac{\sum_{j=l+1}^{K}n_{j}\widehat{P}_{j}}{\sum_{j=l+1}^{K}n_{j}}\right)},
\end{eqnarray*}
where $\alpha_{in}=n_{i}/n$ for $i=1,2,\dots,K$. Since the kernel $k(\cdot,\cdot)$ used in the MMD measure is bounded by $1$, and all of the $\alpha_{in}$'s are less or equal to $1$, it follows that $u_n^{'} \leq 4$ for all $n$. Thus, $\limsup\limits_{n\rightarrow\infty}u_n^{'} < \infty$. To prove that $\liminf\limits_{n\rightarrow\infty}v_n^{'} > 0$, it is to be noted that since $\{v_n^{'}\}$ is a bounded sequence, there exists a subsequence of $\{v_n^{'}\}$ which converges to $\liminf\limits_{n\rightarrow\infty}v_n^{'}$. Also, along that subsequence, there exist subsubsequences of $\{\alpha_{in}\}$ which converges to $\alpha_{i}$, for $i=1,2,\dots,K$. Denote $\{m_{n}\}$ to be a subsubsequence such that $v_{m_n}^{'}\rightarrow\liminf\limits_{n\rightarrow\infty}v_n^{'}$ and $\alpha_{im_{n}}\rightarrow\alpha_{i}$ for $i=1,2,\dots,K$. This type of subsubsequences always exist as the sequences $\{\alpha_{in}\}$ is bounded for $i=1,2,\dots,K$.\\ Then,
\begin{align*}
    &\left|v_{m_n}^{'}-(\sum_{i=1}^{l}\alpha_{i})(\sum_{j=l+1}^{K}\alpha_{j})d\left(\frac{\sum_{i=1}^{l}\alpha_{i}{P}_{i}}{\sum_{i=1}^{l}\alpha_{i}},\frac{\sum_{j=l+1}^{K}\alpha_{j}{P}_{j}}{\sum_{j=l+1}^{K}\alpha_{j}}\right)\right|\\
    &\leq d\left(\frac{\sum_{i=1}^{l}n_{i}\widehat{P}_{i}}{\sum_{i=1}^{l}n_{i}},\frac{\sum_{j=l+1}^{K}n_{j}\widehat{P}_{j}}{\sum_{j=l+1}^{K}n_{j}}\right)|(\sum_{i=1}^{l}\alpha_{in})(\sum_{j=l+1}^{K}\alpha_{jn})-(\sum_{i=1}^{l}\alpha_{i})(\sum_{j=l+1}^{K}\alpha_{j})|\\&+(\sum_{i=1}^{l}\alpha_{i})(\sum_{j=l+1}^{K}\alpha_{j})\left|d\left(\frac{\sum_{i=1}^{l}n_{i}\widehat{P}_{i}}{\sum_{i=1}^{l}n_{i}},\frac{\sum_{j=l+1}^{K}n_{j}\widehat{P}_{j}}{\sum_{j=l+1}^{K}n_{j}}\right)-d\left(\frac{\sum_{i=1}^{l}\alpha_{i}{P}_{i}}{\sum_{i=1}^{l}\alpha_{i}},\frac{\sum_{j=l+1}^{K}\alpha_{j}{P}_{j}}{\sum_{j=l+1}^{K}\alpha_{j}}\right)\right|.
\end{align*}
Since each term of the RHS in the above inequality converges to zero as $n\rightarrow\infty$, it follows that \\$v_{m_n}^{'}\rightarrow(\sum_{i=1}^{l}\alpha_{i})(\sum_{j=l+1}^{K}\alpha_{j})d\left(\frac{\sum_{i=1}^{l}\alpha_{i}{P}_{i}}{\sum_{i=1}^{l}\alpha_{i}},\frac{\sum_{j=l+1}^{K}\alpha_{j}{P}_{j}}{\sum_{j=l+1}^{K}\alpha_{j}}\right)$. Due to the assumptions of theorem, it follows that $(\sum_{i=1}^{l}\alpha_{i})(\sum_{j=l+1}^{K}\alpha_{j})>0$ which ensures that $\liminf\limits_{n\rightarrow\infty}v_n^{'}>0$. Now, $\limsup\limits_{n\rightarrow\infty}u_{n}^{'}<\infty$ and $\liminf\limits_{n\rightarrow\infty}v_{n}^{'}>0$ together imply that $\limsup\limits_{n\rightarrow\infty}u_{n}^{'}/{v_{n}^{'}}<\infty$. Consequently, there exists $n_0^{'}\geq 1$ such that $n\geq n_0^{'}$, we will have $n \geq u_{n}^{'}/{v_{n}^{'}}$, and the value of $d_w^*$ will be smaller than the value at $r=\sum_{i=1}^{l}n_{l}$ for $n \geq n_0^{'}$.
\\
If the $d_w^*$ curve is U-shaped between the ($\sum_{i=1}^{l-1}n_{i}+1$)-th and the $(\sum_{i=1}^{l}n_{i})$-th iterations (which would happen under certain condition), it can be shown that the minimum value of the $d_w^*$ curve between these iterations is larger than the value at $r=n-1$ if and only if  $n \geq \widetilde{u}_n/\widetilde{v}_n$, where
\begin{eqnarray*}
    \frac{\widetilde{u}_n}{\widetilde{v}_n} := \frac{d\left(\frac{\sum_{i=1}^{K-1}n_{i}\widehat{P}_{i}}{n-1}+\frac{n_{K}-1}{n-1}\widehat{P}_{K},\widehat{P}_{K}\right)}{(\sum_{i=1}^{l-1}\alpha_{in}+\gamma_{n})(1-\sum_{i=1}^{l-1}\alpha_{in}-\gamma_{n})d\left(\frac{\sum_{i=1}^{l-1}\alpha_{in}\widehat{P}_{i}}{\sum_{i=1}^{l-1}\alpha_{in}+\gamma_{n}}+\frac{\gamma_{n}}{\sum_{i=1}^{l-1}\alpha_{in}+\gamma_{n}}\widehat{P}_{l},\frac{\alpha_{ln}-\gamma_{n}}{1-\sum_{i=1}^{l-1}\alpha_{in}-\gamma_{n}}\widehat{P}_{l}+\frac{\sum_{j=l+1}^{K}\alpha_{jn}\widehat{P}_{j}}{1-\sum_{i=1}^{l-1}\alpha_{in}-\gamma_{n}}\right)},
\end{eqnarray*}
and
\begin{eqnarray*}
    \gamma_{n} = \frac{(\sum_{i=1}^{l-1}\alpha_{in})\left\{(1-\sum_{i=1}^{l-1}\alpha_{in})\sqrt{d\left(\frac{\sum_{i=1}^{l-1}n_{i}\widehat{P}_{i}}{\sum_{i=1}^{l-1}n_{i}},\widehat{P}_{l}\right)}-(1-\sum_{i=1}^{l}\alpha_{in})\sqrt{d\left(\widehat{P}_{l},\frac{\sum_{j=l+1}^{K}n_{j}\widehat{P}_{j}}{\sum_{j=l+1}^{K}n_{j}}\right)}\right\}}{(\sum_{i=1}^{l-1}\alpha_{in})\sqrt{d\left(\frac{\sum_{i=1}^{l-1}n_{i}\widehat{P}_{i}}{\sum_{i=1}^{l-1}n_{i}},\widehat{P}_{l}\right)}+(1-\sum_{i=1}^{l}\alpha_{in})\sqrt{d\left(\widehat{P}_{l},\frac{\sum_{j=l+1}^{K}n_{j}\widehat{P}_{j}}{\sum_{j=l+1}^{K}n_{j}}\right)}}.
\end{eqnarray*}
If $\limsup\limits_{n\rightarrow\infty}\widetilde{u}_{n}<\infty$ and $\liminf\limits_{n\rightarrow\infty}\widetilde{v}_{n}>0$, then it would imply that $\limsup\limits_{n\rightarrow\infty}\widetilde{u}_{n}/{\widetilde{v}_{n}}<\infty$. Consequently, there exists $\widetilde{n}_0 \geq 1$ such that whenever $n\geq \widetilde{n}_0$, we will have $n\geq \widetilde{u}_{n}/\widetilde{v}_{n}$, and therefore, the value of $d_w^*$ at $r=n-1$ will be smaller than the value at any iteration between $\sum_{i=1}^{l-1}n_{i}+1$ and $\sum_{i=1}^{l}n_{i}$. Since the kernel $k(\cdot,\cdot)$ used in the MMD measure is bounded by $1$ and all of the $\alpha_{in}$'s are less or equal to $1$, it follows that $\widetilde{u}_n \leq 4$ for all $n$. Thus, $\limsup\limits_{n\rightarrow\infty}\widetilde{u}_n < \infty$. To prove that $\liminf\limits_{n\rightarrow\infty}\widetilde{v}_n > 0$, note that since $\{\widetilde{v}_n\}$ is a bounded sequence, there exists a subsequence of $\{\widetilde{v}_n\}$ which converges to $\liminf\limits_{n\rightarrow\infty}\widetilde{v}_n$. Also, along that subsequence, there exist subsubsequences of $\{\alpha_{in}\}$ and $\{\gamma_{n}\}$ which converge to $\alpha_{i}$ and $\gamma$, respectively for $i=1,2,\dots,K$. Denote $\{l_{n}\}$ to be a subsubsequence such that $\widetilde{v}_{l_n}\rightarrow\liminf\limits_{n\rightarrow\infty}\widetilde{v}_n$, $\alpha_{il_{n}}\rightarrow \alpha_{i}$ and $\gamma_{l_n}\rightarrow\gamma$ for $i=1,2,\dots,K$. This type of subsubsequences always exist as the sequences $\{\alpha_{in}\}$ and $\{\gamma_{n}\}$ are bounded, $i=1,2,\dots,K$. Then, by the triangle inequality,
\begin{align*}
    &\left|\widetilde{v}_{l_n}-(\sum_{i=1}^{l-1}\alpha_{i}+\gamma)(1-\sum_{i=1}^{l-1}\alpha_{i}-\gamma)d\left(\frac{\sum_{i=1}^{l-1}\alpha_{i}P_{i}}{\sum_{i=1}^{l-1}\alpha_{i}+\gamma}+\frac{\gamma}{\sum_{i=1}^{l-1}\alpha_{i}+\gamma}P_{l},\frac{\alpha_{l}-\gamma}{1-\sum_{i=1}^{l-1}\alpha_{i}-\gamma}P_{l}+\frac{\sum_{j=l+1}^{K}\alpha_{j}P_{j}}{1-\sum_{i=1}^{l-1}\alpha_{i}-\gamma}\right)\right|\\
    &\leq I_1+I_2+I_3+I_4+I_5+I_6+I_7+I_8, \text{ where}
    \end{align*}
    \begin{align*}
    I_1=d\left(\frac{\sum_{i=1}^{l-1}\alpha_{il_{n}}\widehat{P}_{i}}{\sum_{i=1}^{l-1}\alpha_{il_{n}}},\widehat{P}_{l}\right)\left|(\sum_{i=1}^{l-1}\alpha_{il_{n}})(\alpha_{ll_{n}}-\gamma_{l_n})-(\sum_{i=1}^{l-1}\alpha_{i})(\alpha_{l}-\gamma)\right|,
    \end{align*}
    \begin{align*}
    I_2=d\left(\frac{\sum_{i=1}^{l-1}\alpha_{il_{n}}\widehat{P}_{i}}{\sum_{i=1}^{l-1}\alpha_{il_{n}}},\widehat{P}_{l}\right)\left|\frac{\gamma_{l_n}(\sum_{i=1}^{l-1}\alpha_{il_{n}})(1-\sum_{i=1}^{l-1}\alpha_{il_{n}}-\gamma_{l_n})}{\sum_{i=1}^{l-1}\alpha_{il_{n}}+\gamma_{l_n}}-\frac{\gamma(\sum_{i=1}^{l-1}\alpha_{i})(1-\sum_{i=1}^{l-1}\alpha_{i}-\gamma)}{\sum_{i=1}^{l-1}\alpha_{i}+\gamma}\right|,
    \end{align*}
    \begin{align*}
    I_3=\left|(\sum_{i=1}^{l-1}\alpha_{i})(\alpha_{l}-\gamma)-\frac{\gamma(\sum_{i=1}^{l-1}\alpha_{i})(1-\sum_{i=1}^{l-1}\alpha_{i}-\gamma)}{\sum_{i=1}^{l-1}\alpha_{i}+\gamma}\right|\left|d\left(\frac{\sum_{i=1}^{l-1}\alpha_{il_{n}}\widehat{P}_{i}}{\sum_{i=1}^{l-1}\alpha_{il_{n}}},\widehat{P}_{l}\right) -d\left(\frac{\sum_{i=1}^{l-1}\alpha_{i}P_{i}}{\sum_{i=1}^{l-1}\alpha_{i}},P_{l}\right)\right|, 
    \end{align*}
    \begin{align*}
    I_4=d\left(\frac{\sum_{i=1}^{l-1}\alpha_{il_{n}}\widehat{P}_{i}}{\sum_{i=1}^{l-1}\alpha_{il_{n}}},\frac{\sum_{j=l+1}^{K}\alpha_{jl_{n}}\widehat{P}_{j}}{\sum_{j=l+1}^{K}\alpha_{jl_{n}}}\right)\left|(\sum_{i=1}^{l-1}\alpha_{il_{n}})(1-\sum_{i=1}^{l}\alpha_{il_{n}})-(\sum_{i=1}^{l-1}\alpha_{i})(1-\sum_{i=1}^{l}\alpha_{i})\right|,
    \end{align*}
    \begin{align*}
    I_5=(\sum_{i=1}^{l-1}\alpha_{i})(1-\sum_{i=1}^{l}\alpha_{i})\left|d\left(\frac{\sum_{i=1}^{l-1}\alpha_{il_{n}}\widehat{P}_{i}}{\sum_{i=1}^{l-1}\alpha_{il_{n}}},\frac{\sum_{j=l+1}^{K}\alpha_{jl_{n}}\widehat{P}_{j}}{\sum_{j=l+1}^{K}\alpha_{jl_{n}}}\right)-d\left(\frac{\sum_{i=1}^{l-1}\alpha_{i}\widehat{P}_{i}}{\sum_{i=1}^{l-1}\alpha_{i}},\frac{\sum_{j=l+1}^{K}\alpha_{j}\widehat{P}_{j}}{\sum_{j=l+1}^{K}\alpha_{j}}\right)\right|,
    \end{align*}
    \begin{align*}
I_6=d\left(\widehat{P}_{l},\frac{\sum_{j=l+1}^{K}\alpha_{jl_{n}}\widehat{P}_{j}}{\sum_{j=l+1}^{K}\alpha_{jl_{n}}}\right)|\gamma_{l_n}(1-\sum_{i=1}^{l}\alpha_{il_{n}})-\gamma(1-\sum_{i=1}^{l}\alpha_{i})|,
\end{align*}
    \begin{align*}
    I_7=d\left(\widehat{P}_{l},\frac{\sum_{j=l+1}^{K}\alpha_{jl_{n}}\widehat{P}_{j}}{\sum_{j=l+1}^{K}\alpha_{jl_{n}}}\right)\left|\frac{(\sum_{i=1}^{l-1}\alpha_{il_{n}}+\gamma_{l_n})(\alpha_{ll_{n}}-\gamma_{l_n})(1-\sum_{i=1}^{l}\alpha_{il_{n}})}{1-\sum_{i=1}^{l-1}\alpha_{il_{n}}-\gamma_{l_n}} -\frac{(\sum_{i=1}^{l-1}\alpha_{i}+\gamma)(\alpha_{l}-\gamma)(1-\sum_{i=1}^{l}\alpha_{i})}{1-\sum_{i=1}^{l-1}\alpha_{i}-\gamma}\right|,
    \end{align*}
    \begin{align*}
    I_8=\left|\gamma(1-\sum_{i=1}^{l}\alpha_{i})-\frac{(\sum_{i=1}^{l-1}\alpha_{i}+\gamma)(\alpha_{l}-\gamma)(1-\sum_{i=1}^{l}\alpha_{i})}{1-\sum_{i=1}^{l-1}\alpha_{i}-\gamma}\right|\left|d\left(\widehat{P}_{l},\frac{\sum_{j=l+1}^{K}\alpha_{jl_{n}}\widehat{P}_{j}}{\sum_{j=l+1}^{K}\alpha_{jl_{n}}}\right)-d\left(P_{l},\frac{\sum_{j=l+1}^{K}\alpha_{j}P_{j}}{\sum_{j=l+1}^{K}\alpha_{j}}\right)\right|.
\end{align*}
Since each of $I_i$'s ($i=1,2,\dots,8$) converges to zero as $n\rightarrow\infty$, it follows that $\widetilde{v}_{l_n}\rightarrow (\sum_{i=1}^{l-1}\alpha_{i}+\gamma)(1-\sum_{i=1}^{l-1}\alpha_{i}-\gamma)d\left(\frac{\sum_{i=1}^{l-1}\alpha_{i}P_{i}}{\sum_{i=1}^{l-1}\alpha_{i}+\gamma}+\frac{\gamma}{\sum_{i=1}^{l-1}\alpha_{i}+\gamma}P_{l},\frac{\alpha_{l}-\gamma}{1-\sum_{i=1}^{l-1}\alpha_{i}-\gamma}P_{l}+\frac{\sum_{j=l+1}^{K}\alpha_{j}P_{j}}{1-\sum_{i=1}^{l-1}\alpha_{i}-\gamma}\right)$. Due to the assumptions of theorem, it follows that $(\sum_{i=1}^{l-1}\alpha_{i}+\gamma)(1-\sum_{i=1}^{l-1}\alpha_{i}-\gamma)>0$ which ensures that $\liminf\limits_{n\rightarrow\infty}\widetilde{v}_{n}>0$.\\
Hence, for $n\geq \max\{(n_0^{'},\widetilde{n}_0)\}:=n_0^*$, the value of $d_w^*$ at iteration $r=n-1$ will be smaller than the value at any iteration in between $r=\sum_{i=1}^{l-1}n_{i}$ and $r=\sum_{i=1}^{l}n_{i}$. We can find $n_0^{*}$ for $l=2,3,\dots,K-1$ and denote their maximum by $n_0$. (However, the values of $d_w^*$ at the $n_1$-th iteration and the $(n-n_K)$-th iteration are larger than the value at the $(n-1)$-th iteration as the $d_w^{*}$ curve is increasing between $r=1$ and $r=n_1$, and it is decreasing between $r=\sum_{i=1}^{K-1}n_{i}$ and $r=n-1$). Hence, $N_{\min}^{*}=n-1$ for $n\geq n_0$.
\\
\textit{(c)} Under the condition in part \textit{(b)}, there exists $n_0 \geq 1$ such that $N_{\min}^{*}=n-1$ for $n \geq n_0$. Hence, for $n \geq n_0$, $N_{\min}^{*}(n-N_{\min}^{*})=n-1$. Since the maximum of $d_w^*$ occurs in between iterations $r=n_{1}$ and $r=n-n_{K}$, $N_{\max}^{*} \geq n_{1}$ and $n-N_{\max}^{*} \geq n_{K}$. Therefore,
\begin{eqnarray*}
    H^* = \frac{N_{\max}^{*}(n-N_{\max}^{*})}{N_{\min}^{*}(n-N_{\min}^{*})} \geq \frac{n_{1}n_{K}}{n-1} > \frac{n_{1}n_{K}}{n} = n\alpha_{1n}\alpha_{Kn}.
\end{eqnarray*}
Under the conditions of theorem, $\liminf\limits_{n\rightarrow\infty}\alpha_{1n}\alpha_{Kn}>0$. Thus, $H^*$ can be made arbitrarily large for large values of $n$. Hence, there exists $n^*\geq n_0$ such that for $n \geq n^*$, $H^*>1$.
﻿
\subsection*{Proof of Theorem 2}
 We will prove this theorem for two-class problem ($K=2$), three-class problem ($K=3$) and general $K \:(>3)$-class problem.\\ 
\underline{Case 1 ($K = 2$)}:
First, we consider the case when the observations are coming from two different distributions. Let us denote the two empirical distributions by $\widehat{P}_1$ and $\widehat{P}_2$, which have been constructed based on samples of sizes $n_1$ and $n_2$, respectively. Denote $n = n_1 + n_2$.\\\noindent
\noindent
The following lemma describes the behaviour of $d_w^*$ curve in two-class problem. 
\begin{lemma}\label{lma2}
In two-class problem, the following statements are true.
\begin{enumerate} [label=(\alph*)]
    \item An observation from $\widehat{P}_1$ will be transferred to $C_1$ in the first iteration if and only if $n_1 \leq n_2$. Furthermore, if $n_1 \leq n_2$, an observation of $\widehat{P}_1$ will be transferred to $C_1$ in each of first $n_1$ iterations and $d_w^*$ curve keeps on increasing till the $n_1$-th iteration.
    \item Under conditions of $(a)$, the $d_{w}^{*}$ curve goes on decreasing after the $n_1$-th iteration.
\end{enumerate}
\end{lemma}
\begin{proof} 
\textit{(a)}
If the first observation in $C_1$ is from $\widehat{P}_1$, then
\begin{center}
    $\widetilde{P}_{C_1^{(1)}} = \widehat{P}_1$ and $\widetilde{P}_{C_2^{(2)}} = \frac{n_1-1}{n_1+n_2-1}\widehat{P}_1+\frac{n_2}{n_1+n_2-1}\widehat{P}_2$,
\end{center}
and if the first observation in $C_1$ is from $\widehat{P}_2$, then
\begin{center}
    $\widetilde{P}_{C_1^{(1)}} = \widehat{P}_2$ and $\widetilde{P}_{C_2^{(1)}} = \frac{n_1}{n_1+n_2-1}\widehat{P}_1+\frac{n_2-1}{n_1+n_2-1}\widehat{P}_2$.
\end{center}
Now,
        \begin{align*}
            &d\left(\widehat{P}_1,\frac{n_1-1}{n_1+n_2-1}\widehat{P}_1+\frac{n_2}{n_1+n_2-1}\widehat{P}_2\right) - d\left(\widehat{P}_2,\frac{n_1}{n_1+n_2-1}\widehat{P}_1+\frac{n_2-1}{n_1+n_2-1}\widehat{P}_2\right) \\
            &= \left(1 - \frac{n_1-1}{n_1+n_2-1}\right)^{2}d(\widehat{P}_1,\widehat{P}_2) - \left(\frac{n_1}{n_1+n_2-1}\right)^{2}d(\widehat{P}_1,\widehat{P}_2) \\
            & =\frac{n_2^{2}-n_1^{2}}{(n_1+n_2-1)^{2}}d(\widehat{P}_1,\widehat{P}_2),
        \end{align*}
        which is non-negative if and only if $n_1 \leq n_2$. Since maximizing $d_w^*$ in a particular iteration is equivalent to maximizing $d$ between the representative distributions of the sets of that iteration, $d_w^*(C_1^{(1)},C_1^{(2)})$ is maximized when an observation from $\widehat{P}_1$ is transferred to $C_1$ in the first iteration if and only if $n_1 \leq n_2$.
\\
Now, assume that $r$ ($1 \leq r \leq n_1-1$) many observations of $\widehat{P}_1$ have already been transferred to $C_1$ in first $r$ iterations. After the $r$-th iteration, we have two sets $C_1$ which contains observations from $\widehat{P}_1$ only and $C_2$ which contains observations from both $\widehat{P}_1$ and $\widehat{P}_2$. Then,
\begin{center}
    $\widetilde{P}_{C_{1}^{(r)}} = \widehat{P}_1$ and $\widetilde{P}_{C_{2}^{(r)}} = \frac{n_1-r}{n_1+n_2-r}\widehat{P}_1 + \frac{n_2}{n_1+n_2-r}\widehat{P}_2$.
\end{center}
\noindent
  \noindent
  Then,
  \begin{align} \label{eq:46}
      d_w^*(C_1^{(r)},C_2^{(r)}) &= \frac{r(n-r)}{n}d(\widehat{P}_1,\frac{n_1-r}{n_1+n_2-r}\widehat{P}_1 + \frac{n_2}{n_1+n_2-r}\widehat{P}_2) \nonumber\\
      &=\frac{r(n-r)}{n}\left(\frac{n_2}{n-r}\right)^2d(\widehat{P}_1,\widehat{P}_2)\nonumber\\
      &=\frac{rn_2^2}{n(n-r)}d(\widehat{P}_1,\widehat{P}_2).
  \end{align}
  In the next iteration, we have two options: either transfer an observation of $\widehat{P}_1$ or transfer an observation of $\widehat{P}_2$ to $C_1$. First, we consider the situation where an observation of distribution $\widehat{P}_1$ is transferred from $C_2$ to $C_1$ in the ($r+1$)-th iteration. In this case, $C_{1}^{(r+1)}$ can be considered as a possible realization of $\widehat{P}_1$ and $C_{2}^{(r+1)}$ can be considered as a possible realization of $\frac{n_1-r-1}{n_1+n_2-r-1}$$\widehat{P}_1 + \frac{n_2}{n_1+n_2-r-1}\widehat{P}_2$, respectively.
\noindent
\newline
\noindent
 Thus,
 \begin{center}
$\widetilde{P}_{C_{1}^{(r+1)}} = \widehat{P}_1$ and $\widetilde{P}_{C_{1}^{(r+1)}} = \frac{n_1-r-1}{n_1+n_2-r-1}\widehat{P}_1 + \frac{n_2}{n_1+n_2-r-1}\widehat{P}_2$
 \end{center}
 and 
\begin{align} \label{eq:60}
d_{w}^{*}(C_{1}^{(r+1)},C_{2}^{(r+1)}) 
=& \frac{(r+1)(n_1+n_2-r-1)}{n_1+n_2}d\left(\widehat{P}_1,\frac{n_1-r-1}{n_1+n_2-r-1}\widehat{P}_1 + \frac{n_2}{n_1+n_2-r-1}\widehat{P}_2\right).
\end{align}
\noindent
 If an observation of distribution $\widehat{P}_2$ is transferred from $C_2$ to $C_1$ in the ($r+1$)-th iteration, $C_{1}^{(r+1)}$ can be considered as a possible realization of $\frac{r}{r+1}$$\widehat{P}_1 + \frac{1}{r+1}$$\widehat{P}_2$ and $C_{2}^{(r+1)}$ can be considered as a possible realization of $\frac{n_1-r}{n-r-1}\widehat{P}_1 + \frac{n_2-1}{n-r-1}\widehat{P}_2$.
\noindent
\newline
 Thus,
 \begin{center}
     $\widetilde{P}_{C_{1}^{(r+1)}} = \frac{r}{r+1}\widehat{P}_1 + \frac{1}{r+1}\widehat{P}_2$ and $\widetilde{P}_{C_{2}^{(r+1)}} = \frac{n_1-r}{n-r-1}\widehat{P}_1 + \frac{n_2-1}{n-r-1}\widehat{P}_2$. 
 \end{center}
 Therefore,
\begin{align} \label{eq:61}
d_{w}^{*}(C_{1}^{(r+1)},C_{2}^{(r+1)})
 = \frac{(r+1)(n-r-1)}{n}d\left(\frac{r}{r+1}\widehat{P}_1 + \frac{1}{r+1}\widehat{P}_2,\frac{n_1-r}{n-r-1}\widehat{P}_1 + \frac{n_2-1}{n-r-1}\widehat{P}_2\right).  
\end{align}
\noindent
 It is to be noted that
\begin{align*}
    d\left(\widehat{P}_1,\frac{n_1-r-1}{n-r-1}\widehat{P}_1 + \frac{n_2}{n-r-1}\widehat{P}_2\right) =& \left(1 - \frac{n_1-r-1}{n-r-1}\right)^{2}d(\widehat{P}_1,\widehat{P}_2) \nonumber\\ =& \frac{n_2^2}{(n-r-1)^2}d(\widehat{P}_1,\widehat{P}_2), 
\end{align*}
and
\begin{align*}
d\left(\frac{r}{r+1}\widehat{P}_1 + \frac{1}{r+1}\widehat{P}_2,\frac{n_1-r}{n-r-1}\widehat{P}_1 + \frac{n_2-1}{n-r-1}\widehat{P}_2\right) =& \left(\frac{r}{r+1} - \frac{n_1-r}{n-r-1}\right)^{2}d(\widehat{P}_1,\widehat{P}_2) \nonumber\\
=& \frac{(rn_2-n_1)^2}{(r+1)^{2}(n_1+n_2-r-1)^{2}}d(\widehat{P}_1,\widehat{P}_2). 
\end{align*}
Since $n_2 \geq n_1$, $(rn_2-n_1)^2 \leq r^{2}n_{2}^{2}$, we have
\begin{align*}
    \frac{(rn_2-n_1)^2}{(r+1)^{2}(n-r-1)^{2}}) \leq \frac{r^{2}n_{2}^{2}}{(r+1)^{2}(n-r-1)^{2}} 
    \leq \frac{n_{2}^{2}}{(n-r-1)^{2}}. 
\end{align*}
Hence, the RHS of \eqref{eq:60} will be higher than the RHS of \eqref{eq:61} if $n_1 \leq n_2$. Thus, $d_w^{*}(C_{1}^{(r+1)},C_{2}^{(r+1)})$ will be maximized if an observation of $\widehat{P}_1$ is transferred to $C_1$ in the $(r+1)$-th iteration. Therefore, in each of the first $n_1$ iterations, observations from $\widehat{P}_1$ will be transferred to $C_1$. Also, from \eqref{eq:46}, it is clear that in the first $n_1$ iterations, $d_w^*$ curve is an increasing function of iteration number. Hence, $d_w^*$ curve keeps on increasing till the $n_1$-th iteration.
\\
\textit{(b)} Under the assumption $n_1 \leq n_2$, after the $n_1$-th iteration, observations from $\widehat{P}_2$ will be transferred to $C_1$ in each iteration. Then, for $s = 0,1,2,\dots,n_2-1$, in the $(n_1+s)$-th iteration, $C_1$ and $C_2$ are going to be a possible realization of their size from $\frac{n_1}{n_1+s}\widehat{P}_1 + \frac{s}{n_1+s}\widehat{P}_2$ and $\widehat{P}_2$, respectively. 
\\
\noindent
  Thus,
\begin{center}
    $\widetilde{P}_{C_{1}^{(n_1+s)}} = \frac{n_1}{n_1+s}\widehat{P}_1 + \frac{s}{n_1+s}\widehat{P}_2$, and  $\widetilde{P}_{C_{2}^{(n_1+s)}} = \widehat{P}_2$. 
    \end{center}
Therefore,
\begin{align} \label{eq:47}
    d_w^{*}(C_{1}^{(n_1+s)},C_{2}^{(n_1+s)}) &= \frac{(n_1+s)(n_2-s)}{n}d\left(\frac{n_1}{n_1+s}\widehat{P}_1+\frac{s}{n_1+s}\widehat{P}_2, \widehat{P}_2\right) \nonumber\\
      & = \frac{n_1^{2}(n_2-s)}{n(n_1+s)}d(\widehat{P}_1,\widehat{P}_2),
  \end{align}
  which is a decreasing function of $s$. Therefore, $d_w^*$ curve goes on decreasing after the $n_1$-th iteration.
        \end{proof}
        \noindent
        Combining \eqref{eq:46} and \eqref{eq:47}, we get equation (8) of the main paper, the explicit form of $d_w^*$ in the two-class problem. Lemma \ref{lma2} indicates that the $d_w^*$ curve keeps on increasing till the $n_1$-th iteration and keeps on decreasing after the $n_1$-th iteration. Hence, $d_w^*(C_1^{(r)},C_2^{(r)})$ is maximum at $r=n_1$. As the decision made by the Algorithm SCC* is always correct in the two-class problem under the conditions of Theorem 2 (which follows from Theorem 1), the Algorithm SCC* detects that the set contains observations coming from more than one distribution. Thus, the Algorithm BS* splits the set into two clusters of which one contains all observations of $\widehat{P}_1$ and the other contains all observations of $\widehat{P}_2$. For each of these two sets, the Algorithm SCC* detects that the observations are coming from a single distribution. Hence, the Algorithm CURBS-I* is perfect for $K=2$.  
\noindent
\\
\underline{Case 2 ($K = 3$)}: Denote the three empirical distributions by $\widehat{P}_1$, $\widehat{P}_2$ and $\widehat{P}_3$, which have been constructed based on samples of sizes $n_1$, $n_2$ and $n_3$, respectively. Denote $n = n_1 + n_2 + n_3$.  \\\noindent
\noindent
The following lemma describes the behaviour of $d_w^*$ curve in the three-class problem. Let us restate the conditions $(A1)$, $(A2)$ and $(A3)$ which will be used in proving the lemma.\\
$(A1)$ ($n_2 - n_1$)$d(\widehat{P}_1,\widehat{P}_2) + n_3$\{$d(\widehat{P}_1,\widehat{P}_3) - d(\widehat{P}_2,\widehat{P}_3)$\} $\geq 0$
\\
$(A2)$ ($n_3 - n_1$)$d(\widehat{P}_1,\widehat{P}_3) + n_2$\{$d(\widehat{P}_1,\widehat{P}_2) - d(\widehat{P}_2,\widehat{P}_3)$\} $\geq 0$.
\\
$(A3)$ $n_1(n_2 + n_3 - 1)\{d(\widehat{P}_1,\widehat{P}_3) - d(\widehat{P}_1,\widehat{P}_2)\} + (n_3 - n_2)(n_1 + 1)d(\widehat{P}_2,\widehat{P}_3)$ $\geq 0$.    
﻿
      \begin{lemma}\label{lma3}
      In the three-class problem, the following statements are true. 
        \begin{enumerate} [label=(\alph*)]
            \item In the first iteration, an observation from $\widehat{P}_1$ will be transferred to $C_1$ if and only if both of the conditions $(A1)$ and $(A2)$ hold. Furthermore, if $(A1)$ and $(A2)$ hold, in each of first $n_1$ iterations, an observation from $\widehat{P}_1$ will be transferred to $C_1$, and the $d_w^*$ curve goes on increasing till the $n_1$-th iteration.
            \item Under the conditions $(A1)$ and $(A2)$, an observation from $\widehat{P}_2$ will be be transferred to $C_1$ in the $(n_1+1)$-th iteration if and only if the condition $(A3)$ holds. Furthermore, if $(A1)$, $(A2)$ and $(A3)$ hold, observations from $\widehat{P}_2$ will be transferred to $C_1$ from the $(n_1+1)$-th to the $(n_1+n_2)$-th iterations, and the $d_w^*$ curve is convex between the $n_1$-th and the $(n_1+n_2)$-th iterations.
            \item Under conditions $(A1)$, $(A2)$ and $(A3)$, $d_w^*$ curve goes on decreasing after the ($n_1+n_2$)-th iteration.
        \end{enumerate}
 \end{lemma}
\noindent
\begin{proof}
\textit{(a)} If the first observation transferred is from $\widehat{P}_1$, $C_1$ can be considered as a possible realization of size 1 from $\widehat{P}_1$ while $C_2$ can be considered as a possible realization of size ($n-1$) from mixture of $\widehat{P}_1$ and mixture of $\widehat{P}_2$ and $\widehat{P}_3$. Thus, $\widetilde{P}_{C_1^{(1)}} = \widehat{P}_1$ and 
\begin{align*}
\widetilde{P}_{C_2^{(1)}} =& \frac{n_1-1}{n_1+n_2+n_3-1}\widehat{P}_1 + \frac{n_2+n_3}{n_1+n_2+n_3-1}\left\{\frac{n_2}{n_2+n_3}\widehat{P}_2 + \frac{n_3}{n_2+n_3}\widehat{P}_3\right\} \\
=& \frac{n_1-1}{n_1+n_2+n_3-1}\widehat{P}_1 + \frac{n_2}{n_1+n_2+n_3-1}\widehat{P}_2 + \frac{n_3}{n_1+n_2+n_3-1}\widehat{P}_3.
\end{align*}
Therefore,
\begin{align} \label{eq:40}
    d_w^{*}(C_1^{(1)},C_2^{(1)}) = \frac{n_1+n_2+n_3-1}{n_1+n_2+n_3}d\left(\widehat{P}_1,\frac{n_1-1}{n_1+n_2+n_3-1}\widehat{P}_1+\frac{n_2}{n_1+n_2+n_3-1}\widehat{P}_2 +\frac{n_3}{n_1+n_2+n_3-1}\widehat{P}_3\right).
\end{align}
Similarly, if an observation of $\widehat{P}_2$ is transferred to $C_1$ in the first iteration then
\begin{align} \label{eq : 41}
    d_w^{*}(C_1^{(1)},C_2^{(1)}) = \frac{n_1+n_2+n_3-1}{n_1+n_2+n_3}d\left(\widehat{P}_2,\frac{n_1}{n_1+n_2+n_3-1}\widehat{P}_1+\frac{n_2-1}{n_1+n_2+n_3-1}\widehat{P}_2 +\frac{n_3}{n_1+n_2+n_3-1}\widehat{P}_3\right),
\end{align}
and if an observation of $\widehat{P}_3$ is transferred to $C_1$ in the first iteration then
\begin{align} \label{eq : 42}
    d_w^{*}(C_1^{(1)},C_2^{(1)}) = \frac{n_1+n_2+n_3-1}{n_1+n_2+n_3}d\left(\widehat{P}_3,\frac{n_1}{n_1+n_2+n_3-1}\widehat{P}_1+\frac{n_2}{n_1+n_2+n_3-1}\widehat{P}_2 +\frac{n_3-1}{n_1+n_2+n_3-1}\widehat{P}_3\right).
\end{align}
    Simple algebraic calculation shows that the RHS of \eqref{eq:40} will be higher than the RHS of \eqref{eq : 41} if and only if $(A1)$ holds and the RHS of \eqref{eq:40} will be higher than the RHS of \eqref{eq : 42} if and only if $(A2)$ holds. Also, if both $(A1)$ and $(A2)$ hold, in each of the first $n_1$ iterations, an observation from $\widehat{P}_1$ will be transferred to $C_1$ which can be easily proved using similar argument of part $(a)$ of Lemma \ref{lma2} with $\widehat{P}_2$ replaced by the mixture of $\widehat{P}_2$ and $\widehat{P}_3$.\\
    Hence, for $r=1,2,\dots,n_1$,
    \begin{center}
        $\widetilde{P}_{C_1^{(r)}} = \widehat{P}_1$ and $\widetilde{P}_{C_2^{(r)}} = \frac{n_1-r}{n-r}\widehat{P}_1 + \frac{n_2}{n-r}\widehat{P}_2 + \frac{n_3}{n-r}\widehat{P}_3$.
    \end{center}
 Therefore, for $r=1,2,\dots,n_1$,
 \begin{align} \label{eq:50}
     d_w^*(C_1^{(r)},C_2^{(r)}) &= \frac{r(n-r)}{n}d\left(\widehat{P}_1,\frac{n_1-r}{n-r}\widehat{P}_1 + \frac{n_2}{n-r}\widehat{P}_2 + \frac{n_3}{n-r}\widehat{P}_3\right) \nonumber\\
     &=\frac{r}{n(n-r)}\{n_2(n_2+n_3)d(\widehat{P}_1,\widehat{P}_2)+n_3(n_2+n_3)d(\widehat{P}_1,\widehat{P}_3)-n_2n_3d(\widehat{P}_2,\widehat{P}_3)\},
 \end{align}
which is an increasing function of $r$. Therefore, $d_w^*$ curve keeps on increasing till the $n_1$-th iteration.\\
\textit{(b)}  If an observation of $\widehat{P}_2$ is transferred to $C_1$ in the ($n_1+1$)-th iteration then 
\begin{center}
$\widetilde{P}_{C_1^{(n_1+1)}} = \frac{n_1}{n_1+1}\widehat{P}_1 + \frac{1}{n_1+1}\widehat{P}_2$ and $\widetilde{P}_{C_2^{(n_1+1)}} = \frac{n_2-1}{n_2+n_3-1}\widehat{P}_2 + \frac{n_3}{n_2+n_3-1}\widehat{P}_3$.
\end{center}
 Therefore,    
\begin{align} \label{eq:43}
d_w^{*}(C_1^{(n_1+1)},C_2^{(n_1+1)}) 
= \frac{(n_1+1)(n_2+n_3-1)}{n}d\left(\frac{n_1}{n_1+1}\widehat{P}_1 + \frac{1}{n_1+1}\widehat{P}_2,\frac{n_2-1}{n_2+n_3-1}\widehat{P}_2 + \frac{n_3}{n_2+n_3-1}\widehat{P}_3\right). 
\end{align}
\noindent
\newline
\noindent
   On the other hand if an observation of $\widehat{P}_3$ is transferred to $C_1$ in the ($n_1+1$)-th iteration then 
\begin{center}
$\widetilde{P}_{C_1^{(n_1+1)}} = \frac{n_1}{n_1+1}\widehat{P}_1 + \frac{1}{n_1+1}\widehat{P}_3$ and $\widetilde{P}_{C_2^{(n_1+1)}} = \frac{n_2}{n_2+n_3-1}\widehat{P}_2 + \frac{n_3-1}{n_2+n_3-1}\widehat{P}_3$.
\end{center}
 Therefore,    
\begin{align} \label{eq:44}
d_w^{*}(C_1^{(n_1+1)},C_2^{(n_1+1)}) 
= \frac{(n_1+1)(n_2+n_3-1)}{n}d\left(\frac{n_1}{n_1+1}\widehat{P}_1 + \frac{1}{n_1+1}\widehat{P}_3,\frac{n_2}{n_2+n_3-1}\widehat{P}_2 + \frac{n_3-1}{n_2+n_3-1}\widehat{P}_3\right). 
\end{align}
Now,
\begin{align*}
     &d\left(\frac{n_1}{n_1+1}\widehat{P}_1 + \frac{1}{n_1+1}\widehat{P}_2,\frac{n_2-1}{n_2+n_3-1}\widehat{P}_2 + \frac{n_3}{n_2+n_3-1}\widehat{P}_3\right) - d\left(\frac{n_1}{n_1+1}\widehat{P}_1 + \frac{1}{n_1+1}\widehat{P}_3,\frac{n_2}{n_2+n_3-1}\widehat{P}_2 + \frac{n_3-1}{n_2+n_3-1}\widehat{P}_3\right)\\
     &= -\frac{n_1}{(n_1+1)^2}d(\widehat{P}_1,\widehat{P}_2) - \frac{(n_2-1)n_3}{(n_2+n_3-1)^2}d(\widehat{P}_2,\widehat{P}_3) + \frac{n_1(n_2-1)}{(n_1+1)(n_2+n_3-1)}d(\widehat{P}_1,\widehat{P}_2) + \frac{n_1n_3}{(n_1+1)(n_2+n_3-1)}d(\widehat{P}_1,\widehat{P}_3)\\&+ \frac{n_3}{(n_1+1)(n_2+n_3-1)}d(\widehat{P}_2,\widehat{P}_3) - \left\{-\frac{n_1}{(n_1+1)^2}d(\widehat{P}_1,\widehat{P}_3) - \frac{n_2(n_3-1)}{(n_2+n_3-1)^2}d(\widehat{P}_2,\widehat{P}_3) + \frac{n_1n_2}{(n_1+1)(n_2+n_3-1)}d(\widehat{P}_1,\widehat{P}_2)\right.\\&\left.+ \frac{n_1(n_3-1)}{(n_1+1)(n_2+n_3-1)}d(\widehat{P}_1,\widehat{P}_3) + \frac{n_2}{(n_1+1)(n_2+n_3-1)}d(\widehat{P}_2,\widehat{P}_3)\right\}\\
     &= \frac{n_1+n_2+n_3}{(n_1+1)(n_2+n_3-1)}\left[\frac{n_1}{n_1+1}\{d(\widehat{P}_1,\widehat{P}_3) - d(\widehat{P}_1,\widehat{P}_2)\} + \frac{n_3-n_2}{n_2+n_3-1}d(\widehat{P}_2,\widehat{P}_3)\right].
\end{align*}
Thus, the RHS of \eqref{eq:43} will be higher than the RHS of \eqref{eq:44} if and only if 
$n_1(n_2+n_3-1)\{d(\widehat{P}_1,\widehat{P}_3) - d(\widehat{P}_1,\widehat{P}_2)\} + (n_3-n_2)(n_1+1)d(\widehat{P}_2,\widehat{P}_3) \geq 0$ , i.e., $(A3)$ holds. Therefore, $d_w^*(C_1^{(n+1)},C_2^{(n+1)})$ will be maximized if an observation from $\widehat{P}_2$ is transferred to $C_1$ in the $(n_1+1)$-th iteration if and only if $(A3)$ holds.\\
Let us assume that after transferring all the observations of $\widehat{P}_1$ to $C_1$, we have transferred $r$ many observations of $\widehat{P}_2$ in $C_1$. In the first case we assume that an observation from $\widehat{P}_2$ is transferred to $C_1$ in the next iteration while in the other case we assume that an observation from $\widehat{P}_3$ is transferred to $C_1$ in the next iteration. Let us investigate these two cases separately.\\
 If an observation of distribution $\widehat{P}_2$ is transferred from $C_2$ to $C_1$ in the ($n_1+r+1$)-th iteration, $C_1$ is a possible realization of size $n_1+r+1$ of $\frac{n_1}{n_1+r+1}\widehat{P}_1 + \frac{r+1}{n_1+r+1}\widehat{P}_2$ and $C_2$ is a possible realization of size ($n_2+n_3-r-1$) of the distribution $\frac{n_2-r-1}{n_2+n_3-r-1}\widehat{P}_2 + \frac{n_3}{n_2+n_3-r-1}\widehat{P}_3$. 
\noindent
\begin{center}
$\widetilde{P}_{C_1^{(n_1+r+1)}} = \frac{n_1}{n_1+r+1}$$\widehat{P}_1 + \frac{r+1}{n_1+r+1}$$\widehat{P}_2$ and $\widetilde{P}_{C_2^{(n_1+r+1)}} = \frac{n_2-r-1}{n_2+n_3-r-1}\widehat{P}_2 + \frac{n_3}{n_2+n_3-r-1}\widehat{P}_3$.
\end{center}
 Therefore,    
\begin{align}
d_w^{*}(C_1^{(n_1+r+1)},C_2^{(n_1+r+1)}) 
= \frac{(n_1+r+1)(n_2+n_3-r-1)}{n}&d\left(\frac{n_1}{n_1+r+1}\widehat{P}_1 + \frac{r+1}{n_1+r+1}\widehat{P}_2,\frac{n_2-r-1}{n_2+n_3-r-1}\widehat{P}_2 \right.\nonumber\\
	&\left.+ \frac{n_3}{n_2+n_3-r-1}\widehat{P}_3\right). 
\end{align}
 \noindent
  On the other hand, if an observation of distribution $\widehat{P}_3$ is transferred from $C_2$ to $C_1$ in the $(n_1+r+1)$-th iteration, $C_1$ is a possible realization of size $n_1+r+1$ of $\frac{n_1}{n_1+r+1}\widehat{P}_1 + \frac{r}{n_1+r+1}\widehat{P}_2 + \frac{1}{n_1+r+1}\widehat{P}_3$ and $C_2$ is a possible realization of size $n_2+n_3-r-1$ of the distribution $\frac{n_2-r}{n_2+n_3-r-1}\widehat{P}_2 + \frac{n_3-1}{n_2+n_3-r-1}\widehat{P}_3$. \\
Thus
\begin{center}
$\widetilde{P}_{C_1^{(n_1+r+1)}} = \frac{n_1}{n_1+r+1}$$\widehat{P}_1 + \frac{r}{n_1+r+1}\widehat{P}_2+\frac{1}{n_1+r+1}\widehat{P}_3$ and $\widetilde{P}_{C_2^{(n_1+r+1)}} = \frac{n_2-r}{n_2+n_3-r-1}\widehat{P}_2 + \frac{n_3-1}{n_2+n_3-r-1}\widehat{P}_3$.
\end{center}
 Therefore,    
\begin{align}
d_w^{*}(C_1^{(n_1+r+1)},C_2^{(n_1+r+1)}) 
= \frac{(n_1+r+1)(n_2+n_3-r-1)}{n}d&\left(\frac{n_1}{n_1+r+1}\widehat{P}_1 + \frac{r}{n_1+r+1}\widehat{P}_2 + \frac{1}{n_1+r+1}\widehat{P}_3,\right.\nonumber\\ &\left.\frac{n_2-r}{n_2+n_3-r-1}\widehat{P}_2 + \frac{n_3-1}{n_2+n_3-r-1}\widehat{P}_3\right). \end{align}
 \noindent
Doing some algebraic calculations it can be shown that
\begin{align*}
    &d\left(\frac{n_1}{n_1+r+1}\widehat{P}_1 + \frac{r+1}{n_1+r+1}\widehat{P}_2,\frac{n_2-r-1}{n_2+n_3-r-1}\widehat{P}_2 + \frac{n_3}{n_2+n_3-r-1}\widehat{P}_3\right)\\&-d\left(\frac{n_1}{n_1+r+1}\widehat{P}_1 + \frac{r}{n_1+r+1}\widehat{P}_2 + \frac{1}{n_1+r+1}\widehat{P}_3,\frac{n_2-r}{n_2+n_3-r-1}\widehat{P}_2 + \frac{n_3-1}{n_2+n_3-r-1}\widehat{P}_3\right)\\
    &= \frac{(n_1+1)[n_1(n_2+n_3-1)\{d(\widehat{P}_1,\widehat{P}_3) - d(\widehat{P}_1,\widehat{P}_2)\} + (n_3-n_2)(n_1+1)d(\widehat{P}_2,\widehat{P}_3)]}{(n_1+r+1)^2(n_2+n_3-r-1)} \\&+ \frac{rn(2n_3-1)}{(n_1+r+1)(n_2+n_3-r-1)(n_2+n_3-1)}d(\widehat{P}_2,\widehat{P}_3).
\end{align*}
If $(A3)$ holds, the above quantity will be positive. Thus $d_w^*$ will be higher if an observation from $\widehat{P}_2$ is transferred in the ($n_1+r+1$)-th iteration. Therefore, observations from $\widehat{P}_2$ will be transferred to $C_1$ from the $(n_1+1)$-th to the $(n_1+n_2)$-th iteration.\\
For $r=n_1,n_1+1,\dots,n_1+n_2$,
\begin{center}
    $\widetilde{P}_{C_1^{(r)}} = \frac{n_1}{r}\widehat{P}_1 + \frac{r-n_1}{r}\widehat{P}_2$ and $\widetilde{P}_{C_2^{(r)}} = \frac{n_1+n_2-r}{n-r}\widehat{P}_2 + \frac{n_3}{n-r}\widehat{P}_3$,
\end{center}
and
\begin{align} \label{eq:48}
    d_w^*(C_1^{(r)},C_2^{(r)}) &= \frac{r(n-r)}{n}d\left(\frac{n_1}{r}\widehat{P}_1 + \frac{r-n_1}{r}\widehat{P}_2,\frac{n_1+n_2-r}{n-r}\widehat{P}_2 + \frac{n_3}{n-r}\widehat{P}_3\right) \nonumber\\
    &=\frac{nn_1-r(n_1+n_3)}{n}\left\{\frac{n_1}{r}d(\widehat{P}_1,\widehat{P}_2) - \frac{n_3}{n-r}d(\widehat{P}_2,\widehat{P}_3)\right\}+\frac{n_1n_3}{n}d(\widehat{P}_1,\widehat{P}_3).
\end{align}
If the RHS of \eqref{eq:48} is treated as a continuous and differentiable function of $r$ and we differentiate with respect to $r$ then second derivative of $d_w^*(C_1^{(r)},C_2^{(r)})$ with respect to $r$ would be
\begin{center}
    $\frac{2nn_1^{2}}{nr^3}d(\widehat{P}_1,\widehat{P}_2)+\left\{\frac{2n_3(n_1+n_3)}{(n-r)^{3}}-\frac{2n_1n_3}{(n-r)^{3}}\right\}d(\widehat{P}_2,\widehat{P}_3)=2\left\{\frac{2nn_1^{2}}{nr^3}d(\widehat{P}_1,\widehat{P}_2)+\frac{n_3^{2}}{(n-r)^{3}}d(\widehat{P}_2,\widehat{P}_3)\right\}>0$.
\end{center}
Hence, $d_w^*$ curve is convex between the $n_1$-th and the $(n_1+n_2)$-th iteration.\\
\textit{(c)} After the $(n_1+n_2)$-th iteration, observations from only $\widehat{P}_3$ will be transferred to $C_1$ in each iteration. Hence, for $r=n_1+n_2,n_1+n_2+1,\dots,n-1$,
\begin{center}
$\widetilde{P}_{C_1^{(r)}}=\frac{n_1}{r}\widehat{P}_1+\frac{n_2}{r}\widehat{P}_2+\frac{r-n_1-n_2}{r}\widehat{P}_3$, \quad $\widetilde{P}_{C_2^{(r)}}=\widehat{P}_3$,
\end{center}
and
\begin{align} \label{eq:49}
    d_w^*(C_1^{(r)},C_2^{(r)}) &= \frac{r(n-r)}{n}d\left(\frac{n_1}{r}\widehat{P}_1+\frac{n_2}{r}\widehat{P}_2+\frac{r-n_1-n_2}{r}\widehat{P}_3,\widehat{P}_3\right) \nonumber\\
    &=\frac{n-r}{rn}\{n_2(n_1+n_2)d(\widehat{P}_2,\widehat{P}_3)+n_1(n_1+n_2)d(\widehat{P}_1,\widehat{P}_3)-n_1n_2d(\widehat{P}_1,\widehat{P}_2)\},
\end{align}
which is a decreasing function of $r$. Hence, $d_w^*$ curve goes on decreasing after the $(n_1+n_2)$-th iteration.
\end{proof}
\noindent
Combining \eqref{eq:50}, \eqref{eq:48} and \eqref{eq:49}, we get equation (9) of the main paper, the explicit form of $d_w^*$ in three class problem. Lemma \ref{lma3} tells that $d_w^*$ curve keeps on increasing till the $n_1$-th iteration, keeps on decreasing after the $(n_1+n_2)$-th iteration and is convex between the $n_1$-th and the $(n_1+n_2)$-th iteration. Hence, $d_w^*(C_1^{(r)},C_2^{(r)})$ curve is maximum either at $r=n_1$ or $r=n_1+n_2$. The decision made by the Algorithm SCC* is always correct under the conditions of Theorem 2 (which follows from Theorem 1). Thus, the Algorithm SCC* detects that the observations are coming from more than one distribution and the Algorithm BS* splits the set of all observations into two clusters of which one contains observations from a single distribution while the other contains observations from mixture of the remaining two distributions. Also, for first cluster, the Algorithm SCC* detects that the observations are coming from a single distribution while for the other cluster it detects that the observations are coming from more than one distribution. In the next step, the Algorithm BS* splits the second cluster into two clusters each of which contains observations from a single distribution only. For each of these two clusters, the Algorithm SCC* detects that the observations are coming from a single distribution. Hence, the Algorithm CURBS-I* is perfect for $K=3$.
\\
\underline{Case 3 ($K > 3$)}: Denote the $K$ empirical distributions by $\widehat{P}_1, \ldots, \widehat{P}_K$, which have been constructed based on samples of sizes $n_1, \ldots ,n_K$, respectively. Denote $n=n_1+n_2+\ldots+n_K$. The following lemma describes the behaviour of $d_w^*$ in $K \:(>3)$ class problem. Let us restate the conditions $(A4)$ and $(A5)$ which will be used to prove the lemma.\\
 $(A4)$ $(n^2-n-nn_1+1)\sum_{i=1}^{K}n_{i}d(\widehat{P}_1,\widehat{P}_i) > (n^2-n-nn_j+1)\sum_{i=1}^{K}n_{i}d(\widehat{P}_j,\widehat{P}_i)$ for all $j=2,3,\dots,K$.\\
 $(A5)$  $n(\sum_{i=l+1}^{K}n_i-1)\sum_{i=1}^{l}n_i\{d(\widehat{P}_i,\widehat{P}_j)-d(\widehat{P}_i,\widehat{P}_{l+1})\}+(\sum_{i=1}^{l}n_i+1)^{2}\sum_{i=l+1}^{K}n_{i}\{(n_j-1)d(\widehat{P}_i,\widehat{P}_j)-(n_{l+1}-1)d(\widehat{P}_i,\widehat{P}_{l+1})\}+(\sum_{i=1}^{l}n_{i}+1)(\sum_{i=l+1}^{K}n_{i}-1)\sum_{i=l+1}^{K}n_{i}\{d(\widehat{P}_i,\widehat{P}_{l+1})-d(\widehat{P}_i,\widehat{P}_j)\} > 0$ for all $j=l+2,l+3,\dots,K$.
\begin{lemma} \label{lma4}
In $K \:(>3)$ class problem, the following statements are true.
    \begin{enumerate} [label=(\alph*)]
        \item An observation from $\widehat{P}_1$ will be transferred to $C_1$ in the first iteration if and only if the condition $(A4)$ holds. Furthermore, if $(A4)$ holds, observations from $\widehat{P}_1$ will be transferred to $C_1$ in each of the first $n_1$ iterations and $d_w^*$ curve keeps on increasing till the $n_1$-th iteration.
        \item For $l=1,2,\dots,K-2$, if all observations from $\widehat{P}_1,\widehat{P}_2,\dots,\widehat{P}_l$ are transferred to $C_1$ in the first $\sum_{i=1}^{l}n_i$ iterations, an observation from $\widehat{P}_{l+1}$ will be transferred to $C_1$ in the $(\sum_{i=1}^{l}n_i+1)$-th iteration if and only if the condition $(A5)$ holds. Furthermore, if $(A5)$ holds, observations from $\widehat{P}_{l+1}$ will be transferred to $C_1$ from the $(\sum_{i=1}^{l}n_i+1)$-th to the $\sum_{i=1}^{l+1}n_i$-th iterations, and the $d_w^*$ curve is convex between the $(\sum_{i=1}^{l}n_i)$-th and the $(\sum_{i=1}^{l+1}n_i)$-th iterations.
        \item After the $(n-n_K)$-th iteration, the $d_w^*$ curve keeps on decreasing.
    \end{enumerate}
\end{lemma}
\begin{proof}
\textit{(a)} If an observation from $\widehat{P}_1$ is transferred to $C_1$ in the first iteration then
        \begin{center}
            $\widetilde{P}_{C_1^{(1)}}=\widehat{P}_1$ and $\widetilde{P}_{C_2^{(1)}}=\frac{n_1-1}{n-1}\widehat{P}_1+\frac{\sum_{i=2}^{K}n_{i}\widehat{P}_i}{n-1}$.
        \end{center}
        Thus, 
        \begin{align} \label{eq:51}
            d_w^*(C_1^{(1)},C_2^{(1)})=\frac{n-1}{n}d\left(\widehat{P}_1,\frac{n_1-1}{n-1}\widehat{P}_1+\frac{\sum_{i=2}^{K}n_{i}\widehat{P}_i}{n-1}\right).
        \end{align}
        On the other hand, if an observation from $\widehat{P}_{j}$ ($j \neq 1$) is transferred to $C_1$ in the first iteration then
        \begin{center}
            $\widetilde{P}_{C_1^{(1)}}=\widehat{P}_j$ and $\widetilde{P}_{C_2^{(1)}}=\frac{n_j-1}{n-1}\widehat{P}_j+\frac{\sum_{i \neq j}n_{i}\widehat{P}_i}{n-1}$.
        \end{center}
        Therefore,
        \begin{align} \label{eq:52}
            d_w^*(C_1^{(1)},C_2^{(1)})=\frac{n-1}{n}d\left(\widehat{P}_j,\frac{n_j-1}{n-1}\widehat{P}_j+\frac{\sum_{i\neq j}n_{i}\widehat{P}_i}{n-1}\right).
        \end{align}
        The RHS of \eqref{eq:51} will be greater than the RHS of \eqref{eq:52} if and only if 
        \begin{align} \label{eq:53}
            d\left(\widehat{P}_1,\frac{n_1-1}{n-1}\widehat{P}_1+\frac{\sum_{i=2}^{K}n_{i}\widehat{P}_i}{n-1}\right)>d\left(\widehat{P}_j,\frac{n_j-1}{n-1}\widehat{P}_j+\frac{\sum_{i\neq j}n_{i}\widehat{P}_i}{n-1}\right),
        \end{align}
        for all $j=2,3,\dots,K$. By simple algebraic calculation, it can be shown that \eqref{eq:53} is equivalent to the condition $(A4)$. Hence, an observation from $\widehat{P}_1$ will be transferred to $C_1$ in the first iteration if and only if $(A4)$ holds.\\
        Using the similar argument of part $(a)$ of Lemma \ref{lma2} with $\widehat{P}_2$ replaced by the mixture of $\widehat{P}_2,\widehat{P}_3,\dots,\widehat{P}_K$, it can be easily proved that observations from $\widehat{P}_1$ will be transferred to $C_1$ in in each of the first $n_1$ iterations.
        Therefore, for $r=1,2,\dots,n_1$,
        \begin{center}
            $\widetilde{P}_{C_1^{(r)}}=\widehat{P}_1$ and $\widetilde{P}_{C_2^{(r)}}=\frac{n_1-r}{n-r}\widehat{P}_1+\frac{\sum_{i=2}^{K}n_{i}\widehat{P}_i}{n-r}$.
        \end{center}
        Thus for $r=1,2,\dots,n_1$
        \begin{align} \label{eq:54}
            d_w^*(C_1^{(r)},C_2^{(r)})&=\frac{r(n-r)}{n}d\left(\widehat{P}_1,\frac{n_1-r}{n-r}\widehat{P}_1+\frac{\sum_{i=2}^{K}n_{i}\widehat{P}_i}{n-r}\right) \nonumber\\
            &=\frac{r(n-n_1)^2}{n(n-r)}\left\{\frac{1}{n-n_1}\sum_{i=2}^{K}n_{i}d(\widehat{P}_1,\widehat{P}_i)-\frac{1}{2(n-n_1)^2}\sum_{i,j=2}^{K}n_{i}n_{j}d(\widehat{P}_i,\widehat{P}_j)\right\},
        \end{align}
        which is an increasing function of $r$. Hence, $d_w^*$ curve keeps on increasing till the $n_1$-th iteration.\\
        \textit{(b)} After all observations from $\widehat{P}_1,\widehat{P}_2,\dots,\widehat{P}_l$ are transferred to $C_1$ in first $\sum_{i=1}^{l}n_i$ iterations, if an observation from $\widehat{P}_{l+1}$ is transferred to $C_1$ in the $(\sum_{i=1}^{l}n_i+1)$-th iteration then
        \begin{center}
            $\widetilde{P}_{C_1^{\left(\sum_{i=1}^{l}n_{i}+1\right)}}= \frac{\sum_{i=1}^{l}n_{i}\widehat{P}_{i}}{\sum_{i=1}^{l}n_{i}+1}+\frac{1}{\sum_{i=1}^{l}n_{i}+1}\widehat{P}_{l+1}$, and $\widetilde{P}_{C_2^{\left(\sum_{i=1}^{l}n_{i}+1\right)}}=\frac{n_{l+1}-1}{\sum_{i=l+1}^{K}n_{i}-1}\widehat{P}_{l+1}+\frac{\sum_{i=l+2}^{K}n_{i}\widehat{P}_{i}}{\sum_{i=l+1}^{K}n_{i}-1}$.
        \end{center}
        Therefore,
        \begin{align} \label{eq:56}
            &d_w^*\left(C_1^{\left(\sum_{i=1}^{l}n_{i}+1\right)},C_2^{\left(\sum_{i=1}^{l}n_{i}+1\right)}\right) \nonumber\\&= \frac{\left(\sum_{i=1}^{l}n_{i}+1\right)\left(\sum_{i=l+1}^{K}n_{i}-1\right)}{n}d\left(\frac{\sum_{i=1}^{l}n_{i}\widehat{P}_{i}}{\sum_{i=1}^{l}n_{i}+1}+\frac{1}{\sum_{i=1}^{l}n_{i}+1}\widehat{P}_{l+1},\frac{n_{l+1}-1}{\sum_{i=l+1}^{K}n_{i}-1}\widehat{P}_{l+1}+\frac{\sum_{i=l+2}^{K}n_{i}\widehat{P}_{i}}{\sum_{i=l+1}^{K}n_{i}-1}\right).
        \end{align}
        On the other hand, if an observation from $\widehat{P}_j$ (where $j \in \{l+2,l+3,\dots,K\}$) is transferred to $C_1$ in the $(\sum_{i=1}^{l}n_i+1)$-th iteration then 
                \begin{center}
            $\widetilde{P}_{C_1^{\left(\sum_{i=1}^{l}n_{i}+1\right)}}= \frac{\sum_{i=1}^{l}n_{i}\widehat{P}_{i}}{\sum_{i=1}^{l}n_{i}+1}+\frac{1}{\sum_{i=1}^{l}n_{i}+1}\widehat{P}_{j}$ and $\widetilde{P}_{C_2^{\left(\sum_{i=1}^{l}n_{i}+1\right)}}=\frac{n_{j}-1}{\sum_{i=l+1}^{K}n_{i}-1}\widehat{P}_{j}+\frac{\sum_{i(\neq j)=l+1}^{K}n_{i}\widehat{P}_{i}}{\sum_{i=l+1}^{K}n_{i}-1}$.
        \end{center}
        Therefore,
        \begin{align} \label{eq:57}
            &d_w^*\left(C_1^{\left(\sum_{i=1}^{l}n_{i}+1\right)},C_2^{\left(\sum_{i=1}^{l}n_{i}+1\right)}\right) \nonumber\\&= \frac{\left(\sum_{i=1}^{l}n_{i}+1\right)\left(\sum_{i=l+1}^{K}n_{i}-1\right)}{n}d\left(\frac{\sum_{i=1}^{l}n_{i}\widehat{P}_{i}}{\sum_{i=1}^{l}n_{i}+1}+\frac{1}{\sum_{i=1}^{l}n_{i}+1}\widehat{P}_{j},\frac{n_{j}-1}{\sum_{i=l+1}^{K}n_{i}-1}\widehat{P}_{j}+\frac{\sum_{i(\neq j)=l+1}^{K}n_{i}\widehat{P}_{i}}{\sum_{i=l+1}^{K}n_{i}-1}\right).
        \end{align}
        The RHS of \eqref{eq:56} is greater than the RHS of \eqref{eq:57} if and only if
        \begin{align} \label{eq:58}
            &d\left(\frac{\sum_{i=1}^{l}n_{i}\widehat{P}_{i}}{\sum_{i=1}^{l}n_{i}+1}+\frac{1}{\sum_{i=1}^{l}n_{i}+1}\widehat{P}_{l+1},\frac{n_{l+1}-1}{\sum_{i=l+1}^{K}n_{i}-1}\widehat{P}_{l+1}+\frac{\sum_{i=l+2}^{K}n_{i}\widehat{P}_{i}}{\sum_{i=l+1}^{K}n_{i}-1}\right)\nonumber\\&>d\left(\frac{\sum_{i=1}^{l}n_{i}\widehat{P}_{i}}{\sum_{i=1}^{l}n_{i}+1}+\frac{1}{\sum_{i=1}^{l}n_{i}+1}\widehat{P}_{j},\frac{n_{j}-1}{\sum_{i=l+1}^{K}n_{i}-1}\widehat{P}_{j}+\frac{\sum_{i(\neq j)=l+1}^{K}n_{i}\widehat{P}_{i}}{\sum_{i=l+1}^{K}n_{i}-1}\right).
        \end{align}
        Doing some simple algebraic calculation, it can be shown that \eqref{eq:58} is equivalent to the condition $(A5)$. Hence, an observation from $\widehat{P}_{l+1}$ will be transferred to $C_1$ in the $(\sum_{i=1}^{l}n_i+1)$-th iteration if and only if $(A5)$ holds.\\
        Using similar argument of part $(b)$ of \ref{lma3} with $\widehat{P}_1$ replaced by the mixture of $\widehat{P}_1,\widehat{P}_2,\dots,\widehat{P}_l$, $\widehat{P}_2$ replaced by $\widehat{P}_{l+1}$ and $\widehat{P}_3$ replaced by the mixture of $\widehat{P}_{l+2},\widehat{P}_{l+3},\dots,\widehat{P}_K$, it can be easily shown that observations from $\widehat{P}_{l+1}$ will be transferred to $C_1$ from the $(\sum_{i=1}^{l}n_i+1)$-th to the $\sum_{i=1}^{l+1}n_i$-th iteration if $(A5)$ holds and the $d_w^*$ curve is convex between the $\sum_{i=1}^{l}n_i$-th and the $\sum_{i=1}^{l+1}n_i$-th iteration.\\
        Hence, for $r=\sum_{i=1}^{l}n_i,\sum_{i=1}^{l}n_i+1,\dots,\sum_{i=1}^{l+1}n_i$,
        \begin{center}
            $\widetilde{P}_{C_1^{(r)}}=\frac{\sum_{i=1}^{l}n_i\widehat{P}_i}{r}+\frac{r-\sum_{i=1}^{l}n_i}{r}\widehat{P}_{l+1}$ and $\widetilde{P}_{C_2^{(r)}}=\frac{\sum_{i=1}^{l+1}n_i-r}{n-r}\widehat{P}_{l+1}+\frac{\sum_{i=l+2}^{K}n_i\widehat{P}_i}{n-r}$.
        \end{center}
        Thus for $r=\sum_{i=1}^{l}n_i,\sum_{i=1}^{l}n_i+1,\dots,\sum_{i=1}^{l+1}n_i$,
        \begin{align} \label{eq:59}
            d_w^*(C_1^{(r)},C_2^{(r)})&=\frac{r(n-r)}{n}d\left(\frac{\sum_{i=1}^{l}n_i\widehat{P}_i}{r}+\frac{r-\sum_{i=1}^{l}n_i}{r}\widehat{P}_{l+1},\frac{\sum_{i=1}^{l+1}n_i-r}{n-r}\widehat{P}_{l+1}+\frac{\sum_{i=l+2}^{K}n_i\widehat{P}_i}{n-r}\right) \nonumber\\
            &= \frac{n\left(\sum_{i=1}^{l}n_i \right)-r(n-n_{l+1})}{n}\left\{\frac{\sum_{i=1}^{l}n_i}{r}d\left(\frac{\sum_{i=1}^{l}n_{i}\widehat{P}_i}{\sum_{i=1}^{l}n_i},\widehat{P}_{l+1} \right) - \frac{\sum_{i=l+2}^{K}n_i}{n-r}d\left(\widehat{P}_{l+1},\frac{\sum_{i=l+2}^{K}n_{i}\widehat{P}_i}{\sum_{i=l+2}^{K}n_i}\right)\right\}\nonumber\\&+\frac{\left(\sum_{i=1}^{l}n_i \right)\left(\sum_{i=l+2}^{K}n_i \right)}{n}d\left(\frac{\sum_{i=1}^{l}n_{i}\widehat{P}_i}{\sum_{i=1}^{l}n_i},\frac{\sum_{i=l+2}^{K}n_{i}\widehat{P}_i}{\sum_{i=l+2}^{K}n_i}\right).
        \end{align}
        \noindent\\\noindent
        \textit{(c)} After the $(n-n_K)$-th iteration, observation from $\widehat{P}_K$ will start coming to $C_1$ in each iteration. Hence, for $r=n-n_K+1,n-n_K+2,\dots,n$
        \begin{center}
    $\widetilde{P}_{C_1^{(r)}}=\frac{\sum_{i=1}^{K-1}n_{i}\widehat{P}_i}{r}+\frac{r-\sum_{i=1}^{K-1}n_{i}}{r}\widehat{P}_{K}$ and $\widetilde{P}_{C_2^{(r)}}=\widehat{P}_K$.
        \end{center}
        Thus for $r=n-n_K+1,n-n_K+2,\dots,n$
        \begin{align} \label{eq:55}
            d_w^*(C_1^{(r)},C_2^{(r)})&=\frac{r(n-r)}{n}d\left(\frac{\sum_{i=1}^{K-1}n_{i}\widehat{P}_i}{r}+\frac{r-\sum_{i=1}^{K-1}n_{i}}{r}\widehat{P}_{K},\widehat{P}_{K}\right) \nonumber\\
            &=\frac{n-r}{rn}\left\{(n-n_K)\sum_{i=1}^{K-1}n_{i}d(\widehat{P}_{i},\widehat{P}_{K})-\frac{1}{2}\sum_{i,j=1}^{K-1}n_{i}n_{j}d(\widehat{P}_i,\widehat{P}_j)\right\},
        \end{align}
        which is a decreasing function of $r$. Hence, $d_w^*$ curve keeps on decreasing after the $(n-n_K)$-th iteration.
\end{proof}
\noindent
Combining \eqref{eq:54}, \eqref{eq:59} and \eqref{eq:55}, we get equation (13) of the main paper, the explicit form of $d_w^*$ in $K \:(>3)$ class problem.\\  
By Lemma \ref{lma4}, the $d_w^*$ curve keeps on increasing till the $n_1$-th iteration and keeps on decreasing after the $(n-n_K)$-th iteration. Also, for $l=1,2,\dots,K-2$, the $d_w^*$ curve is convex between the $\sum_{i=1}^{l}n_i$-th and the $\sum_{i=1}^{l+1}n_i$-th iteration. Therefore, $d_w^*(C_1^{(r)},C_2^{(r)})$ curve attains its maximum at $r=\sum_{i=1}^{l}n_i$ for some $l=1,2,\dots,K-1$. Since the decision made by the Algorithm SCC* is always correct under the assumption of Theorem 2 (which follows from Theorem 1), whenever the Algorithm SCC* will be applied on a set containing observations of more than one distribution, it will detect that the observations are coming from more than one distribution. The Algorithm BS* splits each of these sets into two subsets such that these two subsets do not contain observations from any common distribution. On the other hand, when the Algorithm SCC* will be applied on a set containing observations from one distribution only, it will detect that the observations are coming from single distribution. Hence, the Algorithm CURBS-I* is perfect for $K>3$.
\subsection*{Proof of Theorem 3}
To prove that the Algorithm CURBS-II* is POP in oracle scenario, we need to show that 
\begin{enumerate}
\item If the specified cluster number and true cluster number are equal, the partition obtained by the algorithm and the original partition will be same.
    \item If the specified cluster number is smaller than the true cluster number, no two different clusters contain observations from any common distribution.
    \item If the specified cluster number is higher than the true cluster number, there will be no cluster containing observations coming from more than one distribution.
\end{enumerate}
In this proof, the true cluster number and the specified cluster number will be denoted by $K$ and $J$ respectively. Also, we refer to the process of splitting $i$ available clusters and merging $(i-1)$ pairs of clusters as the $i$-th stage of the Algorithm CURBS-II*. In other words, the $i$-th stage is comprised of Steps 3, 4 and 5 of the Algorithm CURBS-II*, $i=1,2,\dots,J-1$. \\
For $K \geq 2$, denote the $K$ empirical distributions by $\widehat{P}_1, \ldots, \widehat{P}_K$, which have been constructed based on samples of sizes $n_1, \ldots ,n_K$, respectively. The following lemma describes the properties of the partition obtained by the Algorithm CURBS-II* for different specified number of clusters in $K \:(\geq 2)$ class problem.
\begin{lemma} \label{lma5} In $K \:(\geq 2)$ class problem, the following statements are true.
    \begin{enumerate} [label=(\alph*)]
        \item If $J=j(<K)$, the Algorithm CURBS-II* produces $j$ clusters such that no two different clusters contain observations from any common distribution.
        \item If $J=K$, the partition obtained by the Algorithm CURBS-II* and the original partition are same.
        \item  If $J>K$, from the $K$-th stage onwards, no set containing observations from $\widehat{P}_{i}$ will be merged with a set containing observations from $\widehat{P}_{j}$ in any stage where $i \neq j$.
    \end{enumerate}
\end{lemma}
\begin{proof}
\textit{(a)} For $K=2$, the Property 2 (that is, when $J=1$) is satisfied trivially. For $K>2$, we will prove this part by mathematical induction. First, we will show that the this property is satisfied for the first stage. Then assuming that it holds for the $i$-th stage, we will show that this holds for the $(i+1)$-th stage also.\\In $K \:(>2)$ class problem, $d_w^*$ attains maximum value at the $\sum_{i=1}^{l}n_i$-th iteration for some $l=1,2,\dots,K-1$ (see the discussion for three class problem and $K \:(>3)$ class problem in Section 4.1 of the main paper). Hence, in the first stage of the Algorithm CURBS-II*, the Algorithm BS* splits the set of observations into two sets such that these two sets do not contain observations from any common distribution. Therefore, Property 2 is satisfied for the first stage. Now, suppose that this property is satisfied after completion of the $i$-th stage, for some $i \in \{1,2,\dots,K-3\}$. That is, among $i+1$ clusters obtained as output after the $i$-th stage, no two different clusters contain observations from any common distribution. Suppose, among these clusters, $U$ be a cluster containing observations from $\widehat{P}_{s}$ only and $W$ be another cluster containing no observation from $\widehat{P}_{s}$, $s \in \{1,2,\dots,K\}$. In the next stage, $U$ is split into $U_1$ and $U_2$ and $W$ is split into $W_1$ and $W_2$. Since $U_1$ and $U_2$ both contain observations from $\widehat{P}_s$, $d_w^*(U_1,U_2)=0$. On the other hand, as $U$ and $W$ do not contain observations from any common distribution, $U_{x}$ and $V_{y}$ do not contain observations from any common distribution for all $x,y=1,2$. Therefore, $d_w^*(U_x,W_y)>0$ for all $x,y=1,2$. In the merging step, $U_1$ and $U_2$ will definitely get merged. Also, in the beginning of the $(i+1)$-th stage, there are at most $i$ many clusters containing observations from 
a single distribution only. Therefore, after splitting each of the clusters, there will be at most $i$ many pairs of clusters such that two clusters of each pair contain observations from same distribution. Hence, if a cluster containing observations from a single distribution is split into two subclusters in any stage then these subclusters will definitely be merged to form back the cluster again. Therefore, after completion of the $(i+1)$-th stage, $i+2$ clusters are produced such that observations from any common distribution are not present in two different clusters simultaneously. On the other hand, if $i+1$ clusters are produced after the $i$-th stage such that no cluster contains observations from single distribution only then in the $(i+1)$-th stage, each of $i$ clusters will be split into two subparts and these subparts do not contain observations from any common distribution. Also, as the statement $(a)$ is true for $i$-th stage, among $i+1$ clusters obtained after $i$-th stage, no two different clusters contain observations from any common distribution. Therefore, among $2i+2$ clusters obtained after splitting in $(i+1)$-th stage, no two different clusters contain observations from any common distribution. Thus, after merging $i$ pairs of clusters, $i+2$ clusters will be produced such that no two different clusters contain observations from same distribution. Hence, Property 2 is satisfied for the $(i+1)$-th stage also. Through mathematical induction, Property 2 is satisfied from the first stage to the $(K-2)$-th stage.\\
        \textit{(b)} For $K=2$, in first stage the Algorithm BS* splits the set of observations into two clusters of which one contains observations from $\widehat{P}_1$ and the other contains observations from $\widehat{P}_2$ (see the discussion for two class problem in Section 4.1 of the main paper). Hence, the original partition and the partition obtained by the Algorithm CURBS-II* are same.
        \par
        For $K>2$, after completion of the $(K-2)$-th stage, $K-1$ clusters are produced such that no two different clusters contain observations from any common distribution (which follows from part $(a)$). Thus, among these $K-1$ clusters, $K-2$ clusters contain observations from a single distribution and the other cluster contains observations from two different distributions. Without loss of generality, assume that $A_{i}^*$ contains observations from $\widehat{P}_i$, $i=1,2,\dots,K-2$ and $A_{K-1}^*$ contains observations from $\widehat{P}_{K-1}$ and $\widehat{P}_{K}$. In the next iteration, each of $A_{i}^*$ ($i=1,2,\dots,K-1$) is split into $A_{i,1}^*$ and $A_{i,2}^*$. Then for $i=1,2,\dots,K-2$,
        \begin{align*}
            &d_w^*(A_{i,1}^*,A_{i,2}^*)=0, \quad d_w^*(A_{K-1,1}^*,A_{K-1,2}^*)=\frac{|A_{K-1,1}^*||A_{K-1,2}^*|}{|A_{K-1,1}^*|+|A_{K-1,2}^*|}d(\widehat{P}_{K-1},\widehat{P}_{K})>0,\\
            &d_w^*(A_{i,j}^*,A_{K-1,1}^*)=\frac{|A_{i,j}^*||A_{K-1,1}^*|}{|A_{i,j}^*|+|A_{K-1,1}^*|}d(\widehat{P}_{i},\widehat{P}_{K-1})>0, \quad d_w^*(A_{i,j}^*,A_{K-1,1}^*)=\frac{|A_{i,j}^*||A_{K-1,2}^*|}{|A_{i,j}^*|+|A_{K-1,2}^*|}d(\widehat{P}_{i},\widehat{P}_{K})>0,  
        \end{align*}
        for $j=1,2$. Also, for $i_1, i_2 \in \{1,2,\dots,K-2\}$ such that $i_1 \neq i_2$, and $j_1,j_2 \in \{1,2\}$,
        \begin{align*}
            d_w^*(A_{i_1,j_1}^*,A_{i_2,j_2}^*)=\frac{|A_{i_1,j_1}^*||A_{i_1,j_2}^*|}{|A_{i_1,j_1}^*|+|A_{i_2,j_2}^*|}d(\widehat{P}_{i_1},\widehat{P}_{i_2})>0.
        \end{align*}
        Since $K-2$ pair of clusters need to be merged in the $(K-1)$-th stage, $A_{i,1}^*$ and $A_{i,2}^*$ will be merged in this stage, $i=1,2,\dots,K-2$. Hence, after completion of the $(K-1)$-th stage, the output will be $K$ clusters each of which contains observations from single distribution only and no two different clusters contain observations from any common distribution. Therefore, the partition obtained by the Algorithm CURBS-II* and the original partition are same. \\
        \textit{(c)} For $K \geq 2$, after completion of the $(K-1)$-th stage, $K$ many clusters are produced each of which contains observations from single distribution only and no two different clusters contain observation from any common distribution (which follows from part \textit{(b)}). In the $K$-th stage, each of these clusters will be split into two subparts. Also, it is to be noted that the $d_w^*$ distance between two sets containing observations from same distribution (which is actually zero) is less than $d_w^*$ distance between two sets containing observations from different distributions. Among $2K$ clusters obtained after splitting, there are $K$ pairs such that both sets of each pair contain observations from same distribution. Therefore, in the merging step, among those $K$ pairs, any $K-1$ pairs of sets will be merged. Hence, Property 3 is satisfied for the $K$-th stage. Now, suppose that this property is satisfied for the $j$-th stage (where $j \geq K$). Also, assume that after completion of the $j$-th stage, there are two clusters $S_{x}$ and $S_{y}$ among the $j+1$ many clusters obtained as output such that $S_{x}$ contains observations from $\widehat{P}_{l_1}$ only and $S_{y}$ contains observations from $\widehat{P}_{l_2}$ only, where $l_1 \neq l_2$. In the next stage, $S_{x}$ is split into $S_{x,1}$ and $S_{x,2}$ and $S_{y}$ is split into $S_{y,1}$ and $S_{y,2}$. Then $d_w^*(S_{x,1},S_{x,2})=d_w^*(S_{y,1},S_{y,2})=0$ and
        \begin{center}
            $d_w^*(S_{x,1},S_{y,1})=\frac{|S_{x,1}||S_{y,1}|}{|S_{x,1}|+|S_{y,1}|}d(\widehat{P}_{l_1},\widehat{P}_{l_2})>0$, $d_w^*(S_{x,1},S_{y,2})=\frac{|S_{x,1}||S_{y,2}|}{|S_{x,1}|+|S_{y,2}|}d(\widehat{P}_{l_1},\widehat{P}_{l_2})>0$.
        \end{center}
        Similarly $d_w^*(S_{x,2},S_{y,1})>0$ and $d_w^*(S_{x,2},S_{y,2})>0$. Thus $d_w^*$ distance between two sets containing observations from same distribution (which is actually zero) is less than $d_w^*$ distance between two sets containing observations from different distributions.  There are at least $j+1$ many pairs of sets containing observations coming from same distribution. Since $j$ many pair of clusters are to be merged in the $(j+1)$-th stage, no subset containing observations from a single distribution will merge with one subset containing observations from another distribution in the $(j+1)$-th stage. Hence, Property 3 is true for the $(j+1)$-th stage also. Using the principle of mathematical induction, this statement is true for all stages starting from the $K$-th stage.
\end{proof}
By part \textit{(a)} and \textit{(b)} of Lemma \ref{lma5}, Property 2 and Property 1 are satisfied for $J<K$ and $J=K$ respectively. By part \textit{(c)} of Lemma \ref{lma5}, if $J > K$, after completion of the $(J-1)$-th stage, the output will be $J$ clusters containing observations from single distribution only. Hence, Property 3 is satisfied for all $J>K$. Therefore, the Algorithm CURBS-II* is POP in the oracle scenario for $K \:(\geq 2)$ class problem.

\end{document}